\documentclass[a4paper,twoside,titlepage,12pt]{scrbook} 

\usepackage{geometry}
\usepackage{scrpage2}

\usepackage{graphicx}
\usepackage[x11names, rgb]{xcolor}

\usepackage{tikz}
\usetikzlibrary{decorations,arrows,petri,topaths,shapes,positioning,fit}

\usepackage{amsmath}
\usepackage{amsfonts}
\usepackage{amssymb}
\usepackage{amsthm}

\usepackage{synttree}
\usepackage{algpseudocode}
\usepackage{algorithm}
\usepackage{makeidx}
\usepackage{listings}

\usepackage[english]{babel}

\DeclareMathOperator{\poly}{poly}
\DeclareMathOperator{\complP}{P}
\DeclareMathOperator{\NP}{NP}
\DeclareMathOperator{\coNP}{coNP}
\DeclareMathOperator{\XP}{XP}
\DeclareMathOperator{\FPT}{FPT}
\DeclareMathOperator{\PH}{PH}

\DeclareMathOperator{\Component}{Component}
\DeclareMathOperator{\Remain}{Remain}

\DeclareMathOperator{\tw}{tw}
\DeclareMathOperator{\VC}{VC}
\DeclareMathOperator{\neiin}{in}
\DeclareMathOperator{\neiout}{out}

\DeclareMathOperator{\RemoveLowDegree}{RLD}
\DeclareMathOperator{\DisLined}{DisLined}

\DeclareMathOperator{\adj}{adj}

\DeclareMathOperator{\iNotKNeighbors}{iNotKNeighbors}
\DeclareMathOperator{\iNotKNeighborsEx}{iNotKNeighborsExactly}

\DeclareMathOperator{\cost}{cost}

\usepackage{stmaryrd}

\usepackage{url}

\usepackage[utf8]{inputenc} 

\usepackage{eso-pic}
\usepackage{scalefnt}
\usepackage{fix-cm}

\newtheorem{observation}{Observation}
\newtheorem{definition}{Definition}
\newtheorem{lemma}{Lemma}
\newtheorem{theorem}{Theorem}
\newtheorem{proposition}{Proposition}
\newtheorem{corollary}{Corollary}

\newtheoremstyle{example}{\topsep}{\topsep}%
{} 
{} 
{\bfseries} 
{} 
{\newline} 
{} 
\theoremstyle{example}

\newtheorem{note}{Note}

\newtheoremstyle{reductionrule}{\topsep}{\topsep}%
     {}
     {}
     {\bfseries}
     {}
     {\newline}
     {\thmname{#1} ``\thmnote{#3}''}
\theoremstyle{reductionrule}
\newtheorem{reductionrule}{Reduction Rule}

\begin{document}

\pagestyle{empty}

\thispagestyle{empty}
\begin{center}
\textcolor{red}{{\Large \bfseries
 Because of the outdated latex system of arxiv, this is a less-featured version of my diploma theses.\\
 ~\\Please see \url{http://robert.bredereck.info} for the full-featured version.
  \\}}
\newpage
\thispagestyle{empty}

{\Large \bfseries Friedrich-Schiller-Universit\"at Jena\\}
\vspace{0,5cm}
\includegraphics[scale=0.2]{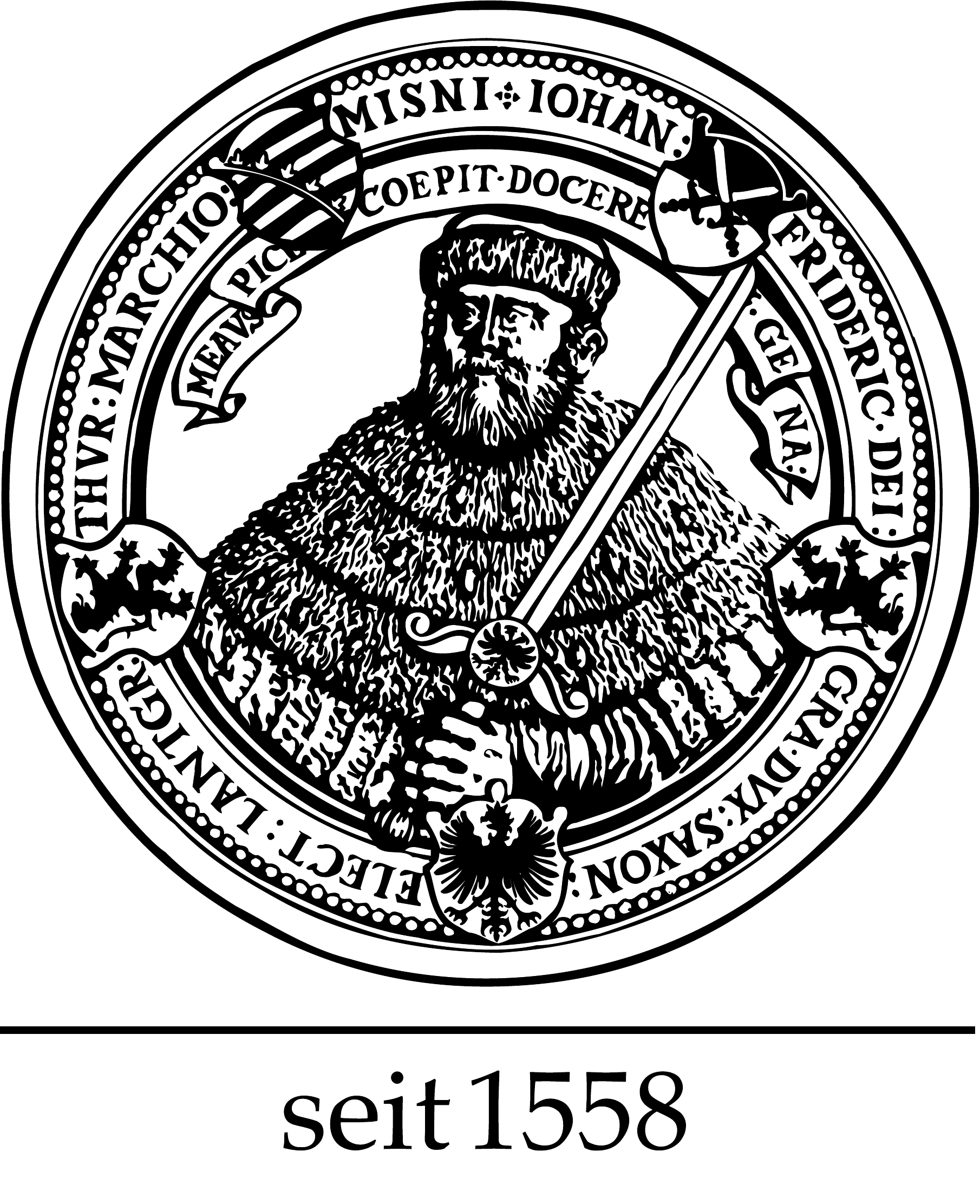}\\
\vspace{1,5cm}
\sffamily
{\LARGE \bfseries \sffamily  D I P L O M A R B E I T}\\
\vspace{1cm}
{\Huge \bfseries Graph and Election Problems \\ \vspace*{0.5ex} Parameterized by \\ \vspace*{1.0ex} Feedback Set Numbers}\\
\vspace{1,0cm}
zur Erlangung des akademischen Grades\\
Diplom-Informatiker\\
\vspace{1,0cm}
ausgef\"uhrt am Lehrstuhl f\"ur Theoretische Informatik I / Komplexit\"atstheorie\\
an der Falkult\"at f\"ur Mathematik und Informatik\\
der Friedrich-Schiller-Universit\"at Jena\\
\vspace{1cm}
unter Anleitung von Dipl.-Bioinf.\ Nadja Betzler, Dipl.-Bioinf.\ Johannes Uhlmann\\
und Prof.\ Dr.\ Rolf Niedermeier\\
\vspace{,5cm}
durch\\
\vspace{,5cm}
Robert Bredereck\\
geb. am 31.01.1985 in Wolgast\\
\vspace{,5cm}
Jena, 25. Mai 2010 \\
\end{center}
\newpage



{\Large \bf Abstract}
\vspace{0.5cm}
\\
\small
This work investigates the parameterized complexity of three related graph modification problems.
Given a directed graph, a distinguished vertex, and a positive integer $k$, \textsc{Minimum Indegree Deletion} asks for a vertex subset of size at most $k$ whose removal makes the distinguished vertex the only vertex with minimum indegree.
\textsc{Minimum Degree Deletion} is analogously defined, but deals with undirected graphs.
\textsc{Bounded Degree Deletion} is also defined on undirected graphs, but has a positive integer $d$ instead of a distinguished vertex as part of the input.
It asks for a vertex subset of size at most $k$ whose removal results in a graph in which every vertex has degree at most $d$.
The first two problems have applications in computational social choice whereas the third problem is used in computational biology.
We investigate the parameterized complexity with respect to the parameters ``treewidth'', ``size of a feedback vertex set'' and ``size of a feedback edge set'' respectively ``size of a feedback arc set''.
Each of these parameters measures the ``degree of acyclicity'' in different ways.
For \textsc{Minimum Indegree Deletion} we show that it is W[2]-hard with respect to both parameters that are defined on acyclic graphs.
We describe a branch-and-bound algorithm whose running time is $O(s \cdot (k+1)^{s} \cdot n^2)$, where $n$ is the number of vertices, $k$ is the ``number of vertices to delete'', and $s$ is the ``size of a feedback set''.
For \textsc{Minimum Degree Deletion} we show W[1]-hardness with respect to the parameter ``number of vertices to delete''.
With respect to each of the parameters that measures the ``degree of acyclicity'' we show fixed-parameter tractability.
We describe a simple search tree algorithm with running time $O(2^{s} \cdot n^3)$ where $n$  is the number of vertices and $s$ is the ``size of a feedback edge set'' and
two concrete fixed-parameter algorithms with respect to the parameter ``size of a feedback vertex set that does not contain the distinguished vertex''.
For \textsc{Bounded Degree Deletion} we present a search-tree algorithm with running time $O(3^{s} \cdot n^2)$ where $n$ is the number of vertices and $s$ is the ``size of a feedback edge set''.

\newpage
{\Large \bf Zusammenfassung}
\vspace{0.5cm}
\\
\small
Diese Arbeit untersucht die parametrisierte Komplexität von drei verwandten Graphmodifikationsproblemen.
Jedes dieser Probleme sucht eine Knotenteilmenge be\-schränk\-ter Größe, deren Löschung zu einem Graph mit einer problemspezifischen Eigenschaft führt.
Das erste betrachtete Problem \textsc{Minimum Indegree Deletion} hat als Eingabe einen gerichteten Graphen, einen ausgewiesenen Knoten und eine natürliche Zahl $k$.  
Die Frage ist, ob der ausgewiesene Knoten durch Löschung von maximal $k$ Knoten der einzige Knoten mit minimalem Eingangsgrad werden kann.
Das Problem \textsc{Minimum Indegree Deletion} wurde im Kontext der Wahlforschung eingeführt.
Genauer gesagt analysiert man für Wahlsysteme ein \glqq Control\grqq -Szenario, in welchem man fragt, ob durch Löschung einer durch die Größe beschränkten Kandidatenmenge das Ergebnis so zu beeinflussen ist, dass ein ausgewählter Kandidat gewinnt.
Die Ergebnisse bezüglich \textsc{Minimum Indegree Deletion} lassen sich eins zu eins auf dieses Szenario übertragen.
Das zweite betrachtete Problem \textsc{Minimum Degree Deletion} ist Analog zu \textsc{Minimum Indegree Deletion} auf ungerichteten Graphen definiert.
Hier fordert man, dass nach Löschung der Knotenmenge ein ausgewählter Knoten als einziger minimalen Grad hat.
Auch dieses Problem lässt sich durch ein natürliches, auf sozialen Netzwerken basierendes Wahlsystem motivieren.
Das dritte betrachtete Problem \textsc{Bounded Degree Deletion} ist ebenfalls auf ungerichteten Graphen definiert.
Es hat jedoch im Gegensatz zu den beiden vorherigen Problemen keinen ausgewiesenen Knoten als Teil der Eingabe sondern fordert, dass nach Löschung von maximal $k$ Knoten alle Knoten maximal Grad $d$ besitzen für ein bestimmtes zur Eingabe gehöriges $d$.
Das Problem \textsc{Bounded Degree Deletion} findet Anwendung bei der Analyse von biologischen Netzwerken.

Neben Ähnlichkeiten in der Definition haben die drei betrachteten Probleme gemeinsam, dass sie auf kreisfreien Graphen in Polynomzeit lösbar sind, während sie jedoch im Allgemeinen NP-schwer sind.
Diese Eigenschaft war Ausgangspunkt für die Untersuchung der parametrisierten Komplexität bezüglich drei verschiedener Parameter, welche jeweils auf unterschiedliche Weise die Distanz des Eingabegraphen zu einem kreisfreien Graphen messen.
Der erste dieser Parameter ist die für die Untersuchung ungerichteter Graphen weit verbreitete \glqq Baumweite\grqq.
Die beiden anderen Parameter messen jeweils die Anzahl der Knoten bzw.\ Kanten deren Löschung zu einem kreisfreien Graphen führt.
Genauer gesagt sind dies die Parameter \glqq Größe einer kreiskritischen Knotenmenge\grqq~und \glqq Größe einer kreiskritischen Kantenmenge\grqq.

Für \textsc{Minimum Indegree Deletion} zeigen wir, dass es bezüglich beider auf ge\-rich\-te\-ten Graphen definierten Parameter W[2]-schwer und damit offenbar nicht festparameterhandhabbar ist.
Betrachtet man hingegen den Parameter \glqq Anzahl zu löschender Knoten\grqq~in Kombination mit einem der beiden anderen, lässt sich Festparameterhandhabbarkeit nachweisen.
Wir beschreiben einen einfachen Algorithmus dessen Laufzeit $O(s \cdot (k+1)^{s} \cdot n^2)$ ist, wobei $n$ die Anzahl der Knoten, $k$ die \glqq Größe der zu löschenden Knotenmenge\grqq~und $s$ die \glqq Größe einer kreiskritischen Menge\grqq~ist.
Für \textsc{Minimum Degree Deletion} wird zusätzlich der Parameter \glqq Anzahl zu löschender Knoten\grqq~betrachtet und für diesen W[1]-Schwere nachgewiesen.
Bezüglich jedem der Parameter \glqq Baumweite\grqq, \glqq Größe einer kreiskritischen Knotenmenge\grqq~sowie \glqq Größe einer kreiskritischen Kantenmenge\grqq~wird Festparameterhandhabbarkeit gezeigt.
Wir zeigen bezüglich \glqq Größe einer kreiskritischen Kantenmenge\grqq~einen Problemkern, dessen Größe linear in der Anzahl der Knoten ist.
Im Gegensatz dazu zeigen wir bezüglich der anderen beiden Parameter, dass es keinen Problemkern polynomieller Größe gibt, es sei denn, die Polynomialzeithierarchie kollabiert.
Wir beschreiben einen Suchbaumalgorithmus mit Laufzeit $O(2^{s} \cdot n^3)$, wobei $n$ die Anzahl der Knoten und $s$ die \glqq Größes einer kreiskritischen Kantenmenge\grqq~ist.
Zwei konkrete Festparameteralgorithmen werden bezüglich des Parameters \glqq Größe einer kreiskritischen Knotenmenge, welche nicht den ausgewiesenen Knoten enthält\grqq~präsentiert.
Wir weisen abschließend für \textsc{Bounded Degree Deletion} Festparameterhandhabbarkeit bezüglich des Parameters \glqq Größe einer kreiskritischen Kantenmenge\grqq~nach.
Ein entsprechender Algorithmus hat die Laufzeit $O(3^{s} \cdot n^2)$, wobei $n$ die Anzahl der Knoten und $s$ die \glqq Größes einer kreiskritischen Kantenmenge\grqq~ist.
\normalsize

\pagestyle{scrheadings}
\clearscrheadfoot
\ohead{\pagemark}
\automark[chapter]{chapter}
\cohead{\headmark}
\setheadsepline{0.5pt}        

\tableofcontents
\newpage


\chapter{Introduction}
\label{Introduction}

This work investigates the parameterized complexity of three related graph modification problems.
Each of these problems asks for a vertex subset of bounded size whose removal results in a graph that satisfies a problem-specific property.
The first considered problem \textsc{Minimum Indegree Deletion} has a directed graph, a distinguished vertex, and a positive integer $k$ as input.
The question is whether the distinguished vertex can be made the only vertex with minimum indegree by removing at most $k$ vertices.
The problem \textsc{Minimum Indegree Deletion} was introduced in the context of computational social choice~\cite{BU09}.
More precisely, one analyzes a so-called ``control'' scenario for a voting rule in which one asks to manipulate the result of the voting rule by removing a size-bounded candidate subset such that a distinguished candidate becomes the only winner.
The results regarding \textsc{Minimum Indegree Deletion} can be transferred to this scenario~\cite{BU09}.
The second considered problem \textsc{Minimum Degree Deletion} is defined analogously to \textsc{Minimum Indegree Deletion}, dealing with undirected graphs.
Here one asks whether it is possible to make a distinguished vertex the only vertex with minimum degree by removing a subset of vertices.
This problem can be motivated by a simple voting rule that is based on a social network.
The third considered problem \textsc{Bounded Degree Deletion}~\cite{Moser10} is also defined on undirected graphs, but in contrast to both previous problems it has no distinguished vertex as part of the input.
It asks for a set of $k$ vertices whose removal results in a graph with maximum degree $d$ for a specific positive integer $d$ that is part of the input.
The problem \textsc{Bounded Degree Deletion} has applications in the analysis of genetic networks~\cite{Moser10}.
In Chapter~\ref{Degree-based vertex deletion problems} we provide formal definitions and more details on applications.

Besides similarities in the definition, all three problems have in common that they are solvable in polynomial time on acyclic graphs, but NP-hard in general.
This property is the starting point for the investigation of the parameterized complexity with respect to three distinct parameters which measure the distance of the input graph to an acyclic graph in three different ways.
The first of these parameters is the ``treewidth'', being well-known in the analysis of undirected graphs.
The two other parameters measure in each case the number of vertices respectively edges whose removal results in an acyclic graph.
More precisely, we have the parameters ``size of a feedback vertex set'' and ``size of a feedback edge set'' respectively ``size of a feedback arc set'' in the directed case.
A formal definition of ``treewidth'' is given in Section~\ref{Tree decomposition and treewidth} whereas the other parameters are considered in detail in Chapter~\ref{Minimum Feedback Sets}.

\begin{figure}
\begin{center}
 \begin{tabular}{ c | c c c }
  parameter & MID & MDD & BDD\\
  \hline
  $t_w$ & ~ & \textbf{FPT} & open \\
  $s_v$ & \textbf{W[2]-hard, in XP} & \textbf{FPT} & open \\
  $s_{a/e}$ & \textbf{W[2]-hard, in XP} & \textbf{FPT} & \textbf{FPT} \\
  $k$ & W[2]-complete & \textbf{W[1]-hard, in XP} & W[2]-complete \\
  $d$ & FPT & FPT & FPT \\
  $(s_v,k)$ & \textbf{FPT} & \textbf{FPT} & open \\ 
 \end{tabular}\\
\end{center}
\small
\begin{center}
 \begin{tabular}{ r | l }
  ~ & parameter description \\
  \hline
  $t_w$ & treewidth of the input graph \\
  $s_v$ & size of a feedback vertex set \\
  $s_{a/e}$ & size of a feedback arc/edge set \\
  $k$ & size of a solution set \\
  $d$ & maximum degree of a vertex \\
 \end{tabular}\\
\end{center}
\caption{Overview of the parameterized complexity of \textsc{Minimum Indegree Deletion} (MID), \textsc{Minimum Degree Deletion} (MDD) and \textsc{Bounded Degree Deletion} (BDD).
         New results are in boldface.
         The remaining results are obtained from \cite{BU09} and \cite{Moser10}.}
\label{Results.}
\end{figure}

In what follows, we briefly describe our results, also see Figure~\ref{Results.}.
For \textsc{Minimum Indegree Deletion} we show that it is W[2]-hard, hence, that it is presumably not fixed-parameter tractable with respect to both parameters that are defined on acyclic graphs (see Section~\ref{Feedback vertex/arc set size as parameter}).
Considering the parameter ``number of vertices to delete'' in combination with one of both parameters, we show fixed-parameter tractability.
We describe a branch-and-bound algorithm whose running time is $O(s \cdot (k+1)^{s} \cdot n^2)$, where $n$ is the number of vertices, $k$ is the ``number of vertices to delete'', and $s$ is the ``size of a feedback set'' (see Section~\ref{Feedback vertex set size and solution size as combined parameter}).
For \textsc{Minimum Degree Deletion} we additionally consider the single parameter ``number of vertices to delete'' and show W[1]-hardness (see Section~\ref{Solution size as parameter for Minimum Degree Deletion}).
With respect to each of the parameters ``treewidth'', ``size of a feedback vertex set'', and ``size of a feedback edge set'' we show fixed-parameter tractability (see Section~\ref{MSO expression for Minimum Degree Deletion}).
We show with respect to ``size of a feedback edge set'' a problem kernel whose size is linear in the number of vertices (see Section~\ref{Size of a Feedback Edge Set as parameter for Minimum Degree Deletion}).
In contrast, we show that with respect to each of the two other parameters there is no problem kernel of polynomial size, assuming that the polynomial-time hierarchy does not collapse (see Section~\ref{No polynomial kernel with respect to s_v}).
We describe a simple search tree algorithm with running time $O(2^{s} \cdot n^3)$ where $n$  is the number of vertices and $s$ is the ``size of a feedback edge set''.
This algorithm complements the kernelization result from Section~\ref{Size of a Feedback Edge Set as parameter for Minimum Degree Deletion}.
We present two concrete fixed-parameter algorithms with respect to the parameter ``size of a feedback vertex set that does not contain the distinguished vertex''.
The first one uses the technique of integer linear programming (see Section~\ref{Integer linear programming.}) whereas the second one is a dynamic programming algorithm (see Section~\ref{Dynamic programming.}).
Besides the intractability results in Section~\ref{Feedback vertex/arc set size as parameter} and Section~\ref{Solution size as parameter for Minimum Degree Deletion}, the two algorithms with respect to the parameter ``size of a feedback vertex set that does not contain the distinguished vertex'' are the most demanding parts of this work.
Finally, we show fixed-parameter tractability for \textsc{Bounded Degree Deletion} with respect to ``size of a feedback edge set''.
A corresponding search-tree algorithm has running time $O(3^{s} \cdot n^2)$ where $n$ is the number of vertices and $s$ is the ``size of a feedback edge set''.

\paragraph{Organization of this work.}
The second chapter gives an overview of the considered problems and its applications.
A third chapter contains some definitions and computational aspects of the considered parameters measuring the distance of the input graph from an acyclic graph.
Each of the Chapters \ref{Minimum Indegree Deletion}-\ref{Bounded Degree Deletion} investigates one of the considered vertex deletion problem.
The last chapter provides a conclusion and outlooks to further investigations in this field.


\section{Preliminaries}
\label{Preliminaries}
We introduce the basic terms and methods which are necessary in this work.

\paragraph{Graph theory.}
\label{Graph theory}
All computational problems that are investigated in this work are graph problems.
An \textit{undirected graph} $G$ is a pair $(V,E)$, where $V$ is a finite set of \textit{vertices} and $E$ is a finite set of \textit{edges} which are unordered pairs of vertices.
A \textit{directed graph} $G$ is also a pair $(V,A)$, where $V$ is a finite set of vertices, but $A$ is a finite set of \textit{arcs} which are ordered pairs of vertices.
All graphs considered in this work are \textit{simple}, that is, there are no multi-arc/multi-edges, and they do not contain \textit{self-loops}, that is, no edge/arc from a vertex to itself.
Let $\{v_1,v_2\}$ be an edge of an undirected graph.
Then, $v_1$ is a \textit{neighbor} of $v_2$ and vice versa.
We denote $\deg(x)$ as the \textit{degree} of the vertex $x$, that is, the number of its neighbors.
Furthermore, the (open) \textit{neighborhood} of a vertex $v$ in an undirected graph $G = (V,E)$ is defined as $N(v) := \{ u \mid \{u,v\} \in E \}$.
Let $(v_1,v_2)$ be an arc of a directed graph.
We denote $v_1$ as \textit{inneighbor} of $v_2$ and $v_2$ as \textit{outneighbor} of $v_1$.
The \textit{indegree} of a vertex is the number of its inneighbors and the \textit{outdegree} of a vertex is the number of its outneighbors.
The (open) \textit{inneighborhood} of a vertex $v$ in a directed graph $G = (V,A)$ is defined as $N_{\neiin}(v) := \{ u \mid (u,v) \in A \}$ and the (open) \textit{outneighborhood} of a vertex $v$ in a directed graph $G = (V,A)$ is defined as $N_{\neiout}(v) := \{ u \mid (v,u) \in A \}$.
For a vertex set $S \subseteq V$, we write $G[S]$ to denote the graph induced by $S$ in $G$, that is, $G[S] := (S, \{ e \in E \mid e \subseteq S\})$ for an undirected graph $G=(V,E)$ respectively $G[S] := (S, \{ (x,y) \in A \mid x \in S \wedge y \in S \})$ for a directed graph $G=(V,A)$.
For a subset $S \subseteq V$, we also write $G - S$ instead of $G[V \setminus S]$.
We define $P_i:=(V_{P_i}:=\{v_1,\dots,v_i\},E_{P_i}:=\{\{a,b\} \mid 1 \leq a < b \leq i\})$ as \textit{path of length $i$} between $v_1$ and $v_i$.
Moreover, $C_i:=(V_{P_i},E_{P_i} \cup \{v_i,v_1\})$ is defined as \textit{cycle of length $i$}.
A graph that contains a cycle as subgraph is called \textit{cyclic}.
Otherwise we say the graph is \textit{acyclic}.
In directed graphs the terms \textit{path of length $i$} ($P_i:=(V_{P_i}:=\{v_1,\dots,v_i\},E_{P_i}:=\{(a,b) \mid 1 \leq a < b \leq i\})$), \textit{cycle of length $i$} ($C_i:=(V_{P_i},E_{P_i} \cup (v_i,v_1))$), and \textit{cyclic/acyclic} are defined analogously.
An undirected graph is \textit{connected} if there is a path between each two vertices.
A connected and acyclic graph is called a \textit{tree}.

\paragraph{Tree decomposition and treewidth.}
\label{Tree decomposition and treewidth}
As mentioned in the introduction, many hard graph problems are easy when restricted on acyclic graph.
The main question is: ``Why is an $\NP$-hard problem in P when restricted on a tree?''.
We need a measure of the ``tree-likeness'' of a given graph.
Robertson and Seymour~\cite{RS86} introduced the concept of tree-decompositions and treewidth of undirected graphs.
\begin{definition}
 Let $G=(V,E)$ be an undirected graph.
 A \textbf{tree decomposition} of $G$ is a pair $(\{X_i \mid i \in I\}, T)$ where each $X_i$ is a subset of $V$, called a \textbf{bag}, and $T$ is a tree with the elements of $I$ as nodes. The following three properties must hold:
 \begin{enumerate}
  \item $\bigcup_{i \in I} X_i = V$,
  \item for every edge $\{u,v\} \in E$, there is an $i \in I$ such that $\{u,v\} \subseteq X_i$, and
  \item for all $i,j,k \in I$, if $j$ lies on the path between $i$ and $k$ in $T$, then $X_i \cap X_k \subseteq X_j$.
 \end{enumerate}
 The \textbf{width} of the tree decomposition $(\{X_i \mid i \in I\}, T)$ equals $\max \{|X_i| \mid i \in I\}-1$.
 The \textbf{treewidth} of $G$ is the minimum $k$ such that $G$ has a tree decomposition of width $k$.
\end{definition}
Clearly, a (non-trivial) tree has treewidth 1, a cycle has treewidth $2$, and a complete graph with $n$ vertices has treewidth $n-1$.
The other two parameters which measure the ``degree of acyclicity'' are introduced in Chapter~\ref{Minimum Feedback Sets}.

\section{Parameterized complexity}
\label{Parameterized complexity}
Many interesting problems in computer science are computationally hard problems in worst case.
The most famous class of such hard problems is the class of $\NP$-hard problems.
The relation between $\complP$ (which includes the ``efficient solvable problems'')  and $\NP$ is not completely clear at the moment.
Even if $\complP=\NP$ it is not self-evident that we are able to design \textit{efficient} polynomial-time algorithms for each $\NP$-hard problem.
But we have to solve $\NP$-hard problems in practice.
Thus, according to the state of the art of computational complexity theory, $\NP$-hardness means that we only have algorithms with exponential running times to solve the corresponding problems exactly.
This is a huge barrier for practical applications.
There are different ways to cope with this situation: heuristic methods, randomized algorithms, average-case analysis (instead of worst-case) and approximation algorithms.
Unfortunately, none of these methods provides an algorithm that efficiently computes an optimal solution in the worst case.
Since there are situations where we need performance and optimality guarantee at least for a specified type of input, another way out is needed.
Fixed-parameter algorithms provide a possibility to redefine problems with several input parameters.
The main idea is to analyze the input structure to find parameters that are ``responsible for the exponential running time''.
The aim is to find such a parameter, whose values are constant or ``logarithmic in the input size'' or ``usually small enough'' in the problem instances of your application.
Thus, we can say something like ``if the parameter is small, we can solve our problem instances efficiently''.
\\
We will use the two-dimensional parameterized complexity theory \cite{DF99,Nie06,FG06} for studying the computational complexity of several graph problems.
A \textit{parameterized problem} (or language) $L$ is a subset $L \subseteq \Sigma^* \times \mathbb{N}$ for some finite alphabet $\Sigma$.
For an element $(x,k)$ of $L$, by convention $x$ is called \textit{problem instance}\footnote{Most parameterized problems originate from classical complexity problems.
One can interpret $x$ as the input of the original/non-parameterized problem.} and $k$ is the \textit{parameter}.
The two dimensions of parameterized complexity theory are the size of the input $n := |(x,k)|$ and the parameter value $k$, which is usually a non-negative integer.
A parameterized language is called \textit{fixed-parameter tractable} if we can determine in $f(k) \cdot n^{O(1)}$ time whether $(x,k)$ is an element of our language, where $f$ is a computable function only depending on the parameter $k$.
The class of fixed-parameter tractable problems is called $\FPT$.
Thus, it is very important to find good parameters.
In the following chapters, we need four of the core tools in the development of parameterized algorithms~\cite{Nie06}: data reduction rules (kernelization), (depth-bounded) search trees, dynamic programming and integer linear programs.

The idea of \textit{kernelization} is to transform any problem instance $x$ with parameter $k$ in polynomial time into a new instance $x'$ with  parameter $k'$ such that the size of $x'$ is bounded from above by some function only depending on $k$ and $k' \leq k$, and $(x,k) \in L$ if and only if $(x',k') \in L$.
The reduced instance $(x',k')$ is called \textit{problem kernel}.
This is done by \textit{data reduction rules}, which are transformations from one problem instance to another.
A data reduction rule that transforms $(x,k)$ to $(x',k')$ is called \textit{sound} if $(x,k) \in L$ if and only if $(x',k') \in L$.

Besides kernelization we use (depth-bounded) \textit{search trees algorithms}.
A search algorithm takes a problem as input and returns a solution to the problem after evaluating a number of possible solutions.
The set of all possible solutions is called the search space.
Depth-bounded search tree algorithms organize the systematic and exhaustive exploration of the search place in a tree-like manner.
Let $(x,k)$ denote the instance of a parameterized problem.
The search tree algorithm replaces $(x,k)$ by a set $H$ of smaller instances $(x_i,k_i)$ with $|x_i| < |x|$ and $k_i < k$ for $1 \leq i \leq |H|$.
Thus, the search tree size (number of nodes) is clearly bounded by $|H|^k$.
Since the running time of the replacement procedure is bounded by a polynomial in the instance size, a constant-size set $H$ always leads to a fixed-parameter algorithm with respect to $k$.
However, there are more refined methods to compute a better upper bound for the search tree size using the so-called branching vector, but they are not important in this work.

Another important technique used in this work is \textit{dynamic programming}.
It goes back to Bell~\cite{Bell03}.
As well as search trees dynamic programming makes exhaustive search, but is more efficient by avoiding the computation of subproblems more than once.
It is used when a problem exhibits the property of having \textit{optimal substructure}, that is, an optimal solution to the problem contains within it optimal solutions to subproblems.
Two further properties are important for a feasible dynamic programming: \textit{Independence}, that is, the solution of one subproblem does not affect the solution of another subproblem of the same problem, and \textit{overlapping subproblems}, that is, the same problem occurs as a subproblem of different problems.
One organizes the computation of the solutions for the subproblems in the \textit{dynamic programming table}.
An established way to compute the table which is used in this work is the \textit{bottom-up} computation.
Typically in fixed-parameter algorithms, the computation of a single table entry depends only on a constant number of parent entries and can be done in polynomial time.
The table size is bounded by a function only depending on the parameter.
Furthermore, the overall-solution can be computed from the completely filled table in polynomial time.

The technique of \textit{integer linear programming} is a special case of \textit{linear programming} which goes back to Kantorovich~\cite{K60}.
It is a technique for the optimization (minimization and maximization) of a (linear) \textit{objective function}, subject to linear equality and linear inequality contains.
In an \textit{integer linear program} (ILP), the unknown variables are required to be integers instead of real number as in a linear program.
Whereas solving a linear program is possible in polynomial time, integer linear programs are $\NP$-hard~\cite{Karp72}.
Anyhow, the technique can be used to show fixed-parameter tractability when the number of unknown variables is bounded by a function only depending on the parameter~\cite{L83,K87}.
More about these four techniques can be found in~\cite{Nie06}.

In many applications one is interested in deciding an $\NP$-hard problem or computing the optimal solution of the corresponding search or optimization problem.
Therefore, fixed-parameter tractability is a desired attribute of a problem together with a well-chosen parameter.
However, there are also some cases where intractability can also be a desired attribute.
For example voting rules seem to be ``more fair'' if a (at this point not more specified) manipulation or control of that rule is computationally hard.
Thus, parameterized intractability can be a positive result as well as negative result.
In this work, we use a characterization of parameterized problems that provides even more than only determining fixed-parameter tractability or intractability.
In analogy to the concepts of $\NP$-hardness, $\NP$-completeness, and polynomial-time many-to-one reductions in classical complexity theory, Downey and Fellows~\cite{DF99} developed a framework of reductions and a hierarchy of parameterized complexity classes.
\begin{definition}
 A \textbf{parameterized reduction} from a parameterized problem $L \subseteq \Sigma^* \times \mathbb{N}$ to another parameterized problem $L' \subseteq \Sigma^* \times \mathbb{N}$ is a function that, given an instance $(I,k)$, returns in time $f(k) \cdot \poly(|(I,k)|)$ an instance $(I',k')$, with $k'$ only depending on $k$, such that $(I,k) \in L$ if and only if $(I',k') \in L'$.
\end{definition}
A parameterized problem $L$ belongs to W[$t$] if $L$ can be reduced to a weighted satisfiability problem for the family of circuits of depth at most some function only depending on the parameter and weft at most $t$, where weft is the maximum number of gates with unrestricted fan-in on an input-output path in the circuit.
Similar to the classical complexity theory we denote a problem as W[$t$]-hard if there is a parameterized reduction from a problem that is already known to be W[t]-complete.
In this work, the most important thing we have to know about the W[$t$]-hierarchy is that every parameterized problem which is W[$t$]-hard for $t \geq 1$ is believed to be not fixed-parameter tractable.
(This holds under the separation hypothesis $\FPT \neq $W[$1$].)
A parameterized problem $L$ belongs to the class $\XP$ if it can be determined in $f(k) \cdot |x|^{g(k)}$ time whether $(x,k) \in L$ where $f$ and $g$ are computable functions only depending on the parameter $k$.
It holds that $$ \FPT \subseteq \text{W[1]} \subseteq \text{W[2]} \subseteq .. \subseteq \XP.$$
More details about this fields can be found in~\cite{DF99,FG06,Nie06}.
An established way of proving W[t]-hardness is to give a parameterized reduction from a problem that is already known to be W[t]-hard.
In this work, we use the following two problems:
\begin{verse}
 \textsc{Independent Set}\\
 \textit{Given:} An undirected graph $G=(V,E)$ and an integer $k \geq 1$.\\
 \textit{Question:} Is there a subset $V' \subseteq V$ of size at least $k$ such that there are no edges in $G[V']$?
\end{verse} 

\begin{verse}
 \textsc{Dominating Set}\\
 \textit{Given:} An undirected graph $G=(V,E)$ and an integer $k \geq 1$.\\
 \textit{Question:} Is there a subset $V' \subseteq V$ of size at most $k$ such that every vertex in $V$ either belongs to $V'$ or has a neighbor in $V'$?
\end{verse}
The \textsc{Independent Set} problem with respect to the parameter $k$ is known to be W[1]-complete and the \textsc{Dominating Set} problem with respect to the parameter $k$ is known to be W[2]-complete~\cite{DF99}.


\chapter{Degree-based vertex deletion problems}
\label{Degree-based vertex deletion problems}
In this work, we analyze three quite similar graph problems.
All problems get a directed or undirected graph as one part of the input.
In each problem one asks to delete a subset of vertices of bounded size to get a modified graph that satisfies a problem-specific property.
Furthermore, this property depends in each of the three problems on the degree of the vertices.
However, the particular problems have quite different applications.
In the following three sections, we give a formal definition and a motivation (``Why is this problem relevant?'').

\section{Minimum Indegree Deletion}
Recently, social choice problems became important in the fields of computational complexity and algorithmics.
In this context, the investigation of voting systems is a relevant area.
There are two recent surveys by Chevaleyre et al.~\cite{CELM07} and Faliszewki et al.~\cite{FHHR09}.
The most obvious application of voting systems might be political elections.
There are also several applications in the fields of rank aggregation and multi-agent systems.
Besides work that focuses on the problem to determine the winner or an optimal ranking for different voting systems, a significant number of papers also investigates how an external agent or a group of voters can influence the election in favor or disfavor of a distinguished candidate.
Concrete scenarios of influencing are manipulation~\cite{BTT89a,BPHSS08,CSL07,MPRZ08}, electoral control~\cite{BTT90,FHHR09b,HHR07}, lobbying~\cite{CFRS07}, and bribery~\cite{FHHR09b}.
This shows that investigations in this field are of high interest.
In this section we present a directed graph problem that is closely related to control in the Llull voting rule.
We start with an introduction to this rule and introduce the corresponding graph problem.

\paragraph{Llull voting.} 
Formally, an \textit{election} $(V,C)$ consists of a multiset of votes $V$ and a set of candidates $C$.
A \textit{vote} is a preference list over all candidates.
In an election, we either ask for a \textit{winner}, that is, one of the candidates who are ``best'' in the election, or for a \textit{unique winner}.
Of course, a unique winner does not always exist.
We only consider the unique-winner case for our control variant, but our results can be easily modified to work for the winner case as well.

The term Llull\footnote{Llull is the special case of Copeland$^\alpha$ where $\alpha=1$.} voting was introduced by Faliszewski et al.~\cite{FHHR09b}.
It is based on pairwise comparisons between candidates:
A candidate wins the pairwise contest against another candidate if it beats the other candidate in more than half of the votes.
The winner of a pairwise contest gets one point and the loser receives no point.
If two candidates are tied, both candidates get one point. 
A Llull winner is a candidate with the highest score.
It is used for example in sport tournaments, chess, or in football leagues, where the teams or players can be considered as candidates.

In the following we consider the concept ``control of a voting rule''.
To \textit{control} an election, an external or internal agent, traditionally called the \textit{chair}, is allowed or able to change the voting procedure to reach certain goals.
For example, a typical question is how many candidates the chair has to delete to make his/her favorite candidate a unique winner.
In most cases it is a desirable attribute of a voting rule to be either immune to control, that is, it is impossible to control the voting rule, or at least to be resistant to control, that is, the corresponding decision problem is $\NP$-hard~\cite{BTT90}.
Unfortunately, $\NP$-hardness does only imply computational hardness in the worst case.
There might be special inputs where is it easy to decide the corresponding problem.

Regarding the complexity of control, Llull voting is resistant to constructive candidate control and vulnerable for destructive candidate control~\cite{FHHR09b}.
Thus, investigating the parameterized complexity of control of Llull voting helps us to extend our knowledge about the ``danger of control of Llull election''.
This is what we do in Chapter~\ref{Minimum Indegree Deletion}.
Therefore, we introduce the directed graph problem which corresponds to candidate control in Llull elections.

\paragraph{Minimum Indegree Deletion.} 
\label{MID_intro}
A Llull election can be depicted by a directed graph where the candidates are represented as vertices and there is an arc from vertex~$c$ to vertex~$d$ if and only if the corresponding candidate~$c$ defeats the corresponding candidate~$d$ in the pairwise comparison contest.
Obviously, the Llull score of a candidate~$c$ can be considered as the total number of candidates minus the number of candidates that beat~$c$ in the pairwise comparison.
Thus, a Llull winner corresponds to a vertex with minimum indegree.
Of course, the deletion of a vertex corresponds to the deletion of a candidate in the election.
These observations motivate the introduction of the following directed graph problem which was originally introduced in~\cite{BU09}:
\begin{verse}
 \textsc{Minimum Indegree Deletion}\\
 \textit{Given:} A directed graph $D=(W,A)$, a distinguished vertex $w_c \in W$, and an integer $k \geq 1$.\\
 \textit{Question:} Is there a subset $W' \subseteq W \setminus \{w_c\}$ of size at most $k$ such that $w_c$ is the only vertex that has minimum indegree in $D - W'$?
\end{verse}
The equivalence of \textsc{Minimum Indegree Deletion} to constructive control by deleting candidates in Llull elections was shown in~\cite{BU09}.

\section{Minimum Degree Deletion}
\label{MDD_intro}
Constructive control by deleting candidates in Llull elections leads to the directed \textsc{Minimum Indegree Deletion}.
From the theoretical point of view, it is an intuitive task to investigate its undirected variant.
From the practical point of view, a corresponding voting problem can be formulated, too:
Given is a \textit{social network}, that is, an undirected graph $G=(V,E)$ where vertices correspond to subjects and edges correspond to relations between two subjects.
For example, consider the relation ``disharmony between two subjects'' and the subjects are candidates of a voting rule where the candidate with fewest disharmonies wins.
Now, the chair is asked to control the election by removing a specified number of candidates from the network to make a specific candidate the winner of the election.
The corresponding graph problem \textsc{Minimum Degree Deletion} is given in the following:

\begin{verse}
 \textsc{Minimum Degree Deletion}\\
 \textit{Given:} An undirected graph $G=(V,E)$, a distinguished vertex $w_c \in V$, and an integer $k \geq 1$.\\
 \textit{Question:} Is there a subset $V' \subseteq V \setminus \{w_c\}$ of size at most $k$ such that $w_c$ is the only vertex that has minimum degree in $G - V'$?
\end{verse}
Clearly, removing candidates from the social network corresponds to removing candidates in the graph $G$.
The equivalence of the graph problem and constructive control of the voting rule by removing candidates is easy to see.

\section{Bounded Degree Deletion}
\label{BDD_intro}
The first two problems ask for a vertex subset of a specific size whose removal satisfies the graph property ``only a distinguished vertex has minimum (in)degree''.
The problem which is introduced in this section no-longer has a distinguished candidate, but an upper bound on the degree.
This means, the underlaying graph property is ``the maximum degree of the vertices in the graph is bounded by $d$'' for a specific positive integer $d$ which is given in the input.
\textsc{Bounded Degree Deletion} was introduced in~\cite{Moser10} and is formally defined as follows:
\begin{verse}
 \textsc{Bounded Degree Deletion}\\
 \textit{Given:} An undirected graph $G=(V,E)$, and integers $d \geq 0$ and $k \geq 0$.\\
 \textit{Question:} Does there exists a subset $V' \subseteq V$ of size at most $k$ whose removal from $G$ yields a graph in which each vertex has degree at most $d$?
\end{verse} 
Whereas \textsc{Minimum Indegree Deletion} and \textsc{Minimum Degree Deletion} correspond to problems in computational social choice, \textsc{Bounded Degree Deletion} has applications in computational biology:
In the analysis of genetic networks based on micro array data, a central tool is to find cliques or ``near-cliques''~\cite{BCKLSWZ05,CLSQGWHMBLTMW05}.
Searching for cliques is a computational hard problem.
Here, \textsc{Vertex Cover} as complementary dual problem to clique detection could be used very successful instead of direct clique detection.
A mathematical concept for near-cliques is the concept of $s$-plexes, that are, subsets of vertices such that each $s$-plex vertex is connected to all other vertices in the $s$-plex but to $s - 1$.
Note that cliques are $1$-plexes.
\textsc{Bounded Degree Deletion} is the complementary dual problem to $s$-plex detection~\cite{Moser10}.
Since \textsc{Bounded Degree Deletion} is a vertex deletion problem for the hereditary graph property ``each vertex has degree at most $d$''~\cite{Moser10}, $\NP$-completeness is given due to a general result by Lewis et al.~\cite{LY80}.
It follows from the reduction in~\cite{LY80} that \textsc{Bounded Degree Deletion} cannot be approximated better than \textsc{Vertex Cover}.
Thus, an approximation lower bound of $1.36$ assuming P$\neq \NP$~\cite{DS05} as well as the APX-hardness of \textsc{Vertex Cover}~\cite{PY91} can be carried over.
The best known approximation factors are $2$~\cite{F98} for $d=1$ and $2+\log(d)$~\cite{OB03} for $d \geq 2$.

Finally we introduce some general terms which are used in the analysis of all three problems.
A yes-instance of each problem can be detected through finding a vertex subset of bounded size whose removal satisfies the problem-specific graph property.
We denote such subsets as \textit{solution sets}.
The result of removing a solution set from the input graph is called \textit{solution graph}.

\chapter{Feedback sets}
\label{Minimum Feedback Sets}
As mentioned in the introduction, parameterized algorithms are designed to confine the combinatorial explosion to a parameter of the input instance.
Popular parameters are, for example, the solution size for problems asking for a specific ``solution set''.
Besides very natural parameters such as ``number of edges'' in graph problems, ``number of candidates'' in voting problems or ``size of the alphabet'' in string-based problems, one is often interested in structural parameters like ``average distance'' or ``size of a vertex cover''~(see \cite{Nie10} for a survey).

There is a huge class of graph problems which are easy to solve on acyclic graphs, but $\NP$-hard in general.
Later on, we will see that every graph problem based on a graph property expressible by some specific (non-trivial) logic is solvable in linear time on trees.
The algorithms which follow directly from this result are quite impractical, but it is a good classification tool.
In most cases, there are simple and direct polynomial-time algorithms for problems whose input graph is acyclic.
In this way, a parameter that measures the ``degree of acyclicity'' of a graph is may be a good candidate for fixed-parameter tractability results.
Often, one uses the treewidth and designs algorithms that operate on tree decompositions.
In this work, we will also focus on two weaker parameters which will be introduced in this section.

There are three reasons for working also on other parameters that measure the ``degree of acyclicity'', different from treewidth:
The first reason is that computing an optimal tree decomposition (or even the treewidth) is difficult:
The corresponding decision problem is $\NP$-hard. 
Although there are some nice theoretical results and some work on practical computation of the treewidth of small graphs~\cite{B06,BGK08}, no efficient algorithm that computes an optimal tree decomposition is known.
Notably, it is an open question whether there is a constant-factor approximation for determining the treewidth of a graph~\cite{B08}.

The second reason is that there is no simple explicit analogy to treewidth on directed graphs:
A concept of ``directed treewidth'' was develop by Johnson et al.~\cite{JRST01}.
Independently, the concept ``DAG-width'' was introduced in~\cite{BDHK06} as well as in~\cite{Obdr06}.
Alternatively, another concept ``d-width'' was introduced by Safari~\cite{S05}.
Finally, a fourth concept ``Kelly-width'' was proposed by Hunter et al.~\cite{HK07}.
Although there is a tendency to ``Kelly-width''~\cite{MTV10}, it is not even clear which of these concepts is the best analogy to undirected treewidth.

The third reason is that there are also problems that are not fixed-parameter tractable with respect to the parameter treewidth.
Examples are \textsc{Capacitated Vertex Cover} and \textsc{Capacitated Dominating Set}~\cite{DLSV08}.
A weaker measure of the ``degree of acyclicity'' may lead to a parameter that allows for fixed-parameter tractability.
To the best of our knowledge, there was no work with main focus on the parameterized complexity with respect to parameters that ``measure the distance of the input graph to an acyclic graph'' so far.

Informally speaking, a \textit{feedback set} is a substructure of the graph such that its deletion makes the graph acyclic.
Depending on the type of input (directed or undirected graph) and on the structural elements (vertices or edges/arcs), we consider four parameters (see Figure~\ref{Overview: Feeback parameters.}). 

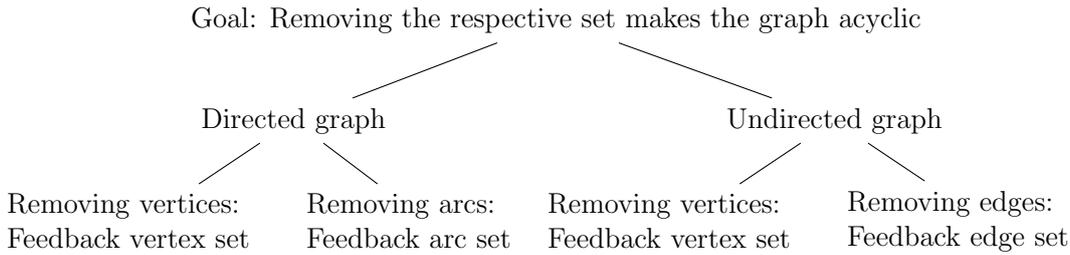
\begin{figure*}
\begin{center}
 \begin{tikzpicture}[>=latex,join=bevel,scale=0.9]
  \node (root_desc) [transform shape] {Goal: Removing the respective set makes the graph acyclic};
  \node (root) [below=of root_desc] [transform shape] {};

  \node (directed_desc) [left=2.2 of root] [transform shape] {Directed graph};
  \node (directed) [below=of directed_desc] [transform shape] {};
  \node (undirected_desc) [right=2.2 of root] [transform shape] {Undirected graph};
  \node (undirected) [below=of undirected_desc] [transform shape] {};

  \node (directed_vertices) [left=-0.1 of directed] [transform shape,text width=4cm] {Removing vertices: Feedback vertex set};
  \node (directed_arcs) [right=-0.1 of directed] [transform shape,text width=3.5cm] {Removing arcs: Feedback arc set};
  \node (undirected_vertices) [left=-0.1 of undirected] [transform shape,text width=4cm] {Removing vertices: Feedback vertex set};
  \node (undirected_edges) [right=-0.1 of undirected] [transform shape,text width=4cm] {Removing edges: Feedback edge set};


  \draw [-] (root_desc) to (directed_desc);
  \draw [-] (root_desc) to (undirected_desc);

  \draw [-] (directed_desc) to (directed_vertices);
  \draw [-] (directed_desc) to (directed_arcs);
  \draw [-] (undirected_desc) to (undirected_vertices);
  \draw [-] (undirected_desc) to (undirected_edges);

 \end{tikzpicture}
\end{center}
\caption{Overview: We consider four types of feedback sets. The corresponding parameters are the cardinalities of these sets.}
\label{Overview: Feeback parameters.}
\end{figure*}
We start with the formal definition of the directed case:
Let $G=(V,A)$ be a directed graph.
One calls a subset of vertices $V' \subseteq V$ \textit{(directed) feedback vertex set} if $G - V'$ is acyclic.
A subset of arcs $A' \subseteq A$ is called \textit{feedback arc set} if $G' := (V, A \setminus A')$ is acyclic.
Analogously, let $G=(V,E)$ be an undirected graph.
One calls a subset of vertices $V' \subseteq V$ \textit{(undirected) feedback vertex set} if $G - V'$ is acyclic.
A subset of edges $E' \subseteq A$ is called \textit{feedback edge set} if $G' := (V, E \setminus E')$ is acyclic.
Determining the sizes of these feedback sets leads directly to the following decision problems:
\begin{verse}
 \textsc{Undirected Feedback Vertex Set}\\
 \textit{Given:} An undirected graph $G=(V,E)$ and an integer $k \geq 1$.\\
 \textit{Question:} Is there an undirected feedback vertex set of size at most $k$?
\end{verse}
\textsc{Undirected Feedback Vertex Set} is $\NP$-complete~\cite{Karp72}.
With deterministic fixed-parameter algorithms it can be solved in $O(5^k\cdot k \cdot n^2)$ time~\cite{CFLLV08}.
There is a randomized algorithm which solves \textsc{Undirected Feedback Vertex Set} in $O(c \cdot 4^k \cdot k n)$ time by finding an undirected feedback vertex set of size $k$ with probability at least $1 - (1 - 4^{-k})^{c 4^k}$ for an arbitrary constant $c$.

\begin{verse}
 \textsc{Directed Feedback Vertex Set}\\
 \textit{Given:} A directed graph $G=(V,A)$ and an integer $k \geq 1$.\\
 \textit{Question:} Is there a directed feedback vertex set of size at most $k$?
\end{verse}
\textsc{Directed Feedback Vertex Set} is also $\NP$-complete.
This follows directly from the reduction given by Karp in~\cite{Karp72}.
However, recent studies proved its fixed-parameter tractability, showing that it can be solved in $O(4^k \cdot k! \cdot n^4 \cdot k^3)$ time~\cite{CLLSR08}.

\begin{verse}
 \textsc{Feedback Arc Set}\\
 \textit{Given:} A directed graph $G=(V,A)$ and an integer $k \geq 1$.\\
 \textit{Question:} Is there a feedback arc set of size at most $k$?
\end{verse}
The directed \textsc{Feedback Arc Set} is $\NP$-complete~\cite{Karp72}.
It was shown by Even et al.~\cite{ENSS98} that \textsc{Feedback Arc Set} and \textsc{Directed Feedback Vertex Set} can be reduced in linear time from one to each other retaining the parameter ``size of a feedback set''.
Hence, \textsc{Feedback Arc Set} can also be solved in $O(4^k \cdot k! \cdot n^4 \cdot k^3)$ time.

\begin{verse}
 \textsc{Feedback Edge Set}\\
 \textit{Given:} A undirected graph $G=(V,E)$ and an integer $k \geq 1$.\\
 \textit{Question:} Is there a feedback edge set of size at most $k$?
\end{verse}
The undirected \textsc{Feedback Edge Set} is polynomial-time solvable.
It is easy to see that a (minimum) feedback edge set can be found by depth-first search.
One just has to find a spanning tree.
A minimal feedback edge set consists of the edges that are not in this tree.

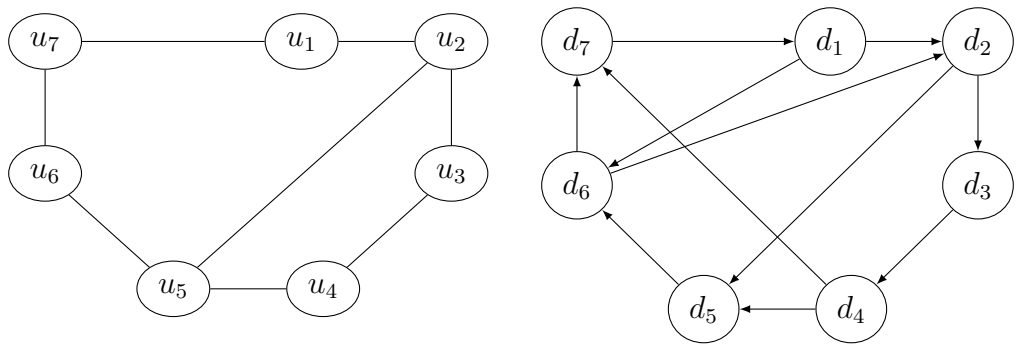
\begin{figure*}
\begin{center}
 \begin{tikzpicture}[>=latex,join=bevel,scale=1.0]
  \node (v_1) [draw,transform shape,ellipse] {$u_1$};
  \node (v_2) [right=of v_1] [draw,transform shape,ellipse] {$u_2$};
  \node (v_3) [below=of v_2] [draw,transform shape,ellipse] {$u_3$};
  \node (v_4) [below left=of v_3] [draw,transform shape,ellipse] {$u_4$};
  \node (v_5) [left=of v_4] [draw,transform shape,ellipse] {$u_5$};
  \node (v_6) [above left=of v_5] [draw,transform shape,ellipse] {$u_6$};
  \node (v_7) [above=of v_6] [draw,transform shape,ellipse] {$u_7$};

  \draw [-] (v_1) to (v_2);
  \draw [-] (v_2) to (v_3);
  \draw [-] (v_3) to (v_4);
  \draw [-] (v_4) to (v_5);
  \draw [-] (v_5) to (v_6);
  \draw [-] (v_6) to (v_7);
  \draw [-] (v_7) to (v_1);
  \draw [-] (v_2) to (v_5);

  \node (dv_1) [right=6 of v_1] [draw,transform shape,ellipse] {$d_1$};
  \node (dv_2) [right=of dv_1] [draw,transform shape,ellipse] {$d_2$};
  \node (dv_3) [below=of dv_2] [draw,transform shape,ellipse] {$d_3$};
  \node (dv_4) [below left=of dv_3] [draw,transform shape,ellipse] {$d_4$};
  \node (dv_5) [left=of dv_4] [draw,transform shape,ellipse] {$d_5$};
  \node (dv_6) [above left=of dv_5] [draw,transform shape,ellipse] {$d_6$};
  \node (dv_7) [above=of dv_6] [draw,transform shape,ellipse] {$d_7$};

  \draw [->] (dv_1) to (dv_2);
  \draw [->] (dv_2) to (dv_3);
  \draw [->] (dv_3) to (dv_4);
  \draw [->] (dv_4) to (dv_5);
  \draw [->] (dv_5) to (dv_6);
  \draw [->] (dv_6) to (dv_7);
  \draw [->] (dv_7) to (dv_1);
  \draw [->] (dv_2) to (dv_5);
  \draw [->] (dv_6) to (dv_2);
  \draw [->] (dv_4) to (dv_7);
  \draw [->] (dv_1) to (dv_6);
 \end{tikzpicture}
\end{center}
\caption{Feedback sets examples. Here we have an undirected and a directed graph. In the undirected graph there is a feedback vertex set of size one ($\{u_2\}$) and a feedback edge set of size two (for example $\{\{u_2,u_5\},\{u_6,u_7\}\}$). In the directed graph there is a feedback vertex set of size two ($\{d_1,d_2\}$) and a feedback arc set of size three (for example $\{(d_2,d_3),(d_2,d_5),(d_1,d_6)\}$).}
\label{Feedback sets examples.}
\end{figure*}

In the following, we always talk about \textit{feedback vertex sets} and the problem \textsc{Feedback Vertex Set} in both cases, the directed and the undirected, whenever it is clear which of them is considered.
Furthermore, the parameters are called ``size of a feedback vertex set'' and so on.
Of course, in most cases it is useful to know the size of the respective minimum feedback set.
However, the algorithms do not depend on the minimality of the parameter value.
This is an advantage when we use heuristics or approximation algorithms to find the feedback sets.
There are cases where one has to know a concrete feedback set and the running time depends on the size of this set (see Section~\ref{Integer linear programming.}, Section~\ref{Dynamic programming.}, and Section~\ref{Size of a Feedback Edge Set as parameter for Bounded Degree Deletion}).
In contrast, there are also cases where the size of the feedback set is used indirectly to prove the worst-case running time (see Section~\ref{Feedback vertex set size and solution size as combined parameter}, Section~\ref{Size of a Feedback Edge Set as parameter for Minimum Degree Deletion}, and Section~\ref{MSO expression for Minimum Degree Deletion}).
There, it is not even necessary to know any feedback vertex set.

Some results can be carried over from one parameter to the other:
\begin{definition}
\label{stronger}
 Let $a$ and $b$ denote two parameters.
 If $a \leq b$ for each instance, then one says $a$ is \textbf{stronger} than $b$ (and $b$ is \textbf{weaker} than $a$).
\end{definition}
Hardness results like W[$t$]-hardness for some positive integer $t$ or non-existence of a problem kernel can be carried over from the weaker to the stronger parameter.
In contrast, results like membership to a parameterized complexity class can be carried over from the stronger to the weaker parameter.
It is easy to see that the treewidth of a graph is at most the size of a feedback vertex set of the same graph.
Furthermore, for each feedback edge set respectively for each feedback arc set there is a feedback vertex set of at most the same size.
Hence, ``treewidth'' is a stronger parameter than ``size of a feedback vertex set'' which is a stronger parameter than ``size of a feedback edge set'' respectively ``size of a feedback arc set''.

\chapter{Minimum Indegree Deletion}
\label{Minimum Indegree Deletion}
In this chapter, we analyze the parameterized complexity of \textsc{Minimum Indegree Deletion}.
The problem is motivated as graph problem corresponding to constructive control by deleting candidates for Llull voting (see Chapter~\ref{Degree-based vertex deletion problems}).
The inputs are a directed graph $D=(W,A)$, a distinguished vertex $w_c \in W$, and an integer $k \geq 1$.
The question is whether there is a subset $W' \subseteq W \setminus \{w_c\}$ of size at most $k$ such that $w_c$ is the only vertex that has minimum indegree in $D - W'$.
In most scenarios, computational hardness for control of a voting rule is a desirable attribute.
However, $\NP$-hardness should only be a first step in this regard, because the hardness does not necessarily hold for special cases of the input which can be typically in real-world instances.
A parameterized analysis will help to discuss computational hardness for special types of input to the voting rule.
We present intractability and tractability results for the parameters $s_v:=$``size of a feedback vertex set'', $s_a:=$``size of a feedback arc set'', $k:=$``number of vertices to remove'', and $d:=$``maximum degree of a vertex'' and their combinations.
\begin{figure}[!b]
\begin{center}
 Parameterized complexity:\\
 \begin{tabular}{ r | c c }
  ~ & single parameter & combined with $k$ \\
  \hline
  $s_v$ & \textbf{W[2]-hard, in XP} & \textbf{FPT} \\
  $s_a$ & \textbf{W[2]-hard, in XP} & \textbf{FPT} \\
  $k$ & W[2]-complete & ~ \\
  $d$ & FPT & FPT \\
 \end{tabular}\\
\end{center}
\begin{center}
 \begin{tabular}{ r | l }
  ~ & parameter description \\
  \hline
  $s_v$ & size of a feedback vertex set \\
  $s_a$ & size of a feedback arc set \\
  $k$ & size of a solution set \\
  $d$ & maximum degree of a vertex \\
 \end{tabular}\\
\end{center}
\caption{Overview of the parameterized complexity of \textsc{Minimum Indegree Deletion}.
         New results are in boldface.
         The remaining results are obtained from \cite{BU09}.}
\label{MID results.}
\end{figure}
Our results are summarized in Figure~\ref{MID results.}.
%

\section{Known results}
\label{Known results for Minimum Indegree Deletion}
We briefly discuss previous result~\cite{BU09}:
There is a simple polynomial-time algorithm that solves \textsc{Minimum Indegree Deletion} in acyclic directed graphs.
In contrast, \textsc{Minimum Indegree Deletion} is W[2]-complete with respect to $k:=$``size of a solution set''.
This even holds if the input graph is a tournament.
With respect to the parameter $d:=$``maximum indegree of a vertex'' \textsc{Minimum Indegree Deletion} is fixed-parameter tractable.
Inspired by the result for acyclic graphs we study parameters measuring the ``degree of acyclicity'' of the input graph in the following.

\section{Feedback vertex/arc set size as parameter}
\label{Feedback vertex/arc set size as parameter}
Here, we show that \textsc{Minimum Indegree Deletion} is W[2]-hard with respect to the parameter $s_v:=$``size of a feedback vertex set'' and the parameter $s_a:=$``size of a feedback arc set''.
To this end, we need the concept of ``domination in an undirected graph''.
\begin{definition}
 Let $G=(V,E)$ be an undirected graph.
 We say that a vertex $d$ \textbf{dominates} a vertex $v$ if $d=v$ or $d$ is a neighbor of $v$.
 A subset $D \subseteq V$ is called \textbf{dominating set} if every vertex in $V$ is dominated by a vertex in $D$.
\end{definition}
In the following, we provide parameterized reductions from \textsc{Dominating Set}, which is W[2]-complete with respect to the parameter $k:=$``size of the dominating set''.
\begin{verse}
 \textsc{Dominating Set}\\
 \textit{Given:} An undirected graph $G=(V,E)$ and an integer $k \geq 1$.\\
 \textit{Question:} Is there a dominating set of size at most $k$?
\end{verse} 

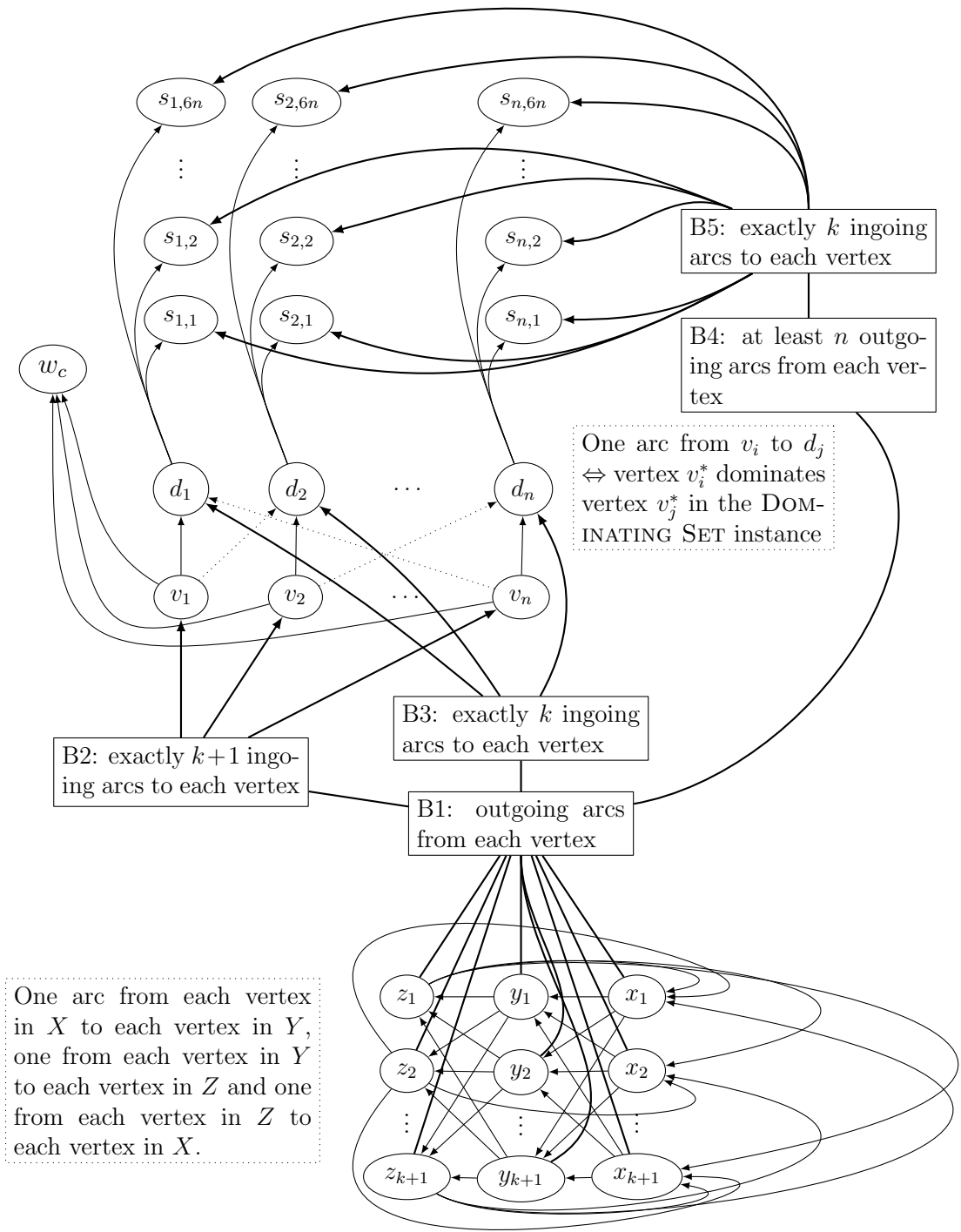
\begin{figure*}
\begin{tikzpicture}[>=latex,join=bevel,scale=0.9]
  \node (w_c) [draw,ellipse] {$w_c$};

  \node (d_1) [below right=2 of w_c] [draw,transform shape,ellipse] {$d_1$};
  \node (d_2) [right=of d_1] [draw,transform shape,ellipse] {$d_2$};
  \node (d_dots) [right=of d_2] [transform shape] {$\dots$};
  \node (d_n) [right=of d_dots] [draw,transform shape,ellipse] {$d_n$};

  \node (v_1) [below=of d_1] [draw,transform shape,ellipse] {$v_1$};
  \node (v_2) [right=of v_1] [draw,transform shape,ellipse] {$v_2$};
  \node (v_dots) [right=of v_2] [transform shape] {$\dots$};
  \node (v_n) [right=of v_dots] [draw,transform shape,ellipse] {$v_n$};

  \draw [->] (v_1) to (d_1);
  \draw [->] (v_2) to (d_2);
  \draw [->] (v_n) to (d_n);

  \draw [->,dotted] (v_1) to (d_2);
  \draw [->,dotted] (v_2) to (d_n);
  \draw [->,dotted] (v_n) to (d_1);
  \node (possExpl) [above right=0.75 of v_n] [draw,transform shape,text width=4.1cm,rectangle,dotted,,text justified] {One arc from $v_i$ to $d_j$ $\Leftrightarrow$ vertex $v_i^*$ dominates vertex $v_j^*$ in the \textsc{Dominating Set} instance};

  \draw [->] (v_1) to [out=150,in=290] (w_c);
  \draw [->] (v_2) .. controls +(195:3.5) .. (w_c);
  \draw [->] (v_n) .. controls +(190:8) .. (w_c);

  \node (s_1_1) [above=2 of d_1] [draw,transform shape,ellipse] {$s_{1,1}$};
  \node (s_1_2) [above=0.5 of s_1_1] [draw,transform shape,ellipse] {$s_{1,2}$};
  \node (s_1_dots) [above=0.5 of s_1_2] [transform shape] {$\vdots$};
  \node (s_1_6n) [above=0.2 of s_1_dots] [draw,transform shape,ellipse] {$s_{1,6n}$};

  \draw [->] (d_1) to [out=110,in=230] (s_1_1);
  \draw [->] (d_1) to [out=110,in=230] (s_1_2);
  \draw [->] (d_1) to [out=110,in=230] (s_1_6n);

  \node (s_2_1) [above=2 of d_2] [draw,transform shape,ellipse] {$s_{2,1}$};
  \node (s_2_2) [above=0.5 of s_2_1] [draw,transform shape,ellipse] {$s_{2,2}$};
  \node (s_2_dots) [above=0.5 of s_2_2] [transform shape] {$\vdots$};
  \node (s_2_6n) [above=0.2 of s_2_dots] [draw,transform shape,ellipse] {$s_{2,6n}$};

  \draw [->] (d_2) to [out=110,in=230] (s_2_1);
  \draw [->] (d_2) to [out=110,in=230] (s_2_2);
  \draw [->] (d_2) to [out=110,in=230] (s_2_6n);

  \node (s_n_1) [above=2 of d_n] [draw,transform shape,ellipse] {$s_{n,1}$};
  \node (s_n_2) [above=0.5 of s_n_1] [draw,transform shape,ellipse] {$s_{n,2}$};
  \node (s_n_dots) [above=0.5 of s_n_2] [transform shape] {$\vdots$};
  \node (s_n_6n) [above=0.2 of s_n_dots] [draw,transform shape,ellipse] {$s_{n,6n}$};

  \draw [->] (d_n) to [out=110,in=230] (s_n_1);
  \draw [->] (d_n) to [out=110,in=230] (s_n_2);
  \draw [->] (d_n) to [out=110,in=230] (s_n_6n);

  \node (y_1) [below=6 of v_n] [draw,transform shape,ellipse] {$y_1$};
  \node (y_2) [below=0.5 of y_1] [draw,transform shape,ellipse] {$y_2$};
  \node (y_dots) [below=0 of y_2] [transform shape] {$\vdots$};
  \node (y_k) [below=0.2 of y_dots] [draw,transform shape,ellipse] {$y_{k+1}$};

  \node (x_1) [right=1 of y_1] [draw,transform shape,ellipse] {$x_1$};
  \node (x_2) [below=0.5 of x_1] [draw,transform shape,ellipse] {$x_2$};
  \node (x_dots) [below=0 of x_2] [transform shape] {$\vdots$};
  \node (x_k) [below=0.2 of x_dots] [draw,transform shape,ellipse] {$x_{k+1}$};

  \node (z_1) [left=1 of y_1] [draw,transform shape,ellipse] {$z_1$};
  \node (z_2) [below=0.5 of z_1] [draw,transform shape,ellipse] {$z_2$};
  \node (z_dots) [below=0 of z_2] [transform shape] {$\vdots$};
  \node (z_k) [below=0.2 of z_dots] [draw,transform shape,ellipse] {$z_{k+1}$};

  \node (XYZExpl) [left=1 of z_2] [draw,transform shape,text width=5cm,rectangle,dotted,text justified] {
  One arc from each vertex in $X$ to each vertex in $Y$, one from each vertex in $Y$ to each vertex in $Z$ and
  one from each vertex in $Z$ to each vertex in $X$. 
  };

  \draw [->] (x_1) to (y_1);
  \draw [->] (x_1) to (y_2);
  \draw [->] (x_1) to (y_k);
  \draw [->] (x_2) to (y_1);
  \draw [->] (x_2) to (y_2);
  \draw [->] (x_2) to (y_k);
  \draw [->] (x_k) to (y_1);
  \draw [->] (x_k) to (y_2);
  \draw [->] (x_k) to (y_k);

  \draw [->] (y_1) to (z_1);
  \draw [->] (y_1) to (z_2);
  \draw [->] (y_1) to (z_k);
  \draw [->] (y_2) to (z_1);
  \draw [->] (y_2) to (z_2);
  \draw [->] (y_2) to (z_k);
  \draw [->] (y_k) to (z_1);
  \draw [->] (y_k) to (z_2);
  \draw [->] (y_k) to (z_k);

  \draw [->] (z_1) .. controls +(30:2) and +(10:2.5) .. (x_1);
  \draw [->] (z_1) .. controls +(30:2.5) and +(10:7.5) .. (x_2);
  \draw [->] (z_1) .. controls +(30:3) and +(10:12.5) .. (x_k);
  \draw [->] (z_2) .. controls +(130:6) and +(360:5) .. (x_1);
  \draw [->] (z_2) .. controls +(330:2) and +(340:2.5) .. (x_2);
  \draw [->] (z_2) .. controls +(230:6) and +(360:5) .. (x_k);
  \draw [->] (z_k) .. controls +(330:3) and +(350:12.5) .. (x_1);
  \draw [->] (z_k) .. controls +(330:2.5) and +(350:7.5) .. (x_2);
  \draw [->] (z_k) .. controls +(330:2) and +(350:2.5) .. (x_k);

  \node (T) [above=2 of y_1] [draw,transform shape,text width=3.5cm,rectangle,text justified] {B1: outgoing arcs from each vertex};
  \node (TV) [below=2 of v_1] [draw,transform shape,text width=4cm,rectangle,text justified] {B2: exactly $k+1$ ingoing arcs to each vertex};
  \node (TD) [below=1.3 of v_n] [draw,transform shape,text width=4cm,rectangle,text justified] {B3: exactly $k$ ingoing arcs to each vertex};
  \node (TS) [right=2 of s_n_2] [draw,transform shape,text width=4cm,rectangle,text justified] {B5: exactly $k$ ingoing arcs to each vertex};
  \node (TS_o) [below=0.75 of TS] [draw,transform shape,text width=4cm,rectangle,text justified] {B4: at least $n$ outgoing arcs from each vertex};

  \draw [-,thick] (T) to (TV);
  \draw [-,thick] (T) to (TD);
  \draw [-,thick] (T) to [out=10,in=315] (TS_o);
  \draw [-,thick] (TS_o) to (TS);

  \draw [-,thick] (x_1) to (T);
  \draw [-,thick] (x_2) to (T);
  \draw [-,thick] (x_k) to (T);
  \draw [-,thick] (y_1) to (T);
  \draw [-,thick] (y_2) to [out=40,in=270] (T);
  \draw [-,thick] (y_k) to [out=30,in=270] (T);
  \draw [-,thick] (z_1) to (T);
  \draw [-,thick] (z_2) to (T);
  \draw [-,thick] (z_k) to (T);

  \draw [->,thick] (TV) to (v_1);
  \draw [->,thick] (TV) to (v_2);
  \draw [->,thick] (TV) to (v_n);

  \draw [->,thick] (TD) to [bend right=5] (d_1);
  \draw [->,thick] (TD) to [bend right=10] (d_2);
  \draw [->,thick] (TD) to [bend right=30] (d_n);

  \draw [->,thick] (TS) to [out=90,in=0] (s_n_6n);
  \draw [->,thick] (TS) to [out=157,in=0] (s_n_2);
  \draw [->,thick] (TS) to [out=210,in=0] (s_n_1);
  
  \draw [->,thick] (TS) to [out=90,in=15] (s_2_6n);
  \draw [->,thick] (TS) to [out=157,in=15] (s_2_2);
  \draw [->,thick] (TS) to [out=210,in=340] (s_2_1);

  \draw [->,thick] (TS) to [out=90,in=30] (s_1_6n);
  \draw [->,thick] (TS) to [out=157,in=30] (s_1_2);
  \draw [->,thick] (TS) to [out=210,in=340] (s_1_1);
 \end{tikzpicture}
 \caption{\textsc{Minimum Indegree Deletion} instance obtained from a parameterized reduction from a \textsc{Dominating Set} instance.}
 \label{Reduction from Dominating Set (with parameter $k$) to Minimum Indegree Deletion (with parameter $s_v$/$s_a$)}
\end{figure*}

\paragraph{Parameterized reduction.}
The following reduction is illustrated in Figure~\ref{Reduction from Dominating Set (with parameter $k$) to Minimum Indegree Deletion (with parameter $s_v$/$s_a$)}.
Given a \textsc{Dominating Set} instance $(G^*=(V^*,E^*),k)$ with $V^*=\{v^*_1,v^*_2,\dots,v^*_n\}$, we construct a directed graph $G$ with feedback vertex set size $k+1$ and feedback arc set size $(k+1)^2$ such that $(G,w_c,n-k)$ is a yes-instance of \textsc{Minimum Indegree Deletion} if and only if $(G^*,k)$ is a yes-instance of \textsc{Dominating Set}.
This means we have a parameterized reduction if we can do the construction in polynomial time.
The vertex set of $G$ consists of $w_c$ and the union of the following disjoint vertex sets:
\begin{itemize}
 \item $V, D$, each containing one vertex for every vertex in $V^*$,
       \begin{itemize}
        \item $V:=\{v_i \mid i\in\{1,\dots,n\}\}$ represents the set of ``dominating vertices'' which is illustrated by the nodes $v_1$, $v_2$ and $v_n$ in Figure~\ref{Reduction from Dominating Set (with parameter $k$) to Minimum Indegree Deletion (with parameter $s_v$/$s_a$)},
        \item $D:=\{d_i \mid i\in\{1,\dots,n\}\}$ represents the set of ``dominated vertices'' which is illustrated by the nodes $d_1$, $d_2$ and $d_n$ in Figure~\ref{Reduction from Dominating Set (with parameter $k$) to Minimum Indegree Deletion (with parameter $s_v$/$s_a$)}.
       \end{itemize}
       The arcs between vertices in $V$ and $D$ ensure that ``undominated vertices in a solution graph would have a lower indegree than $w_c$''.
       They are illustrated by dotted arrows in Figure~\ref{Reduction from Dominating Set (with parameter $k$) to Minimum Indegree Deletion (with parameter $s_v$/$s_a$)}.
 \item $X, Y, Z$, each containing $k+1$ vertices,
       \begin{itemize}
        \item $X:=\{x_i \mid i\in\{1,\dots,k+1\}\}$,
        \item $Y:=\{y_i \mid i\in\{1,\dots,k+1\}\}$,
        \item $Z:=\{z_i \mid i\in\{1,\dots,k+1\}\}$.
       \end{itemize}
        These vertex sets are designed to ``increase the indegree of vertices without increasing the size of a (minimum) feedback vertex set''.
        We will see that each of them is a feedback vertex set and every cycle in $G$ passes trough at least one vertex in $X$, $Y$, and $Z$.
        The set $X$ is illustrated by the nodes $x_1$, $x_2$, and $x_{k+1}$, $Y$ is illustrated by $y_1$, $y_2$, and $y_{k+1}$, and $Z$ is illustrated by $z_1$, $z_2$, and $z_{k+1}$ in Figure~\ref{Reduction from Dominating Set (with parameter $k$) to Minimum Indegree Deletion (with parameter $s_v$/$s_a$)}.
  \item $S_1, \dots, S_n$ each containing $6n$ vertices.\\
        These vertex sets are designed to ensure that the set of removed vertices corresponding to any yes-instance $(G,w_c,n-k)$ of \textsc{Minimum Indegree Deletion} only contains vertices of $V$.
        We will show that removing vertices that are not in $V$ would be ``punished by being forced to remove more than $n$ further vertices in $\bigcup_{i=1}^n S_i$''.
        Each set $S_i$ is illustrated by the nodes $s_{i,1}$, $s_{i,2}$, and $s_{i,6n}$ in Figure~\ref{Reduction from Dominating Set (with parameter $k$) to Minimum Indegree Deletion (with parameter $s_v$/$s_a$)}.
\end{itemize}
The arcs of $G$ are designed as follows:
\begin{figure}
\begin{center}
\begin{tabular}{l | l}
 start set & end set \\
 \hline
 $\{x_1,\dots,x_k\}$ & $\{s_{i,j} \mid i \in \{1,\dots,n\} \text{ and } j \in \{1,\dots,n\}\}$ \\
 $\{x_2,\dots,x_{k+1}\}$ & $\{s_{i,j} \mid i \in \{1,\dots,n\} \text{ and } j \in \{n+1,\dots,2n\}\}$ \\
 $\{y_1,\dots,y_k\}$ to & $\{s_{i,j} \mid i \in \{1,\dots,n\} \text{ and } j \in \{2n+1,\dots,3n\}\}$ \\
 $\{y_2,\dots,y_{k+1}\}$ & $\{s_{i,j} \mid i \in \{1,\dots,n\} \text{ and } j \in \{3n+1,\dots,4n\}\}$ \\
 $\{z_1,\dots,z_k\}$ & $\{s_{i,j} \mid i \in \{1,\dots,n\} \text{ and } j \in \{4n+1,\dots,5n\}\}$ \\
 $\{z_2,\dots,z_{k+1}\}$ & $\{s_{i,j} \mid i \in \{1,\dots,n\} \text{ and } j \in \{5n+1,\dots,6n\}\}$ \\
\end{tabular}
\end{center}
\caption{Concrete assignment of the arcs that are incident to vertices in $\bigcup_{i=1}^n S_i$.
         For each row of the table there is one arc from each vertex from the start set to each vertex from the end set.}
\label{Concrete assignement of the arcs}
\end{figure}
\begin{itemize}
 \item There is an arc from $v_i$ to $d_j$ if and only if $v^*_i$ dominates $v^*_j$.
       These arcs are illustrated by dotted arrows in Figure~\ref{Reduction from Dominating Set (with parameter $k$) to Minimum Indegree Deletion (with parameter $s_v$/$s_a$)}.
 \item There is an arc from each vertex in $V$ to $w_c$.
       These arcs are illustrated by thin arrows from the nodes $v_1$, $v_2$, and $v_n$ to the node $w_c$ in Figure~\ref{Reduction from Dominating Set (with parameter $k$) to Minimum Indegree Deletion (with parameter $s_v$/$s_a$)}.
 \item There is an arc from each vertex in $X$ to each vertex in $Y$, from each vertex in $Y$ to each vertex in $Z$, and from each vertex in $Z$ to each vertex in $X$. 
       These arcs are illustrated by thin arrows between the nodes $x_1$, $x_2$, $x_{k+1}$, $y_1$, $y_2$, $y_{k+1}$, $z_1$, $z_2$, and $z_{k+1}$ in Figure~\ref{Reduction from Dominating Set (with parameter $k$) to Minimum Indegree Deletion (with parameter $s_v$/$s_a$)}.
 \item There are $k$ arcs from arbitrary vertices in $X \cup Y \cup Z$ to each vertex in $D$ and $k+1$ arcs from arbitrary vertices in $X \cup Y \cup Z$ to each vertex in $V$.
       Outgoing arcs are illustrated by fat lines from the nodes $x_1$, $x_2$, $x_{k+1}$, $y_1$, $y_2$, $y_{k+1}$, $z_1$, $z_2$, and $z_{k+1}$ to the box B1.
       Ingoing arcs are illustrated by two fat lines from the box B1 to the boxes B2 and B3 and by fat arrows from the boxes B2 and B3 to the nodes $v_1$, $v_2$, $v_n$, $d_1$, $d_2$, and $d_n$ in Figure~\ref{Reduction from Dominating Set (with parameter $k$) to Minimum Indegree Deletion (with parameter $s_v$/$s_a$)}.
 \item For each $i \in \{1,\dots,n\}$ there is an arc from $d_i$ and further $k$ outgoing arc from vertices in $X \cup Y \cup Z$ to each vertex in $S_i$ such that every vertex in $X \cup Y \cup Z$ has at least $n$ outneighbors in $\bigcup_{i=1}^n S_i$.
       Removing a vertex in $X \cup Y \cup Z$ decreases the indegree of its $n$ outneighbors in $\bigcup_{i=1}^n S_i$ too much.
       A concrete arc assignment is given in Figure~\ref{Concrete assignement of the arcs}.
       The outgoing arcs from vertices in $D$ are illustrated by thin arrows from the nodes $d_1$, $d_2$, and $d_n$ to the nodes $s_{1,1}$, $s_{1,2}$, $s_{1,6n}$, $s_{2,1}$, $s_{2,2}$, $s_{2,6n}$, $s_{n,1}$, $s_{n,2}$, and $s_{n,6n}$.
       The outgoing arcs from vertices in $X \cup Y \cup Z$ are illustrated by a fat line from B1 to B4, a fat line from B4 to B5, and fat arrows from the box B5 to the nodes $s_{i,1}$, $s_{i,2}$, and $s_{i,6n}$ for $i \in \{1,\dots,n\}$ in Figure~\ref{Reduction from Dominating Set (with parameter $k$) to Minimum Indegree Deletion (with parameter $s_v$/$s_a$)}.
\end{itemize}
The construction ensures that an optimal solution deletes $n-k$ vertices from $V$ and no other vertex.
This will be proved in detail later on, but it is necessary for the further argumentation to investigate the indegrees of the vertices:
The indegree of $w_c$ is $n$ and the indegree of each vertex in $V$ is $k+1$.
Since each vertex $v^*_i$ dominates itself, the vertices in $D$ have indegree at least $k+1$.
Every vertex in $X \cup Y \cup Z$ has trivially also indegree $k+1$.
What remains are the vertex sets $S_1,\dots,S_n$.
The vertices in $\bigcup_{i=1}^n S_i$ have also indegree $k+1$.
In the following we give an intuition why they are useful:
We want to ensure that if $(G,w_c,n-k)$ is a yes-instance of \textsc{Minimum Indegree Deletion}, then every solution set only contains vertices of $V$ and especially no vertices of $D$.
For each vertex~$d_i$ with~$i \in \{1,\dots\,n\}$ there is a set of vertices~$S_i:=\{s_{i,1},\dots\,s_{i,6n}\}$ with an outgoing arc to each of the vertices in $S_i$.
This realizes the ``punishment'': Removing $d_i$ is not possible without removing all vertices in $S_i$ (which are more than $n$).
A similar argumentation holds for the vertices in $X \cup Y \cup Z$, because each of them has at least $n$ outneighbors in $\bigcup_{i=1}^n S_i$.
This finishes the description of the construction.
Now, we give several lemmata and observations to prove the correctness of the construction, that is, $(G^*,k)$ is a yes-instance of \textsc{Dominating Set} if and only if $(G,w_c,n-k)$ is a yes-instance of \textsc{Minimum Indegree Deletion}.

\begin{lemma}
\label{lemma_DS2MID}
 If $(G^*,k)$ is a yes-instance of \textsc{Dominating Set}, then $(G,w_c,n-k)$ is a yes-instance of \textsc{Minimum Indegree Deletion}.
\end{lemma}

\begin{proof}
 Let $(G^*,k)$ be a yes-instance and $V^*_d \subseteq V^*$ be a dominating set of size $k$ for $G^*$.
 Now, we delete a vertex $v_i \in V$ from $G$, if $v^*_i \notin V^*_d$.
 Since $V^*_d$ is of size $k$ and $V$ of size $n$, we deleted $n-k$ vertices.
 The vertex $w_c$ has now indegree $k$.
 We have to show that each other vertex has indegree at least $k+1$.
 By construction every vertex in $G$ has indegree at least $k+1$.
 Since we removed only vertices in $V$, the only vertices which can have a decreased indegree are vertices in $D$.
 (These are the only vertices which can be outneighbors of a vertex in $V$.)
 Due to the fact, that $V^*_d$ is a dominating set, every vertex in $D$ keeps at least one inneighbor in $V$.
 By construction there are $k$ further inneighbors in $X \cup Y \cup Z$.
 Hence, each vertex in $D$ has at least indegree $k+1$.
 Trivially, all other vertices keep indegree at least $k+1$, because we did not remove any inneighbor.
 Thus, $w_c$ has minimum indegree and $(G,w_c,n-k)$ is a yes-instance of \textsc{Minimum Indegree Deletion}.
\end{proof}

Lemma~\ref{lemma_DS2MID} showed the direction from left to right of the parameterized equivalence between the \textsc{Dominating Set} solution and the \textsc{Minimum Indegree Deletion} solution; now, we show the reverse direction.
Consider $(G,w_c,n-k)$ resulting from the parameterized reduction.
Let $M_d$ be a solution set for $(G,w_c,n-k)$.
Several observations are very useful for the further argumentation.

\begin{observation}
\label{obs_MFVS_MFAS_sizes}
 The constructed graph $G$ has a feedback vertex set with at most $k+1$ vertices and a feedback arc set with at most $(k+1)^2$ arcs.
\end{observation}

\begin{proof}
 We show by contradiction that $G - X$ is acyclic.
 Assume that there is a (non-empty) cycle $C=(c_1,\dots,c_l)$.
 Since $w_c$ and each vertex in $\bigcup_{i=1}^n S_i$ has no outneighbor, neither $w_c$ nor any vertex in $\bigcup_{i=1}^n S_i$ can be part of $C$.
 Each vertex in $d_i \in D$ has only outneighbors in $S_i$.
 Thus, $C$ contains no vertex from $D$.
 Vertices from $V$ have only outneighbors in $\{w_c\} \cup D$ which implies that $C$ does not contain any vertex from $V$.
 Vertices from $Z$ have only outneighbors in $\bigcup_{i=1}^n S_i \cup V \cup D$, a set that contains no vertex in $C$.
 The remaining vertices from $Y$ have only outneighbors in $Z \cup V \cup D \bigcup_{i=1}^n S_i$.
 Thus, $C$ must be empty; a contradiction.
\end{proof}


\begin{observation}
\label{obs_MID_atleastk}
 It holds that $w_c$ has indegree at least $k$ in $G - M_d$.
\end{observation}

\begin{proof}
 Assume that $w_c$ has indegree at most $k-1$ in $G - M_d$.
 Since $w_c$ has indegree $n$ in the original graph $G$, we must have deleted $n-k+1$ inneighbors of $w_c$.
 Hence, $M_d$~is a solution set of size at least $n-k+1$; a conflict.
\end{proof}

\begin{observation}
\label{obs_MID_notDXYZ}
 The solution set $M_d$ does not contain vertices of $D$, $X$, $Y$, or $Z$.
\end{observation}

\begin{proof}
 Assume that $u \in M_d$ is a vertex from $D \cup X \cup Y \cup Z$.
 By construction of $G$, vertex~$u$ has at least $n$ outneighbors in $\bigcup_{i=1}^n S_i$.
 We call them $S$-outneighbors in the following.
 After removing $u$, every $S$-outneighbor must have indegree at most $k$ in $G - M_d$, because it has indegree exactly $k+1$ in $G$ (by construction).
 Since the final degree of $w_c$ is at least $k$ (see Observation~\ref{obs_MID_atleastk}) and $w_c$ is the only vertex with minimum indegree in $G - M_d$, every $S$-outneighbor must be also in $M_d$.
 Thus, $M_d$ has size at least $n+1$; a conflict.
\end{proof}

\begin{observation}
\label{obs_MID_exactk}
 In $G - M_d$, the distinguished vertex $w_c$ has indegree exactly $k$ .
\end{observation}

\begin{proof}
 Assume that $w_c$ has an indegree at least $k+1$.
 Since $w_c$ has to be the only vertex with minimum indegree, we have to delete each vertex with an indegree of at most $k+1$.
 Remember that each vertex in $X \cup Y \cup Z$ has an indegree of exactly $k+1$.
 However, no vertex in $X \cup Y \cup Z$ is part of the solution set (see Observation~\ref{obs_MID_notDXYZ}); a conflict.
\end{proof}

\begin{observation}
\label{obs_MID_onlyV}
 It holds that $M_d$ only contains vertices from $V$.
\end{observation}

\begin{proof}
 In the original graph $G$ the distinguished vertex $w_c$ has indegree $n$.
 Due to Observation~\ref{obs_MID_exactk} the solution set $M_d$ must contain at least $n-k$ inneighbors of $w_c$.
 Due to the assumption that $M_d$ has size at most $n-k$ there is no other vertex in $M_d$.
\end{proof}

\begin{lemma}
\label{lemma_MID2DS}
 If $(G,w_c,n-k)$ is a yes-instance of \textsc{Minimum Indegree Deletion}, then $(G^*,k)$ is a yes-instance of \textsc{Dominating Set}.
\end{lemma}

\begin{proof}
 We show that if $M_d$ is a solution set, then there is a dominating set $V^*_d \subseteq V^*$ with at most $k$ vertices.
 One can construct a dominating set $V^*_d$ with $|V^*_d|=k$ as follows:
 For each vertex $v_i \in V \setminus M_d$ add vertex $v^*_i$ to $V^*_d$.
 Due to Observations~\ref{obs_MID_exactk} and \ref{obs_MID_onlyV}, $V^*_d$ has size $k$.
 It remains to show that $V^*_d$ is a domination set.
 Assume that there is a vertex $v^*_f$ that is not dominated by any vertex in $V^*_d$.
 Thus, neither $v_f$ nor any $v_d \in N(v_f)$ is in $M_d$.
 Due to Observation~\ref{obs_MID_notDXYZ}, no vertex in $D$ was deleted.
 Due to the construction of $G$ there is no arc from any vertex in $V$ to the undominated vertex $d_f$.
 Hence, $d_f$ has indegree of $k$ and $w_c$ is not the only vertex with minimum indegree; a contradiction.
\end{proof}

Putting all together, we arrive at the following theorem:
\begin{theorem}
\label{MID_w2}
 \textsc{Minimum Indegree Deletion} is W[2]-hard with respect to the parameter $s_v$ as well as with respect to the parameter $s_a$.
\end{theorem}

\begin{proof}
We show that the transformation from $(G^*,k)$ to $(G,w_c,n-k)$ is a parameterized reduction:
By construction, the transformation can be executed in $f(k) \cdot \poly(|x|)$ time with $f(k)$ being a function only depending on $k$ and $|x|$ being the size of the input.
Due to Observation~\ref{obs_MFVS_MFAS_sizes} the parameter $s_v$ ($s_a$) is bounded by a function only depending on $k$.
The equivalence follows Lemma~\ref{lemma_DS2MID} and Lemma~\ref{lemma_MID2DS}.
\end{proof}

Theorem~\ref{MID_w2} provides a relative lower bound for the parameterized complexity with respect to the feedback set parameters $s_v$ and $s_a$.
An upper bound, namely the membership in $\XP$, is a corollary of the main result of the next section.

\section{Feedback vertex set size and solution size as combined parameter}
\label{Feedback vertex set size and solution size as combined parameter}
Taking $s_v$:=``size of a feedback vertex set'' ($s_a$:=``size of a feedback arc set'') as parameter \textsc{Minimum Indegree Deletion} does not yield fixed-parameter-tractability.
As mentioned in Section~\ref{Known results for Minimum Indegree Deletion}, $k$:=``number of vertices to delete'' as a parameter also leads to W[2]-hardness.
So, it makes sense to consider the combined parameters ($s_v$,$k$) as well as ($s_a$,$k$) for \textsc{Minimum Indegree Deletion}.
In this section, we show that each of these combined parameters leads to a fixed-parameter algorithm.
The algorithm described in the following gets as input a directed graph, a distinguished vertex and a positive integer $k$ and outputs a solution set of size $k$.
Hence, with an additional factor $k$ to the running time it also solves the corresponding minimization problem, where one asks for a minimum-size solution.
The algorithm does not need to know or compute $s_v$ or even a feedback vertex set.

We start with looking at the acyclic special case to describe the basic idea of the algorithm.
One motivation to investigate \textsc{Minimum Indegree Deletion} with respect to the parameters that measure the ``degree of acyclicity'' is that the problem is $\NP$-hard in general but polynomial-time solvable in the acyclic special case~\cite{BU09}. 
The corresponding algorithm is based on the fact that a directed acyclic graph always contains a vertex with indegree zero.
It is an exhaustive application of the following step:
If there is a vertex ($\neq w_c$) with indegree zero, then remove it.
This is trivially correct, since we want that $w_c$ is the only vertex with minimum indegree (zero).

There are two reasons why we cannot apply this algorithm to general directed graphs:
\begin{enumerate}
 \item Adapting the idea directly by ``removing every indegree-zero vertex except $w_c$'' is not correct, because it is possible that there is no vertex with indegree zero.
 \item Adapting the idea in a more generalized way as ``removing every vertex with minimum indegree except $w_c$'' is not correct.
       Unfortunately, it is even possible that we do not need to remove all or even any vertex that has minimum indegree in the input graph.
\end{enumerate}
The first point is trivial.
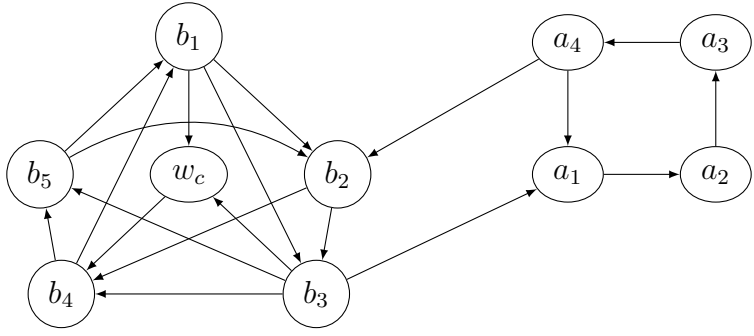
\begin{figure*}
\begin{center}
 \begin{tikzpicture}[>=latex,join=bevel,scale=1.0]
  \node (w_c) [draw,transform shape,ellipse] {$w_c$};

  \node (b_1) [above=of w_c] [draw,transform shape,ellipse] {$b_1$};
  \node (b_2) [right=of w_c] [draw,transform shape,ellipse] {$b_2$};
  \node (b_3) [below right=of w_c] [draw,transform shape,ellipse] {$b_3$};
  \node (b_4) [below left=of w_c] [draw,transform shape,ellipse] {$b_4$};
  \node (b_5) [left=of w_c] [draw,transform shape,ellipse] {$b_5$};

  \node (a_1) [right=4 of w_c] [draw,transform shape,ellipse] {$a_1$};
  \node (a_2) [right=of a_1] [draw,transform shape,ellipse] {$a_2$};
  \node (a_3) [above=of a_2] [draw,transform shape,ellipse] {$a_3$};
  \node (a_4) [above=of a_1] [draw,transform shape,ellipse] {$a_4$};

  \draw [->] (b_1) to (w_c);
  \draw [->] (b_3) to (w_c);
  \draw [->] (w_c) to (b_4);

  \draw [->] (b_1) to (b_2);
  \draw [->] (b_1) to (b_3);
  \draw [->] (b_2) to (b_3);
  \draw [->] (b_2) to (b_4);
  \draw [->] (b_3) to (b_4);
  \draw [->] (b_3) to (b_5);
  \draw [->] (b_4) to (b_5);
  \draw [->] (b_4) to (b_1);
  \draw [->] (b_5) to (b_1);
  \draw [->] (b_5) to [bend left=30] (b_2);

  \draw [->] (a_1) to (a_2);
  \draw [->] (a_2) to (a_3);
  \draw [->] (a_3) to (a_4);
  \draw [->] (a_4) to (a_1);

  \draw [->] (a_4) to (b_2);
  \draw [->] (b_3) to (a_1);

 \end{tikzpicture}
\end{center}
\caption{Directed cyclic graph where a minimum-size solution set of \textsc{Minimum Indegree Deletion} consists of $b_1$ and $b_3$ whereas $a_2, a_3$, and $a_4$ have minimum indegree of $1$.}
\label{Example of MID with a directed cyclic graph.}
\end{figure*}
The second point is illustrated in Figure~\ref{Example of MID with a directed cyclic graph.}.
As we see in this example another strategy to make $w_c$ the only vertex with minimum indegree (besides removing vertices with smaller indegree) is to remove inneighbors of $w_c$ to directly decrease its indegree.
This leads to the algorithm \texttt{MinimumIndegreeDeletion} in Figure~\ref{Fixed-parameter algorithm that solves that solves Minimum Indegree Deletion with respect to the parameter (s_v,k).}:
\begin{figure}
\begin{algorithmic}[1]
\Procedure{\texttt{MinimumIndegreeDeletion}}{}
\For{each $i:=0$ to $s_v$}
  \If{$|N(w_c)| \leq i+k$}
    \For{each size-$i$ subset $U \subseteq N(w_c)$}
      \State Remove $D:=N(w_c) \setminus U$ from $G$.
      \State $M_d:=D$
      \While{there is a vertex $d \neq w_c$ with indegree at most $i$}
        \State Remove $d$ from $G$.
        \State $M_d:=M_d \cup \{d\}$
      \EndWhile
      \If{$|M_d| \leq k$}
        \State \Return $M_d$
      \EndIf
    \EndFor
  \EndIf
\EndFor
\State \Return ``no'';
\EndProcedure
\end{algorithmic}
\caption{Fixed-parameter algorithm that solves \textsc{Minimum Indegree Deletion} with respect to the parameter $(s_v,k)$. The variable $i$ represents the final indegree of $w_c$ in the solution graph.}
\label{Fixed-parameter algorithm that solves that solves Minimum Indegree Deletion with respect to the parameter (s_v,k).} 
\end{figure}
To analyze the algorithm we take a closer look at our directed graph instance.
First, we see that the minimum indegree in the solution graph is bounded by the size of any feedback vertex set.
\begin{lemma}
\label{lemma_min_indegree_leq_MFVS}
 Let $G=(V,E)$ denote a directed graph.
 Let $d_{min}$ denote the minimum indegree of the vertices in $V$ and let $V_f \subseteq V$ be a feedback vertex set.
 Then, $|V_f| \geq d_{min}$.
\end{lemma}
\begin{proof}
 There exists a vertex $v_z$ with indegree zero in $G - V_f$, because $G - V_f$ is acyclic.
 Since $G$ has no vertex with indegree smaller than $d_{min}$, the feedback vertex set $V_f$ must contain at least $d_{min}$ inneighbors of $v_z$.
 Thus, $|V_f| \geq d_{min}$.
\end{proof}
A consequence is that the final indegree of $w_c$ in a solution graph is at most the size of a feedback vertex set $s_v$.
This holds due to Lemma~\ref{lemma_min_indegree_leq_MFVS} and the fact that a feedback vertex set in $G$ implies a feedback vertex set in $G - M_d$ of at most the same size.
Hence, the loop in line 2 of the algorithm \texttt{MinimumIndegreeDeletion} is defined correctly.
The main part (lines 4-14) of \texttt{MinimumIndegreeDeletion} is exhaustive exploration of the search space.
It remains to show that the condition in line 3 is correct.
Assume towards a contradiction that $|N(w_c)| \geq k+i$ in a yes-instance of \textsc{Minimum Indegree Deletion}.
Removing all but $i$ neighbors implies a solution set of size greater than $k$; a contradiction.
Hence, the algorithm is correct and we arrive at the following theorem:
\begin{theorem}
\label{theorem_combinedFPT}
 \textsc{Minimum Indegree Deletion} with respect to the combined parameter $(s_v, k)$, with $s_v$ being the size of a feedback vertex set and $k$ being the size of a solution set, is solvable in $O(s_v \cdot (k+1)^{s_v} \cdot n^2)$ time.
\end{theorem}

\begin{proof}
 The correctness of the algorithm was already shown.
 It remains to analyze the running time.
 In the worst case, we have to start at most $s_v$ times with the first (out-most) loop.
 In the second loop, we try at most $\binom{s_v+k}{k}$ subsets.
 Thus, we have
 \begin{displaymath}
  \binom{s_v+k}{k} = \frac{(s_v+k)!}{k! \cdot (s_v+k-k)!} = \frac{\prod_{i=1}^{s_v} k+i}{s_v!} \leq \left(\frac{k+1}{1}\right)^{s_v}
 \end{displaymath}
 subsets.
 The third loop can be done in $O(n^2)$ time.
\end{proof}


Looking at the second loop of the algorithm helps us to see that \textsc{Minimum Indegree Deletion} with respect to the single parameter $s_v$ is actually at least in $\XP$.
Branching into all possible subsets of $i \leq s_v$ inneighbors of $w_c$ is of course possible in $O(n^i)$ time.
In a more formal way, we arrive at the following corollary:
\begin{corollary}
 The problem \textsc{Minimum Indegree Deletion} with respect to the parameter $s_v:=$``size of a feedback vertex set'' is in $\XP$.
\end{corollary}

\begin{proof}
 Since $(k+1)$, $n$, and $s_v$ are upper bounded by the input size, this follows directly from Theorem~\ref{theorem_combinedFPT}.
 The running time of \texttt{MinimumIndegreeDeletion} is bounded by $O(|x|^{s_v+3})$ with $|x|$ being the input size.
\end{proof}

\chapter{Minimum Degree Deletion}
\label{Minimum Degree Deletion}
In this section, we investigate the parameterized complexity of \textsc{Minimum Degree Deletion} which can be considered as undirected variant of \textsc{Minimum Indegree Deletion}.
It models electoral control by removing candidates of a special voting rule, see Section~\ref{MDD_intro}.
In Chapter~\ref{Minimum Indegree Deletion}, we showed fixed-parameter intractability of \textsc{Minimum Indegree Deletion} with respect to the parameters $s_v:=$``size of a feedback vertex set'' and $s_a:=$``size of a feedback arc set'', respectively.
In contrast, we show in this chapter fixed-parameter tractability even for the stronger parameter ``treewidth of the input graph''.
Furthermore, we differentiate the parameters by comparing their kernel sizes.
Firstly, we show that \textsc{Minimum Degree Deletion} is as well as its directed variant \textsc{Minimum Indegree Deletion} fixed-parameter intractable with respect to the parameter $k:=$``number of vertices to delete''.
An overview about the kernel sizes and the parameterized complexity with respect to the (combined) parameters is given in Figure~\ref{MDD results.}.
\begin{figure}[!b]
\begin{center}
Parameterized complexity:\\
\hspace{0.7cm}
 \begin{tabular}{ r | c c c }
  ~ & parameter description & complexity \\
  \hline
  $t_w$ & treewidth of the input graph & \textbf{FPT} \\
  $s_v$ & size of a feedback vertex set & \textbf{FPT} \\
  $s_a$ & size of a feedback arc set & \textbf{FPT} \\
  $k$ & size of a solution set & \textbf{W[1]-hard, in XP} \\
  $d$ & maximum degree of a vertex & FPT \\
 \end{tabular}\\
\end{center}
\begin{center}
Kernel sizes:\\
\hspace{0.7cm}
 \begin{tabular}{ r | c c }
  ~ & single parameter & combined with $k$ \\
  \hline
  $t_w$ & \textbf{no polynomial} & \textbf{no polynomial} \\
  $s_v$ & \textbf{no polynomial} & \textbf{no polynomial} \\
  $s_a$ & \textbf{vertex-linear} & \textbf{vertex-linear} \\
  $k$ & \textbf{no kernel} & ~ \\
  $d$ & open & open \\
 \end{tabular}\\
\end{center}
 \caption{Overview of the parameterized complexity of \textsc{Minimum Degree Deletion} and the corresponding kernel sizes.
         New results are in boldface.
         The results for the parameter $d:=$``maximum degree of a vertex'' can be directly transferred from the results for the directed variant in \cite{BU09}.
         A vertex-linear kernel is a problem kernel whose size is linear in the number of vertices.}
 \label{MDD results.}
\end{figure}

\section{Solution size as parameter}
\label{Solution size as parameter for Minimum Degree Deletion}
In this section, we analyze the parameterized complexity of \textsc{Minimum Degree Deletion} with respect to the parameter $k:=$``number of vertices to delete''.
Using methods similar to those in Section~\ref{Feedback vertex/arc set size as parameter} we show W[1]-hardness by presenting a parameterized reduction from \textsc{Independent Set} with respect to the parameter ``independent set size''.
Below, we give a description and illustration of the reduction.

We first show that one can assume that each \textsc{Independent Set} instance has an even number of edges.
Let $G=(V,E)$ be an undirected graph with $n:=|V|$ and $m:=|E|$.
One can transform each instance with an odd value of $m$ to a new instance with an even value of $m$ such that the new instance is a yes-instance if and only if the original instance is a yes-instance and the parameter value does not change.
To this end, we have to consider two cases:
The first case is that the number of vertices is odd.
Then we only have to add a new vertex and connect it to each vertex of the original graph.
The second case is that the number of vertices is even.
Here, we have to add a clique of three new vertices and connect each of these vertices of each vertex of the original graph.
So, we get $3n + 3$ new edges which is odd since $n$ is even.
Trivially, none of the new vertices will ever be a part of an independent set of size at least two\footnote{Every graph with at least one vertex has an independent set of size one.} in both cases.

\begin{figure*}
\begin{tikzpicture}[>=latex,join=bevel,scale=0.9]
  \node (w_c) [draw,ellipse] {$w_c$};

  \node (v_1) [below right=2 of w_c] [draw,transform shape,ellipse] {$v_1$};
  \node (v_1_bl) [below left=0.2 of v_1] {};
  \node (v_2) [right=of v_1] [draw,transform shape,ellipse] {$v_2$};
  \node (v_dots) [right=of v_2] [transform shape] {$\dots$};
  \node (v_n) [right=of v_dots] [draw,transform shape,ellipse] {$v_n$};

  \node (C_v_1) [below=1.5 of v_1_bl] [draw,transform shape,star,star points=7,star point ratio=0.8,thick] {$C_{v_1}$};
  \node (C_v_2) [right=of C_v_1] [draw,transform shape,star,star points=7,star point ratio=0.8,thick] {$C_{v_2}$};
  \node (C_v_dots) [right=of C_v_2] [transform shape,thick] {$\dots$};
  \node (C_v_n) [right=of C_v_dots] [draw,transform shape,star,star points=7,star point ratio=0.8,thick] {$C_{v_n}$};

  \node (C'_v_1) [below=of C_v_1] [draw,transform shape,star,star points=7,star point ratio=0.8,thick] {$C_{1}$};
  \node (C'_v_2) [below=of C_v_2] [draw,transform shape,star,star points=7,star point ratio=0.8,thick] {$C_{2}$};
  \node (C'_v_dots) [right=of C'_v_2] [transform shape,thick] {$\dots$};
  \node (C'_v_n) [below=of C_v_n] [draw,transform shape,star,star points=7,star point ratio=0.8,thick] {$C_{n}$};

  \node (e_1) [above=of v_1] [draw,transform shape,ellipse] {$e_1$};
  \node (e_1_bl) [below left=0.2 of e_1] {};
  \node (e_2) [right=of e_1] [draw,transform shape,ellipse] {$e_2$};
  \node (e_dots) [right=of e_2] [transform shape] {$\dots$};
  \node (e_m) [right=of e_dots] [draw,transform shape,ellipse] {$e_m$};

  \node (C_e_1) [above=2.5 of e_1_bl] [draw,transform shape,star,star points=5,star point ratio=0.8,thick] {$C_{e_1}$};
  \node (C_e_2) [right=of C_e_1] [draw,transform shape,star,star points=5,star point ratio=0.8,thick] {$C_{e_2}$};
  \node (C_e_dots) [right=of C_e_2] [transform shape,thick] {$\dots$};
  \node (C_e_m) [right=of C_e_dots] [draw,transform shape,star,star points=5,star point ratio=0.8,thick] {$C_{e_m}$};


  \node (C_Expl2) [right=0.5 of C_e_m] [draw,transform shape,text width=5cm,rectangle,dotted,text justified]
                  {Exactly $n-k$ vertices in $C_x$ with $x \in \{v_1,\dots,v_n\} \cup \{e_1,\dots,e_m\}$ are connected to the vertex $x$.};
  \node (C_Expl3) [above=of C_e_dots] [draw,transform shape,text width=15cm,rectangle,dotted,text justified]
                  {Each vertex from $C_{e_j}$ with $j \in \{1,3,5,\dots,m-1\}$ is connected to exactly one vertex from $C_{e_{j+1}}$.};
  \node (C_Expl4) [below=1.5 of C_Expl2] [draw,transform shape,text width=5cm,rectangle,dotted,text justified]
                  {There is an edge between $e_x$ and $v_y$ if and only if $v^*_y$ is incident to $e^*_x$};
  \node (VE_Expl) [right=0.5 of C_v_n] [draw,transform shape,text width=5cm,rectangle,dotted,text justified]
                  {Each vertex from $C_{v_i}$ with $i \in \{1,2,\dots,n\}$ is connected to exactly one vertex from $C_i$.};

  \draw [-] (v_1) to [out=150,in=290] (w_c);
  \draw [-] (v_2) .. controls +(195:3.5) .. (w_c);
  \draw [-] (v_n) .. controls +(190:8) .. (w_c);

  \draw [-,dotted] (v_1) to (e_1);
  \draw [-,dotted] (v_2) to (e_1);
  \draw [-,dotted] (v_1) to (e_2);
  \draw [-,dotted] (v_n) to (e_2);

  \draw [-,thick] (C_v_1) --node[below left] {$n-k$} (v_1);
  \draw [-,thick] (C_v_2) --node[below left] {$n-k$} (v_2);
  \draw [-,thick] (C_v_n) --node[below left] {$n-k$} (v_n);


  \draw [-,thick] (C_v_1) --node[left] {$n-k+1$} (C'_v_1);
  \draw [-,thick] (C_v_2) --node[left] {$n-k+1$} (C'_v_2);
  \draw [-,thick] (C_v_n) --node[left] {$n-k+1$} (C'_v_n);

  \draw [-,thick] (C_e_1) --node[left] {$n-k$} (e_1);
  \draw [-,thick] (C_e_2) --node[left] {$n-k$} (e_2);
  \draw [-,thick] (C_e_m) --node[left] {$n-k$} (e_m);

  \draw [-,thick] (C_e_1) to [bend left] node[above] {$n-k$} (C_e_2);

 \end{tikzpicture}
 \caption{\textsc{Minimum Degree Deletion} instance obtained from a parameterized reduction from an \textsc{Independent Set} with an even number of edges.
          Each star $C_x$ with $x \in \{v_1,\dots,v_n\}$ represents a clique with $n-k$ vertices.
          Each star $C_x$ with $x \in \{1,\dots,n\} \cup \{e_1,\dots,e_m\}$ represents a clique with $n-k+1$ vertices.}
 \label{Reduction from Independent Set (with parameter solution size) to Minimum Degree Deletion (with parameter solution size)}
\end{figure*}

\paragraph{Parameterized reduction.}
The following reduction is illustrated in Figure~\ref{Reduction from Independent Set (with parameter solution size) to Minimum Degree Deletion (with parameter solution size)}.
Given an \textsc{Independent Set} instance $(G^*=(V^*,E^*),k)$, with $V^*=\{v^*_1,v^*_2,\dots,v^*_n\}$ and $E^*=\{e^*_1,e^*_2,\dots,e^*_m\}$,
we construct an undirected graph $G$, with a distinguished vertex~$w_c$ such that $(G,w_c,k)$ is a yes-instance of \textsc{Minimum Degree Deletion} if and only if $(G^*,k)$ is a yes-instance of \textsc{Independent Set}.
The vertex set of $G$ consists of $w_c$ and the union of the following disjoint vertex sets:
\begin{itemize}
 \item $V$ containing one vertex for every vertex in $V^*$ and $E$ containing one vertex for every edge in $E^*$.
  \begin{itemize}
   \item $V:=\{v_i \mid i\in\{1,\dots,n\}\}$ represents the set of vertices which is illustrated by the nodes $v_1$, $v_2$ and $v_n$ in Figure~\ref{Reduction from Independent Set (with parameter solution size) to Minimum Degree Deletion (with parameter solution size)}.
   \item $E:=\{e_i \mid i\in\{1,\dots,m\}\}$ represents the set of ``connections between vertices'' which is illustrated by the nodes $e_1$, $e_2$, and $e_m$ in Figure~\ref{Reduction from Independent Set (with parameter solution size) to Minimum Degree Deletion (with parameter solution size)}.
  \end{itemize}
 \item For each vertex $v_i \in V$ there are two cliques $C_{v_i}$ and $C_{i}$ of size $n-k+1$, and
       for each vertex $e_i \in E$ there is one clique $C_{e_i}$ of size $n-k$.
       These cliques are illustrated by stars in Figure~\ref{Reduction from Independent Set (with parameter solution size) to Minimum Degree Deletion (with parameter solution size)}.
\end{itemize}
The edges (besides inner-clique edges) of $G$ are drawn as follows:
\begin{itemize}
 \item There is an edge between $v_i$ and $e_j$ if and only if $v^*_i$ is incident to $e^*_j$.
       These edges are illustrated by dotted lines in Figure~\ref{Reduction from Independent Set (with parameter solution size) to Minimum Degree Deletion (with parameter solution size)}.
 \item The distinguished vertex $w_c$ is connected to each vertex in $V$.
       The corresponding edges are illustrated by thin lines between the node $w_c$ and the nodes $v_1$, $v_2$, and $v_n$ in Figure~\ref{Reduction from Independent Set (with parameter solution size) to Minimum Degree Deletion (with parameter solution size)}.
 \item Each vertex of every clique $C_{e_j}$ with $j \in \{1,3,\dots,m-1\}$ is connected to exactly one (different) vertex in $C_{e_{j+1}}$.
       (Note that $m-1$ is odd.)
       The corresponding edges are illustrated by a fat line between $C_{e_1}$ and $C_{e_2}$ which is labeled with $n-k$ in Figure~\ref{Reduction from Independent Set (with parameter solution size) to Minimum Degree Deletion (with parameter solution size)}.
 \item Exactly $n-k$ vertices of each clique $C_{x}$  with $x \in \{v_1,\dots,v_n,e_1,\dots,e_m\}$ are connected to the vertex $x$.
       The corresponding edges are illustrated by the remaining fat lines which are labeled with $n-k$ in Figure~\ref{Reduction from Independent Set (with parameter solution size) to Minimum Degree Deletion (with parameter solution size)}.
 \item Each vertex of every clique $C_{v_i}$ with $i \in \{1,\dots,n\}$ is connected to exactly one (different) vertex in $C_i$.
       The corresponding edges are illustrated by fat lines which are labeled with $n-k+1$ in Figure~\ref{Reduction from Independent Set (with parameter solution size) to Minimum Degree Deletion (with parameter solution size)}.
\end{itemize}
Our goal is to ensure that an optimal solution deletes $k$ vertices from $V$ and no other vertex.
Therefore, we set the degree of $w_c$ to $n$ and for each other vertex to at least $n-k+1$.
Each vertex in $C_i$ with $i \in \{1,\dots,n\}$ and $C_{e_j}$ with $j \in \{1,\dots,m\}$ has degree $n-k+1$ and each vertex in $C_{v_i}$ with $i \in \{1,\dots,n\}$ has degree $n-k+1$ or $n-k+2$.
Each vertex $e_j$ with $j \in \{1,\dots,m\}$ has degree $n-k+2$ and each vertex $v_i$ with $i \in \{1,\dots,n\}$ has degree at least $n-k+1$.
This finishes the description of the construction.
Now, we give several lemmata and observations to prove the correctness of the construction, that is, $(G^*,k)$ is a yes-instance of \textsc{Independent Set} if and only if $(G,w_c,k)$ is a yes-instance of \textsc{Minimum Degree Deletion}.

\begin{lemma}
\label{lemma_IS2MDD}
 If $(G^*,k)$ is a yes-instance of \textsc{Independent Set}, then $(G,w_c,k)$ is a yes-instance of \textsc{Minimum Degree Deletion}.
\end{lemma}

\begin{proof}
 Let $(G^*,k)$ be a yes-instance and $V^*_d \subseteq V^*$ a size-$k$ set of independent vertices.
 We delete each vertex $v_i \in V$ from $G$, if $v^*_i \in V^*_d$.
 Since $|V^*_d| = k$ and $|V| = n$, we deleted $k$ vertices and the vertex $w_c$ has now degree of $n-k$.
 It holds that each vertex $e_j$ with $j \in \{1,\dots,m\}$ has degree at least $n-k+1$, because $V^*_d$ is an independent set which means that at most one of the neighbors from $e_j$ in $V$ is deleted.
 Trivially, the degree of all other vertices is at least $n-k+1$.
 Thus, $w_c$ has minimum degree and $(G,w_c,k) \in$ \textsc{Minimum Degree Deletion}.
\end{proof}

Next, we will show the opposite direction of the equivalence.
To this end, several observations are very useful for further argumentation.
In the following, let $G^*=(V^*,E^*)$ be an undirected graph and $k$ a positive integer.
We apply our reduction and denote the resulting graph and its vertices and vertex sets as described above.
Let $M_d$ denote a set of $k$ vertices of $G$ such that its deletion makes $w_c$ become the only vertex with minimum degree.

\begin{observation}
\label{obs_MDD_exactly_n-k}
 It holds that $w_c$ has degree exactly $n-k$ after deleting $M_d$ from $G$.
\end{observation}

\begin{proof}
 Assume towards a contradiction that the degree of $w_c$ does not equal $n-k$.
 First consider the case that $w_c$ has degree less than $n-k$ after deleting $M_d$ from $G$.
 In the original graph $w_c$ has degree $n$.
 To reach degree less than $n-k$ we have to delete more than $k$ vertices; a contradiction.
 Second consider the case that $w_c$ has degree more than $n-k$ after deleting $M_d$ from $G$.
 This means that the degree of $w_c$ is at least $n-k+1$.
 Since $M_d$ is a solution set, no other vertex with degree $n-k+1$ can exist after deleting $M_d$ from $G$.
 So, we have to delete more than $n>k$ vertices, since there are more than $n$ cliques with vertices with degree exactly $n-k+1$; a contradiction.
 Altogether, $w_c$ has degree exactly $n-k$ after deleting $M_d$ from $G$.
\end{proof}

\begin{observation}
\label{obs_MDD_onlyV}
 It holds that $M_d$ only contains vertices of $v_i$ with $i \in \{1,\dots,n\}$.
\end{observation}

\begin{proof}
 The argumentation is very simple.
 Due to Observation~\ref{obs_MDD_exactly_n-k} it is clear that $M_d$ contains at least $k$ vertices $v_i$ with $i \in \{1,\dots,n\}$.
 Since $M_d$ has size $k$, there is no other vertex in $M_d$.
\end{proof}

\begin{lemma}
\label{lemma_MDD2IS}
 If $(G,w_c,k)$ is a yes-instance of \textsc{Minimum Degree Deletion}, then $(G^*,k)$ is a yes-instance of \textsc{Independent Set}.
\end{lemma}

\begin{proof}
 Let $(G,w_c,k)$ be a yes-instance of \textsc{Minimum Degree Deletion}.
 We can build a size-$k$ independent set $V^*_d$ as follows:
 For each vertex $v_i \in M_d$ add vertex $v^*_i$ to $V^*_d$.
 Due to Observations~\ref{obs_MDD_exactly_n-k} and \ref{obs_MDD_onlyV}, $V^*_d$ has size $k$.
 It remains to be shown that $V^*_d$ is an independent set.
 Assume that there is an edge $e^*_j$ with $j \in \{1,\dots,m\}$ that is incident to any two vertices $v_a,v_b \in V^*_d$.
 Thus, both $v_a$ and $v_b$ are in $M_d$.
 Due to the construction of $G$ the vertex $e_j$ has degree of $n-k$ after deleting $v_a$ and $v_b$ and must be deleted, too; a contradiction to Observation~\ref{obs_MDD_onlyV}.
\end{proof}

Altogether, we arrive at the following:
\begin{theorem}
\label{MDD_w2}
 \textsc{Minimum Degree Deletion} is W[1]-hard with respect to the parameter $k:=$``number of vertices to delete''.
\end{theorem}

\begin{proof}
 We show that the transformation from $(G^*,k)$ to $(G,w_c,k)$ is a parameterized reduction:
 By construction, it can be executed in $f(k) \cdot \poly(|x,k|)$ time.
 The new parameter equals the original one. 
 The equivalence follows Lemma~\ref{lemma_IS2MDD} and Lemma~\ref{lemma_MDD2IS}.
\end{proof}

Theorem~\ref{MDD_w2} provides a relative lower bound for the parameterized complexity with respect to the parameter $k$.
An upper bound, namely the membership in $\XP$, is quite easy to see.
Simply checking for each $V' \subseteq V$ whether $w_c$ is the only vertex with minimum degree in $G - V'$ already leads to the following:

\begin{proposition}
 \textsc{Minimum Degree Deletion} is in $\XP$ with respect to the parameter $k:=$``number of vertices to delete''.
\end{proposition}

\section{Size of a feedback edge set as parameter}
\label{Size of a Feedback Edge Set as parameter for Minimum Degree Deletion}
In Section~\ref{Feedback vertex set size and solution size as combined parameter}, we showed that \textsc{Minimum Indegree Deletion} is fixed-parameter-tractable with respect to the combined parameter ``size of a feedback vertex/arc set'' and ``size of a solution set''.
In contrast, it is W[2]-hard for both single parameters (Section~\ref{Feedback vertex/arc set size as parameter} and~\cite{BU09}).
It is easy to adapt the algorithm of Section~\ref{Feedback vertex set size and solution size as combined parameter} to the undirected version \textsc{Minimum Degree Deletion}.
Hence, we can show fixed-parameter tractability for the combined parameter ``size of a feedback vertex/edge set'' and ``size of a solution set''.
However, the parameterized intractability for the single parameters ``size of a feedback vertex set'' and ``size of a feedback arc/edge set'' cannot be transferred as easily.

In contrast to the hardness results of the directed problem we will show fixed-parameter tractability for both parameters.
We will start with a data reduction rule that achieves a \textit{vertex-linear kernel}, that is, a problem kernel whose size is linear in the number of vertices, for \textsc{Minimum Indegree Deletion} with respect to the parameter ``size of a feedback edge set''.

\paragraph{A vertex-linear kernel.}
Our kernelization is based on a simple data reduction rule.
The following lemma ensures polynomial running time for this rule.

\begin{lemma}
\label{lemma_mddWithD01}
 Let $G=(V,E)$ be an undirected graph and $k$ a positive integer.
 In $O(n^2 \cdot k)$ time and $O(n^2 + n)$ space one can determine whether there is a set of vertices $M_d \subseteq V$ with $| M_d | \leq k$ such that $w_c$ is the only vertex with minimum degree in $G - M_d$ and $\deg(w_c) \leq 1$.
\end{lemma}

\begin{proof}
 Determining whether there is a solution set $M_d \subseteq V$ with $| M_d | \leq k$ such that $w_c$ is the only vertex with minimum degree and $\deg(w_c) \leq 1$ works as follows:
 To manage the degree information of the vertices we use an adjacency matrix and store the sums of each columns and each row.
 The column and row sums give us directly the degree of each vertex.
 The initialization of the matrix costs $O(n^2)$ time and $O(n^2 + n)$ space.
 In a first step we remove all neighbors of $w_c$ if $\deg(w_c)=0$ and all but one neighbor if $\deg(w_c)=1$.
 Subsequently, we remove all vertices with degree at most $\deg(w_c)$.
 Removing one vertex costs $O(n)$ time, because we need to update the matrix.
 There are at most $k$ removal steps.
 If $\deg(w_c)=1$ there are up to $n$ possible neighbor which are not removed.
 This means an additional factor of $O(n)$ in the case $\deg(w_c)=1$.
 Thus, we need $O(n \cdot (k \cdot n ) ) = O( n ^ 2 \cdot k )$ time and $O(n ^ 2 + n)$ space.
\end{proof}

\begin{reductionrule}[Remove Low Degree]
\label{Remove Low Degree}
 Let $G=(V,E)$ be an undirected graph and $k$ be a positive integer.
 We denote by $\RemoveLowDegree(G)$ the graph resulting by the following data reduction:
 If there is set of vertices $M_d \subseteq V$ with $| M_d | \leq k$ such that $w_c$ is the only vertex with minimum degree and $\deg(w_c) \leq 1$ in $G - M_d$, then replace $G$ with a new graph which only contains the single vertex $w_c$ and set the parameter to zero.
 Otherwise, $w_c$ has degree at least two in every optimal solution, iteratively remove each vertex with degree at most two and decrease the parameter by one in each removal step.
\end{reductionrule}

Due to Lemma~\ref{lemma_mddWithD01}, Reduction Rule ``Remove Low Degree'' can be executed in polynomial time.
It is easy to verify that the rule is sound.
The next observation follows directly from the construction of $\RemoveLowDegree(G)$.

\begin{observation}
 \label{lemma_mddDegreeGT2}
 Every vertex ($\neq w_c$) in $\RemoveLowDegree(G)$ has degree at least three.
\end{observation}

Now, we are ready to bound the number of vertices in $\RemoveLowDegree(G)$ with the help of this observation.
For the ease of argumentation we firstly define a transformation which, given a forest, removes every inner vertex with degree two by connecting its neighbors.

\begin{definition}
 Let $G$ denote a forest.
 The function $\DisLined(G)$ denotes the result of the exhaustive application of the following procedure:
 If there is an induced $P_3$ (path of length three) such that the middle vertex $v'$ in $P_3$ has degree two in $G$, then remove $v'$ and connect both its neighbors.
\end{definition}

\begin{theorem}
\label{MDD_linearkernel}
 There is a $2s_e$-vertex kernel for \textsc{Minimum Degree Deletion} with respect to the parameter $s_e:=$``size of a feedback edge set''.
 It is computable in $O(n^2 \cdot k)$ time with $n$ being the number of vertices and $k$ being the number of vertices to delete.
\end{theorem}

\begin{proof}
 Let $G$ denote an undirected graph.
 We show that there are at most $2 \cdot s_e$ vertices with $s_e$ being the size of a feedback edge set in $\RemoveLowDegree(G)$.
 Let $E_d$ denote a feedback edge set of size $s_e$.
 Clearly, $G - E_d$ is a forest.
 Since each vertex in $G$ has degree at least three (Observation~\ref{lemma_mddDegreeGT2}), each leaf in $G - E_d$ must be incident to at least two edges in $E_d$.
 It holds that $G - E_d$ contains $l \leq s_e$ leaves, because each leaf must be incident to two edges of the feedback edge set and each edge of the feedback edge set can be incident to at most two leaves.
 Furthermore, the sum of incidences of the edges in $E_d$ is $2s_e$.
 Each inner vertex of degree two in $G - E_d$ must be incident to at least one edge in $E_d$.
 Since there are $l$ leaves in $G - E_d$, only $2s_e - 2l$ incidences are left over.
 Hence, $G - E_d$ contains at most $2s_e - 2l$ inner vertices with degree two.
 It holds that $G - E_d$ has at most $l$ inner vertices with degree at least three if $\DisLined(G - E_d)$ has at most $l$ inner vertices with degree at least three.
 Even if $\DisLined(G - E_d)$ is a complete binary tree there are at most $l/2 + l/4 + \dots + 1 = l$ inner vertices.
 Altogether, $G$ has at most $l + 2s_e - 2l + l = 2s_e$ vertices.
 The running time follows Lemma~\ref{lemma_mddWithD01}.
\end{proof}

\paragraph{A search tree algorithm.}
Of course, the $2s_e$-vertex-kernel already provides fixed-parameter tractability.
The running time of a corresponding brute-force algorithm is $O(4^{s_e} \cdot n^2)$ with $s_e$ being the size of a feedback edge set.
The polynomial factor is for checking the correctness of the ``guessed'' solution set.
A simple refinement helps to give the search tree algorithm \texttt{MDD-search} (see Figure~\ref{MDD_search}).
The input is an instance of \textsc{Minimum Degree Deletion} $(G=(V,E),w_c,k)$ and a feedback edge set $E_f$ with $|E_f| = s_e$.
\begin{figure}
\begin{algorithmic}[1]
\Procedure{\texttt{MDD-search}}{$G=(V,E)$,$w_c$,$k$,$E_f$}
\State $(G',k') := \RemoveLowDegree(G)$
 \For{each $N' \subseteq N_E$ with $N_E := \{x \mid \{x,w_c\} \in E_f\}$}
  \State Remove $N'$ from $G'$.
  \State $k':= k - |N'|$
  \State $(G',k') := \RemoveLowDegree(G')$
  \For{each $N'' \subseteq (N(w_c) \setminus N_E)$ in $G'$}
   \State Remove $N''$ from $G'$.
   \State $k':= k' - |N''|$
   \While{there is a vertex $v$ with $\deg(v) \leq \deg(w_c)$ in $G'$}
    \State Remove $v$ from $G'$.
    \State $k' := k' - 1$
   \EndWhile
   \If{$k' \geq 0$}
    \State \Return ``yes''
   \EndIf
  \EndFor
 \EndFor
\State \Return ``no''
\EndProcedure
\end{algorithmic}
\caption{Fixed-parameter algorithm that solves \textsc{Minimum Degree Deletion} with respect to the parameter in $O(2^{s_e} \cdot n^2)$ time.}
\label{MDD_search} 
\end{figure}
We start with showing the correctness of \texttt{MDD-search}.
After applying the data reduction rule (see line 2), \texttt{MDD-search} branches over each $N' \subseteq N_E$ with $N_E$ being the set of $w_c$-neighbors that are connected to $w_c$ by an edge of $E_d$ (see line 3).
Here, the set $N'$ expresses the ``vertices from $N_E$ that are part of the solution set''.
\texttt{MDD-search} removes $N'$ from the graph, decreases the parameter value, and applies the data reduction rule again (see line 4-6).
Then, \texttt{MDD-search} branches over $N'' \subseteq N(w_c) \setminus N_E$ which are the ``remaining neighbors of $w_c$ that are part of the solution set'' (see line 7).
The algorithm removes $N''$ from the graph and decreases the parameter value (see lines 8-9).
Finally, \texttt{MDD-search} determines the remaining part of the solution set by iteratively removing all vertices with degree $\leq \deg(w_c)$ (see lines 10-13).
Clearly, the algorithm is correct, because it finally branches over all neighbors of $w_c$ that can be part of a solution set (see lines 2 and 9) and detects the remaining solution set vertices which are uniquely determined in each branching (see lines 10-13).

It remains to analyze the running time.
The first loop (line 3) iterates $O(2^{|N_E|})$ times.
By definition, $|N_E| \leq s_e$.
Now consider the graph $G' - N_E$ in line 7.
Clearly, $G'- N_E$ has a feedback edge set $E'_f$ with $|E'_f| = s_e - |N_E|$.
Due to the proof of Theorem~\ref{MDD_linearkernel} we already know that $(G' - N_E) - E'_f$ has at most $|E'_f|$ leaves.
Since the ``remaining neighborhood of $w_c$ in $G'$'', namely $N(w_c) \setminus N_E$ in $G'$, does not change after removing $N_E$ and $E'_f$ from $G'$, $|N(w_c) \setminus N_E|$ is also bounded by $|E'_f|$.
Hence, the second loop iterates at most $O(2^{|E'_f|}) = O(2^{s_e - |N_E|})$ times.
The remaining operations can be done in $O(n^3)$ time.
Putting all together we arrive at the following:
\begin{theorem}
 \textsc{Minimum Degree Deletion} is solvable in $O(2^{s_e} \cdot n^3)$ time with $s_e$ being the size of a feedback edge set and $n$ being the number of vertices.
\end{theorem}
As already mentioned we also show fixed-parameter tractability for the parameter $s_v$.
Since $t_w:=$``treewidth of the input graph'' is a lower bound for $s_v$, fixed-parameter tractability for $t_w$ would trivially imply fixed-parameter tractability for $s_v$.
Hence, we start with the investigation of the parameterized complexity of \textsc{Minimum Degree Deletion} with respect to the parameter $t_w$ in the following section.

\section{Treewidth as parameter}
In the previous sections, we have presented explicit fixed-parameter algorithms to prove that the problems are tractable with respect to the corresponding parameters.
It is sometimes hard to find explicit algorithms that solve a parameterized problem. 
Fortunately, there are results that state that large classes of problems can be solved in linear time when a tree decomposition with constant treewidth is known (see~\cite{Cou90,Cou09}).
Note that the method stated in this section is of purely theoretical interest.
The corresponding running times have huge constant factors and combinatorial explosions with respect to the parameter treewidth.
Hence, for practical applications one should search for an effective problem-specific algorithm.

\subsection{Monadic second-order logic}
We use a tool called \textit{monadic second-order logic} (or short \textit{MSO}).
This is an extension to the well-known first-order logic by quantification over sets.
Courcelle's Theorem~\cite{Cou90} says that the verification of a graph property is fixed-parameter tractable with respect to the parameter treewidth if the property can be expressed with monadic second-order logic.
An extensive overview about the field of monadic second-order logic can be found in~\cite{Cou09}.
In the following, we describe the language and syntax of MSO-formulae.
An MSO-formula consists of:
\begin{itemize}
 \item an infinite supply of \textit{individual variables}, by convention denoted by small letters $x$, $y$, $z$, $\dots$,
 \item an infinite supply of \textit{set variables}, by convention denoted by capital letters $X$, $Y$, $Z$, $\dots$,
 \item to express graph properties two unary relations, by convention denoted as~$V$ and~$E$, and a binary relation, by convention denoted as $I$, where
 \begin{itemize}
  \item the relation $V$ can be interpreted as ``being a vertex'',
  \item the relation $E$ can be interpreted as ``being an edge'',
  \item and the relation $I$ can be interpreted as ``being incident'',
 \end{itemize}
 \item some logical operators ($\neg$, $\wedge$, $\vee$, $\rightarrow$, and $\leftrightarrow$) and the quantifiers $\exists$ and $\forall$.
\end{itemize}
The relations will be used in prefix notation.
For a graph $G=(V,E)$ let $U:=V \cup E$.
An \textit{assignment} $\alpha$ for an MSO-formula maps each individual variable to an element of~$U$ and each set variable to a subset of $U$.
One defines the concept of an assignment $\alpha$ \textit{satisfying} an MSO-formula $\phi$, written $(G,\alpha) \models \phi$ for a given graph $G$.
Now, we can define the \textit{atomic} MSO-formulae and their semantics.
Let $G=(V,E)$ be a graph, $x$ and $y$ being individual variables, and $X$ be a set variable.
We have the following atomic MSO-formulae:
\begin{center}
\begin{tabular}{l | l}
 atomic formula & semantics \\
 \hline
 $x = y$ & $(G,\alpha) \models x = y$ $\Leftrightarrow$ $\alpha(x) = \alpha(y)$ \\
 $Vx$ & $(G,\alpha) \models Vx$ $\Leftrightarrow$ $\alpha(x) \in V$ \\
 $Ex$ & $(G,\alpha) \models Ex$ $\Leftrightarrow$ $\alpha(x) \in E$ \\
 $Ixy$ & $(G,\alpha) \models Ixy$ $\Leftrightarrow$ $\alpha(x) \in V$ is incident to $\alpha(y) \in E$ \\
 $Xx$ & $(G,\alpha) \models Xx$ $\Leftrightarrow$ $\alpha(x) \in X$ \\
\end{tabular}
\end{center}
Moreover, all other (more complex) MSO-formulae can be inductively built as follows:
\begin{itemize}
 \item If $\phi$ is an MSO-formula, then $\neg \phi$ is an MSO-formula as well.
 \item If $\phi$ and $\psi$ are MSO-formulae, then $\phi \wedge \psi$, $\phi \vee \psi$, $\phi \rightarrow \psi$, and $\phi \leftrightarrow \psi$ are MSO-formulae as well.
 \item If $\phi$ is an MSO-formula, $x$ is an individual variable, and $X$ is a set variable, then $\exists x \phi$, $\forall x \phi$, $\exists X \phi$, and $\forall X \phi$ are MSO-formulae as well.
\end{itemize}
Although ``$\rightarrow$'' and ``$\leftrightarrow$'' are not explicitly necessary we list them for sake of completeness.
Their semantics is analogous to first-order logic by combining ``$\wedge$'', ``$\vee$'', and ``$\neg$''.
The constructions have the following semantics:
\begin{center}
\begin{tabular}{l | l}
 construct & semantics \\
 \hline
 $\neg \phi$ & $(G,\alpha) \models \neg \phi$ $\Leftrightarrow$ $(G,\alpha) \not\models \phi$ \\
 $\phi \wedge \psi$ & $(G,\alpha) \models \phi \wedge \psi$ $\Leftrightarrow$ $(G,\alpha) \models \phi$ and $(G,\alpha) \models \psi$ \\
 $\phi \vee \psi$ & $(G,\alpha) \models \phi \vee \psi$ $\Leftrightarrow$ $(G,\alpha) \models \phi$ or $(G,\alpha) \models \psi$ \\
 $\exists x$ & $(G,\alpha) \models \exists x \phi$ $\Leftrightarrow$ there exists an $a \in U$ such that $(G,\alpha \frac{a}{x}) \models \phi$ \\
 $\forall x$ & $(G,\alpha) \models \forall x \phi$ $\Leftrightarrow$ for all $a \in U$ it holds that $(G,\alpha \frac{a}{x}) \models \phi$ \\
 $\exists X$ & $(G,\alpha) \models \exists X \phi$ $\Leftrightarrow$ there exists an $A \subseteq U$ such that $(G,\alpha \frac{A}{X}) \models \phi$ \\
 $\forall X$ & $(G,\alpha) \models \forall X \phi$ $\Leftrightarrow$ for all $A \subseteq U$ it holds that $(G,\alpha \frac{A}{X}) \models \phi$ \\
\end{tabular}
Herein $\alpha \frac{\xi}{\delta}$ denotes an assignment with $\alpha \frac{\xi}{\delta}(\delta)=\xi$ and $\alpha \frac{\xi}{\delta}(\zeta)=\alpha(\zeta)$ for all $\zeta \neq \delta$.
\end{center}

Analogously to first order logic, an MSO-\textit{sentence} is an MSO-formula without free variables.
Now, we are ready to present the main result in this field according to fixed-parameter algorithms.
To this end, we define the following problem: 

\begin{verse}
 \textsc{MSO-Check}\\
 \textit{Given:} A graph $G$ and an MSO-sentence $\varphi$.\\
 \textit{Question:} Is there an assignment $\alpha$ such that $(G,\alpha) \models \varphi$?
\end{verse}
It is easy to see that one can reduce every graph problem that is based on a graph property, expressible by an MSO-sentence, to \textsc{MSO-Check}:
Compute an MSO-sentence that expresses the graph property and take the graph and the sentence as input for the \textsc{MSO-Check} algorithm.
Courcelle developed the following important theorem~\cite{Cou90}:
\begin{theorem}[Courcelle's Theorem]
\label{theorem_courcelle}
 \textsc{MSO-Check} is fixed-parameter-tractable with respect to the combined parameter $(\tw(G),|\varphi|)$.
 Moreover, there is a computable function~$f$ and an algorithm that solves \textsc{MSO-Check} in time $f(\tw(G),|\varphi|) \cdot |G| + O(|G|)$.
\end{theorem}

We end with some simple examples for the application of Theorem~\ref{theorem_courcelle}.
\paragraph{Examples and extension.}
We start with a simple and well-studied graph property: \textit{3-colorability}.
A (vertex-)\textit{coloring} of an undirected graph is a mapping from the set of vertices to a (finite) set of colors, such that no two adjacent vertices have the same color.
We say that the graph is 3-colorable if there is a coloring with three colors.
The graph property ``to be 3-colorable'' is expressible with an MSO-sentence:
\begin{gather*}
 \varphi = \exists C_1 \exists C_2 \exists C_3 \bigg( \forall x: V x \rightarrow ( C_1 x \vee C_2 x \vee C_3 x) \bigg) \wedge \\
 \bigg( \forall e, \forall a \neq b: ( E e \wedge I a e \wedge I b e) \rightarrow \neg( (C_1 a \wedge C_1 b) \vee (C_2 a \wedge C_2 b) \vee (C_3 a \wedge C_3 b)) \bigg).
\end{gather*}
Thus, due to Theorem~\ref{theorem_courcelle} to decide whether a graph is 3-colorable is fixed-parameter tractable with respect to the parameter treewidth.

Arnborg et al.~\cite{ALS91} generalized this theorem to the case of ``extended monadic second-order logic''.
Here, we can additionally make use of set cardinalities.
This also includes the optimization problems regarding the set sizes. 
In this work, we only need the operation ``$\min X: P(X)$'' that expresses that a specific predicate $P$ holds for a minimum-size vertex set $X$.
Independently, Borie et al.~\cite{BPT92} obtained similar results, but from a more algorithmic point of view.
In their Theorem~3.5, they explicitly say that if a ``minimum edge/vertex deletion problem`` \textsc{Prob} that is based on a graph property which is expressible in MSO-logic, then \textsc{Prob} is expressible in extended monadic second-order logic.
To become familiar with the minimization extension, we give an easy example here:
Let $G=(V,E)$ be an undirected graph.
A \textit{vertex cover} is a subset $C$ of vertices in $V$ such that every edge in $E$ is incident to at least one vertex in $C$.
The predicate $\VC(X):=$``to be a vertex cover'', with $X$ being a vertex set, is expressible in MSO-logic:
$$ \VC(C) = \forall e: Ee \rightarrow (\exists v: (Cv \wedge I v e)).$$
The optimization variant of vertex cover where one asks for a minimum-size vertex cover is expressible in extended MSO-logic:
$$ \min C: \VC(C).$$
We use extended MSO in the next section to investigate the parameterized complexity of \textsc{Minimum Degree Deletion} with respect to the parameter treewidth.

\subsection{MSO expression for Minimum Degree Deletion}
\label{MSO expression for Minimum Degree Deletion}
In the following, we give a monadic second order sentence to prove the fixed-parameter tractability for \textsc{Minimum Degree Deletion} with respect to the parameter $t_w:=$``treewidth of the input graph''.
Since $s_v \geq t_w$ it follows that \textsc{Minimum Degree Deletion} is also fixed-parameter tractable with respect to the parameter $s_v$.

We start with an observation that helps us gaining an alternative view on the problem.

\begin{observation}
\label{wc_maxdegree_tw}
 Let $G=(V,E)$ be an undirected graph with treewidth $t_w$.
 Let $M^*$ be any solution set.
 Let $C$ be a tree decomposition of $G - M^*$ with maximum bag-size $t_w + 1$ and $C$ is minimal, that is, it is not possible is obtain another tree decomposition by removing vertices from the bags.
 It holds that $w_c$ has degree of at most $t_w - 1$ in $G - M^*$.
\end{observation}

\begin{proof}
 Assume towards a contradiction that $w_c$ has degree at least $t_w$ in $G - M^*$.
 Let $L$ be a leaf bag of $C$.
 There is a vertex $v_{l}$ that is only in bag $L$, because $C$ is minimal.
 Since $L$ has size $t_w+1$ by definition, $v_l$ can have at most $t_w$ neighbors; a contradiction to the fact that $w_c$ is the only vertex of minimum degree in $G - M^*$.
\end{proof}

Inspired by Observation~\ref{wc_maxdegree_tw} we can formalize \textsc{Minimum Degree Deletion} as follows:

\begin{verse}
 \textsc{Minimum Degree Deletion$^*$}\\
 \textit{Given:} An undirected graph $G=(V,E)$, a distinguished vertex $w_c \in V$, an integer $k \geq 1$, and an integer $i \leq t_w$.\\
 \textit{Question:} Is there a subset $K \subseteq V \setminus {w_c}$ of size $k$ such that $w_c$ has exactly $i$ neighbors not in $K$ and each other vertex is in $K$ or has more than $i$ neighbors not in $K$.
\end{verse}
It is easy to see that $(G,k)$ is a yes-instance of \textsc{Minimum Degree Deletion} if and only if $(G,k,i)$ is a yes-instance of \textsc{Minimum Degree Deletion$^*$} for some $i \leq t_w$.
Hence, we describe an MSO-sentence $\varphi$ expressing the graph property given by the question of \textsc{Minimum Degree Deletion$^*$}.
For the ease of presentation we first describe some of the parts of the sentence:
The formula part $\adj(x,y)$ is satisfiable if and only if $x$ and $y$ are adjacent:
\begin{gather*}
 \adj(x,y) = Vx \wedge Vy \wedge ( \exists e ( Ee \wedge Ixe \wedge Iye )).
\end{gather*}

The formula part $\iNotKNeighbors(x,K,i)$ is satisfiable if and only if there are at least $i$ neighbors of $x$ which are not contained in the vertex subset $K$:
\begin{gather*}
 \iNotKNeighbors(x,K,i) = \\ \exists n_1 \exists n_2 \dots \exists n_i \bigg( \bigg(\bigwedge_{1 \leq a \leq i} \adj(x,n_a) \wedge \neg K n_a \bigg) \wedge \bigg(\bigwedge_{1 \leq a < b \leq i} n_a \neq n_b \bigg) \bigg).
\end{gather*}

The formula part $\iNotKNeighborsEx(x,K,i)$ is satisfiable if and only if there are exactly $i$ neighbors of $x$ which are not contained in the vertex subset $K$:
\begin{gather*}
 \iNotKNeighborsEx(x,K,i) = \\
 \exists n_1 \exists n_2 \dots \exists n_i \bigg( \bigg( \bigwedge_{1 \leq a \leq i} \adj(x,n_a) \wedge \neg K n_a \bigg) \wedge \bigg( \bigwedge_{1 \leq a < b \leq i} n_a \neq n_b \bigg) \\
 \wedge \bigg( \forall n_z \bigg( \bigg( \adj(x,n_z) \wedge \bigg( \bigwedge_{1 \leq b \leq i} n_z \neq n_b \bigg) \bigg) \rightarrow Kn_z \bigg) \bigg) \bigg).
\end{gather*}

Putting all together we get:
\begin{gather*}
 \varphi_i = \min K \bigg( \iNotKNeighborsEx(w_c,K,i) \\
 \wedge \bigg( \forall x ( Vx \wedge x \neq w_c ) \rightarrow ( \iNotKNeighbors(x,K,i+1) \vee Kx ) \bigg) \bigg).
\end{gather*}

It is easy to see that $\varphi_i$ is an MSO-sentence.
By construction there is an assignment $\alpha$ such that $(G,\alpha) \models \varphi_i$ if and only if there is a subset $K \subseteq V$ such that $w_c$ is the only vertex with minimum degree and $\deg(w_c)=i$ in $G - K$.
Due to Theorem~\ref{theorem_courcelle} it follows that \textsc{Minimum Degree Deletion$^*$} is fixed-parameter tractable with respect to the combined parameter $(t_w,|\varphi_i|)$.
Due to Observation~\ref{wc_maxdegree_tw} an upper bound for the complexity of each $\varphi_i$ only depends linear on $t_w$ (which defines the maximum value for $i$).
Thus, we get the following theorem:
\begin{theorem}
\label{theorem_MDDtw}
  \textsc{Minimum Degree Deletion} is fixed-parameter tractable with respect to the parameter $t_w:=$``treewidth of the input graph''.
\end{theorem}

\section{Size of a feedback vertex set as parameter}
\label{Size of a Feedback Vertex Set as parameter for Minimum Degree Deletion}
In this section, we investigate the parameterized complexity of \textsc{Minimum Degree Deletion} with respect to the parameter $s_v$.
Since treewidth is a stronger parameter than $s_v$, the fixed-parameter tractability of \textsc{Minimum Degree Deletion} with respect to the parameter $s_v$ follows from Theorem~\ref{theorem_MDDtw}.
This result is mainly a classification.
A practicable algorithm is not given trough Courcelle's Theorem~\cite{Cou09}.
However, we investigate an interesting special case of the parameter ``size of a feedback vertex set''.
More  precisely, we show fixed-parameter tractability for \textsc{Minimum Degree Deletion} with respect to the parameter $s_v^*:=$``size of a feedback vertex set that does not contain $w_c$''.

Let $(G=(V,E),w_c,k)$ be the \textsc{Minimum Degree Deletion} instance and let $V_f:=\{v_{f_1},\dots,v_{f_{s_v^*}}\}$ be a concrete feedback vertex set that does not contain $w_c$.
The following observation bounds the final degree of $w_c$ in a solution graph from above by the parameter $s_v$:
\begin{observation}
\label{wc_maxdegree}
 Let $M^*$ be any solution set.
 It holds that $w_c$ has degree of at most $|V_f|$ in $G - M^*$.
\end{observation}
\begin{proof}
 Assume that $w_c$ has degree at least $|V_f| + 1$ in the solution graph $G^* := G - M^*$.
 Consider a minimal feedback vertex set $V_f^*$ of the solution graph $G^*$.
 Due to the minimality of $V_f^*$, $G^* - V_f^*$ is a forest with at least two vertices.
 Hence, $G^* - V_f^*$ contains at least two vertices with degree at most $1$.
 This means  $G^*$ contains at least two vertices with at most $|V_f| + 1$ neighbors, because each of them has at most $|V_f^*| \leq |V_f|$ neighbors in the feedback vertex set.
 Thus, $w_c$ is not the only with minimum degree; a contradiction.
\end{proof}
\begin{figure}
 \begin{algorithmic}[1]
  \Procedure{\texttt{MDD-solv}}{$G$, $w_c$, $k$, $V_f$}
   \For{each $V_f^* \subseteq V_f$ with $|V_f^*| \leq k$}
    \State $G' := G - V_f^*$
    \State $k' := k - |V_f^*|$
    \For{$i:=$0 to $|V_f|$}
     \While{there is a vertex $v \neq w_c$ with degree at most $i$}
      \State Remove $v$ from $G'$.
      \State $k':=k' - 1$
     \EndWhile
     \If{\texttt{AnnotatedMDD}($G'$, $w_c$, $k'$, $V'_f := V_f \setminus V_f^*$, $i$)}
      \State \Return 'yes'
     \EndIf
    \EndFor
   \EndFor
  \State \Return 'no'
  \EndProcedure
 \end{algorithmic}
 \caption{The algorithm \texttt{MDD-solv} solves \textsc{Minimum Degree Deletion}. The variable~$i$ represents the final degree of $w_c$ in the solution graph. Due to Observation~\ref{wc_maxdegree}, this final degree is bounded from above by the parameter.}
 \label{Fixed-parameter algorithm that solves that solves Minimum Degree Deletion with respect to the parameter s_v.} 
\end{figure}
We first introduce a template algorithm \texttt{MDD-solv} (see Figure~\ref{Fixed-parameter algorithm that solves that solves Minimum Degree Deletion with respect to the parameter s_v.}).
It solves \textsc{Minimum Degree Deletion} by calling a subroutine that solves a slightly modified version of \textsc{Minimum Degree Deletion} after doing some preprocessing and branching.
The preprocessing and branching part \texttt{MDD-solv} works as follows:
\begin{enumerate}
 \item Iterate over all subsets of the feedback vertex set $V_f$ (of size at most $k$) to fix which of the feedback vertex set vertices belongs to the solution set.
       Remove these vertices from $G$ and from $V_f$ and decrease the parameter accordingly. (lines 2-4)
 \item Iterate over all possible values $i$ for the final degree of $w_c$ in the solution graph. (line 5)
 \item Remove iteratively every vertex ($\neq w_c$) with degree at most $i$ and decrease the parameter accordingly. (lines 6-9)
\end{enumerate}
Now, we specify \texttt{AnnotatedMDD} which called as subroutine in the algorithm \texttt{MDD-solv}.
In line 10, \texttt{AnnotatedMDD} is used to solve the following problem:
\newpage
\begin{verse}
 \textsc{Annotated Minimum Degree Deletion}\\
 \textit{Given:} An undirected graph $G=(V,E)$, a feedback vertex set $V_f$, a distinguished vertex $w_c \in V$ with $w_c \notin V_f$, and two positive integer $k$ and $i$.
                 Every vertex in $G$ except $w_c$ has degree at least $i+1$. \\
 \textit{Question:} Is there a subset $M^* \subseteq V \setminus (V_f \cup \{w_c\})$ of size at most $k$ such that $w_c$ has degree $i$ in $G - M^*$ and every other vertex from $G - M^*$ has degree at least $i+1$?
\end{verse}

\begin{figure}
 \begin{algorithmic}[1]
  \Procedure{\texttt{Annotated-MDD-XP}}{$G$, $w_c$, $k$, $V_f$}
   \For{each $N_r \subseteq N(w_c)$ with $|N_r| = i$}
    \State $M^*:= V \setminus N_r$
    \State Remove each vertex in $M^*$ from $G$
    \While{there is a vertex $v \neq w_c$ with degree at most $i$}
     \State Remove $v$ from $G$.
     \State $M^*:= M^* \cup \{v\}$
    \EndWhile
    \If{$M^* \cap V_f = \emptyset$}
     \If{$|M^*| \leq k$}
      \State \Return 'yes'
     \EndIf
    \EndIf
   \EndFor
  \State \Return 'no'
  \EndProcedure
 \end{algorithmic}
 \caption{The $\XP$-algorithm \texttt{Annotated-MDD-XP} solves \textsc{Annotated Minimum Degree Deletion} in $O(|x|^{|V_f|})$ time with $|x|$ being the input size.}
 \label{The XP-algorithm.} 
\end{figure}
The preprocessing and branching of \texttt{MDD-solv} takes $O(f(s_v^*) \cdot \poly(|x|)$ time such that fixed-parameter tractability for \textsc{Annotated Minimum Degree Deletion} with respect to $s_v^*$ implies fixed-parameter tractability for \textsc{Minimum Degree Deletion} with respect to $s_v^*$.
Hence, our goal is to develop a fixed-parameter algorithm for \textsc{Annotated Minimum Degree Deletion}.
To become familiar with the problem, we start with an exponential-time algorithm:
The algorithm \texttt{Annotated-MDD-XP} in Figure~\ref{The XP-algorithm.} solves \textsc{Annotated Minimum Degree Deletion} by branching over all possible size-$i$ subsets of the neighborhood of $w_c$ (line 2).
The neighbors of $w_c$ that are not in this subset are removed (lines 3-4). 
This ensures that $w_c$ will have final degree $i$.
The algorithm removes iteratively each vertex with degree at most $i$ (lines 5-8).
This ensures that the remaining graph is a solution graph.
Clearly, \texttt{Annotated-MDD-XP} solves \textsc{Annotated Minimum Degree Deletion}, but branching over the neighborhood subsets takes $O(\binom{n}{i})$ time.
A fixed-parameter algorithm needs an improved approach.
The following two subsection present two different methods for showing fixed-parameter tractability for \textsc{Annotated Minimum Degree Deletion} with respect to the parameter $s_v^*$.

\subsection{Integer linear programming}
\label{Integer linear programming.}
In this section, we present a fixed-parameter algorithm that solves \textsc{Annotated Minimum Degree Deletion} with respect to the parameter $s_v^*:=$``size of a feedback vertex set that does not contain $w_c$'' by using the technique of integer linear programming together with a result from Lenstra~\cite{L83}.
Consider an instance $(G=(V,E), V_f, w_c, k, i)$ of \textsc{Annotated Minimum Degree Deletion} with $V_f = \{v_1, \dots, v_{s_v^*}\}$.
Without loss of generality we assume that $w_c$ has final degree at least $2$ in the solution graph.
Instances with smaller final degree for $w_c$ can be identified and solved in polynomial time.
Furthermore, we use the terms of solution set and nearly-solution set.
Note that the following definition of solution set is conform to its general definition in Section~\ref{BDD_intro}.
\begin{definition}
 Let $(G = (V,E), V_f, w_c, k, i)$ be an instance of \textsc{Annotated Minimum Degree Deletion}.
 We say $M^* \subseteq V$ is a \textbf{solution set} if:
 \begin{enumerate}
  \item $|M^*| \leq k$,
  \item $M^* \cap (\{w_c\} \cup V_f) = \emptyset$,
  \item $w_c$ has degree $i$ in $G - M^*$,
  \item every vertex $v \neq w_c$ has degree at least $i+1$ in $G - M^*$.
 \end{enumerate}
 We say $M' \subseteq V$ is a \textbf{nearly-solution set} if:
 \begin{enumerate}
  \item $|M'| \leq k$,
  \item $M' \cap (\{w_c\} \cup V_f) = \emptyset$,
  \item $w_c$ has degree $i-1$ in $G - M'$,
  \item every vertex $v \neq w_c$ has degree at least $i$ in $G - M'$.
 \end{enumerate}
\end{definition}
We start with an important observation concerning the existence of optimal solution sets.
More precisely, we show that if there is a an optimal solution set $M^*$, then $M^*$ is detectable in time $O(f(|V_f|) \cdot \poly(|x|))$ or there is a nearly-solution set that is detectable in time $O(f(|V_f|) \cdot \poly(|x|))$ for a computable function $f$ only depending on the parameter value $|V_f|$ and a polynomial only depending on the input size $|x|$.

\begin{proposition}
\label{Md_remaining}
 Let $M^*$ be an optimal solution set.
 It holds that:
 \begin{enumerate}
  \item $N(w_c) \setminus (M^* \cup V_f) = \emptyset$, or
  \item $M^* \subseteq N(w_c)$, or
  \item $\exists M' \subseteq N(w_c)$ such that $|M'| \leq |M^*|$ and $M'$ is a nearly-solution set.
 \end{enumerate}
\end{proposition}
\begin{proof}
 We show if $N(w_c) \setminus M^*$ contains a vertex that is not a feedback vertex set element and $M^*$ contains a vertex that is not a neighbor of $w_c$, then $\exists M' \subseteq N(w_c)$ such that $|M'| \leq |M^*|$ and  $M'$ is a nearly-solution set.
 We start with a simple partition of $M^*$ into $M_y$ and $M_x$. 
 Let $M_y$ be $M^* \cap N(w_c)$ or in other words the solution set vertices which are neighbors of $w_c$.
 Let $M_x$ be $M^* \setminus M_y$ or in other words the solution set vertices which are not neighbors of $w_c$.
 Clearly, it holds that $M_x \neq \emptyset$.
 We build an alternative solution set  $M' = M_y \cup {a}$ with $a$ being a neighbor of $w_c$ but not already in $M_y$ and $a \notin V_f$.
 The vertex $a$ exists, because $w_c$ has degree at least $2$ in $G - M^*$ and $N(w_c) \setminus (M^* \cup V_f) \neq \emptyset$.
 First, we show that $M'$ is at most as big as $M^*$.
 Since $M_x \neq \emptyset$, $|M_x| \geq 1$.
 Thus, $|M'| = |M^*| - |M_x| + 1 \leq |M^*|$.
 Second, we show that $M'$ is a nearly-solution set by showing each point of the definition:
 \begin{enumerate}
  \item Since $|M^*| \leq k$ and $|M'| \leq |M^*|$, it holds that $|M'| \leq k$.
  \item The set $M'$ contains only one vertex that is not also in $M^*$, namely $a$.
        Since $a \neq w_c$ and $a \notin V_f$, it hold that $M' \cap (\{w_c\} \cup V_f) = \emptyset$.
  \item By definition $w_c$ has degree $i$ in $G -M^*$.
        Since $M'$ contains only one neighbor of $w_c$ that is not also in $M^*$, it holds that $w_c$ has degree $i-1$ in $G - M'$.
  \item Each vertex from $V \setminus M^*$ except $w_c$ has degree at least $i+1$ in $G - M^*$.
        The additional vertex $a$ can decrease the degree of each vertex from $V \setminus M^*$ at most by one.
        Hence, each vertex from $V \setminus M^*$ except $w_c$ has degree at least $i$ in $G - M'$.
        It remains to show that the degree of each vertex from $M_x$ is at least $i$ in $G - M'$, too.
        Clearly, each vertex from $M_x$ has degree $i+1$ in the input graph $G$.
        Assume that $\exists x \in M_x$ with $\deg(x) < i$ in $G - M'$.
        Thus, $x$ must have two neighbors in $M'$.
        Furthermore, $x$ must be adjacent to two distinct vertices $b,c \in N(w_c)$, because $M' \subseteq N(w_c)$.
        Since $\{x,b,c,w_c\} \cap V_f = \emptyset$, there is a cycle $G[\{x,b,c,w_c\}]$ in the forest $G - V_f$; a contradiction.
 \end{enumerate}
\end{proof}

\begin{figure}
\textbf{Coefficients}:\\
\begin{itemize}
 \item $\tilde{v}_j$ denotes the number of neighbors of $v_j \in V_f$ that are not in $N(w_c)$
 \item $\tilde{g}_j$ denotes the number of group-$g_j$ vertices in $N(w_c)$.
 \item $\hat{g}_j$ denotes the number of group-$g_j$ vertices in $N(w_c)$ that have a neighbor $x \neq w_c$ with $\deg(x)=i+1$
\end{itemize}
The variable $x_j$ denotes the number of group-$g_j$ vertices in $N(w_c)$ that are not part of the solution set.\\
~\\
\textbf{ILP-1}\\
\textbf{minimize:}
\begin{gather*}
 \sum_{1 \leq j \leq 2^{s_v^*}} x_j
\end{gather*}
\textbf{subject to}
\begin{align*}
 \text{for all } 1 \leq j \leq 2^{s_v^*}:& \text{ } \hat{g}_j \leq x_j \leq \tilde{g}_j \\
 \text{for all } 1 \leq q \leq {s_v^*}:& \text{ } \tilde{v}_q + \sum_{j \text{ with } v_q \in G_j} x_j \geq i+1\\
\end{align*}
~\\
\textbf{ILP-2}\\
\textbf{minimize:}
\begin{gather*}
 \sum_{1 \leq j \leq 2^{s_v^*}} x_j
\end{gather*}
\textbf{subject to}
\begin{align*}
 \text{for all } 1 \leq j \leq 2^{s_v^*}:& \text{ } x_j \leq \tilde{g}_j \\
 \text{for all } 1 \leq q \leq {s_v^*}:& \text{ } \tilde{v}_q + \sum_{j \text{ with } v_q \in G_j} x_j \geq i\\
\end{align*}
 \caption{Integer linear programs ``ILP-1'' and ``ILP-2''. There is an optimal solution set $M^1 \subseteq N(w_c)$ if and only if the objective of ILP-1 is at most $i$ and a nearly-solution set $M^2 \subseteq N(w_c)$ if and only if the objective of ILP-2 is at most $i-1$.}
 \label{ILPs.}
\end{figure}
Due to Proposition~\ref{Md_remaining}, one can determine whether $(G, V_f, w_c, k, i)$ is a yes-instance of \textsc{Annotated Minimum Degree Deletion} by checking each of the three cases.
The first case is quite simple.
Since all neighbors of $w_c$ which are not in the optimal solution set are feedback vertex set elements and the solution set does not contain any feedback vertex set element, the optimal solution can be found as follows:
Firstly, remove $N(w_c) \setminus V_f$ from $G$.
The remaining solution set vertices are determined by iteratively removing every vertex with degree at most $i$.
Finally, if one did not remove any vertex from $V_f$ and the total number of removed vertices is at most $k$, then $(G, V_f, w_c, k, i)$ is a yes-instance.

The remaining two cases can be handles by two quite similar integer linear programs.
We briefly discuss the two cases and develop the integer linear programs in Figure~\ref{ILPs.}.
For the second case it remains to search for a set $M^1 \subseteq N(w_c)$ of size at most $k$ such that $w_c$ is the only degree-$i$ vertex in $G - M^1$ and every other vertex has degree at least $i+1$.
For the third case it remains to search for a set $M^2 \subseteq N(w_c)$ of size at most $k$ such that $w_c$ is the only degree-$i-1$ vertex in $G - M^2$ and every other vertex has degree at least $i$.
If $M^1$ or $M^2$ exists, then $(G, V_f, w_c, k, i)$ is a yes-instance.
To determine whether $M^1$ or $M^2$ exists, we consider the input graph $G = (V,E)$.
Every vertex in $G$ (except $w_c$) has degree at least $i+1$.
Clearly, no vertex in $M^1$ can have a neighbor with degree $i+1$ in $G$ (except $w_c$).
In contrast, $M^2$ may contain also vertices with degree-$i+1$ neighbors.
The main question is, which neighbors of $w_c$ are in $M^1$ respectively in $M^2$.

To this end, we consider the dual question: Which neighbors of $w_c$ are not in $M^1$ respectively in $M^2$.
Formally, we denote $N^1 := N(w_c) \setminus M^1$ and $N^2 := N(w_c) \setminus M^2$ as the remaining neighbors of $w_c$.
We have to ensure that each vertex from $V_f$ has final degree at least $i+1$ in $G - M^1$ respectively final degree at least $i$ in $G - M^2$.
Hence, the $V_f$-neighborhood of a vertex in $N(w_c)$ is important to decide whether the vertex is in $N^1$ respectively $N^2$.
Of course, each of the vertices in $V_f$ can have neighbors that are not in $N(w_c)$.
To express this, we need the following definition:
\begin{definition}
 We denote the number of neighbors of $v_j \in V_f$ that are not in $N(w_c)$ as $\tilde{v}_j$.
\end{definition}
The remaining neighbors of each feedback vertex set element must be in $N^1$ respectively $N^2$.
Moreover, we can group the vertices in $N(w_c)$ by the subsets of vertices in $V_f=\{v_1,\dots,v_{s_v^*}\}$ that are connected to them.
Note that in the case of searching for $N^1$, we already know that each vertex that has a neighbor $x \neq w_c$ with $\deg(x) = i+1$ must be in $N^1$.
\begin{definition}
 Let $\{G_1,\dots,G_{2^{s_v^*}}\}$ be the subsets of $V_f$.
 We say a vertex $v$ is a group-$g_j$ vertex if $N(v) \cap V_f = G_j$.
 Furthermore, $\tilde{g}_j$ denotes the number of group-$g_j$ vertices in $N(w_c)$ and $\hat{g}_j$ denotes the number of group-$g_j$ vertices in $N(w_c)$ that have a neighbor $x \neq w_c$ with $\deg(x)=i+1$.
\end{definition}
Besides the existence of a degree-$i+1$ neighbor, it is not important which vertex of each group is contained in $N^1$ respectively $N^2$.
Only the total number of vertices in $N^1$ respectively $N^2$ must be as small as possible under the condition that the final degree of each feedback vertex set element is at least $i+1$ respectively at least $i$.

Summarizing, in the second case, we have to find a minimum-size set $N^1 \subseteq N(w_c)$ by determining $x_j$, that is, the number of vertices of each group $g_j$ that are in $N^1$ for $j \in \{1,\dots,2^{s_v^*}\}$.
The first condition is that from each group $g_j$ at least $\hat{g}_j$ and at most $\tilde{g}_j$ vertices must be in $N^1$.
The second condition is that for each feedback vertex set element $v_l$ the sum of $x_j$ with $v_l \in G_j$ plus $\tilde{v}_j$ exceeds $i+1$.
Clearly, if the minimal sum of $x_j$ with $j \in \{1,\dots,2^{s_v^*}\}$ is at most $i$, the corresponding sets $N^1$ and $M^1$ exists as well.

In the third case, we have to find a minimum-size set $N^2 \subseteq N(w_c)$ by determining $x_j$, that is, the number of vertices of each group $g_j$ that are in $N^1$ for $j \in \{1,\dots,2^{s_v^*}\}$.
The first condition is that from each group $g_j$ at most $\tilde{g}_j$ vertices must be in $N^1$.
The second condition is that for each feedback vertex set element $v_l$ the sum of $x_j$ with $v_l \in G_j$ plus $\tilde{v}_j$ exceeds $i$.
Clearly, if the minimal sum of $x_j$ with $j \in \{1,\dots,2^{s_v^*}\}$ is at most $i-1$, the corresponding sets $N^2$ and $M^2$ exists as well.
Both tasks are formulated as integer linear programs (see Figure~\ref{ILPs.}).

Due to results of Lenstra~\cite{L83} and Kannan~\cite{K87} we use the following theorem to show fixed-parameter tractability for the computation of the integer linear program results:
\begin{theorem}[Lenstra's theorem]
\label{theorem_lenstra}
 \textsc{Integer Linear Programming Feasibility} can be solved with $O(p^{9p/2} \cdot L)$ arithmetic operations in integers of $O(p^{2p} \cdot L)$ bits in size, where $p$ is the number of ILP variables and $L$ is the number of bits in the input.
\end{theorem}
The running time of the integer linear program is in $O({|V_f|}^{\frac{9}{2} \cdot |V_f|} \cdot \poly(|x|))$ time.
Putting all together, we arrive at the following:
\begin{theorem}
 The problem \textsc{Annotated Minimum Degree Deletion} can be solved in $O({s_v^*}^{\frac{9}{2} \cdot {s_v^*}} \cdot \poly(|x|))$ time with $s_v^*$ being the size of a feedback vertex set that does not contain $w_c$ and $|x|$ being the size of the input.
\end{theorem}

\subsection{Dynamic programming}
\label{Dynamic programming.}
In this section, we present a fixed-parameter algorithm that solves \textsc{Annotated Minimum Degree Deletion} with respect to the parameter $s_v^*:=$``size of a feedback vertex set that does not contain $w_c$'' by using the technique of dynamic programming.
Let $(G=(V,E), V_f, w_c, k, i)$ be a \textsc{Annotated Minimum Degree Deletion} instance with $V_f:=\{v_1,\dots,v_{s_v^*}\}$.
We start with some general observation and define basic concepts.
By definition, every solution set contains only vertices in $V_S := V \setminus (V_f \cup \{w_c\})$.
The following observation considers the neighborhood of the vertices in $V_S$.
\begin{definition}
 Let $G=(V,E)$ be an undirected graph and $x \in V_S$ a vertex.
 We denote $\Component_{G[V_S]}(x)$ as the connected component from $G[V_S]$ including the vertex $x$.
\end{definition}
\begin{observation}
\label{obs_compo}
 Let $n_a,n_b$ be two distinct neighbors of $w_c$.
 It holds that:
 $$\Component_{G[V_S]}(n_a) \cap \Component_{G[V_S]}(n_b) = \emptyset$$.
\end{observation}
\begin{proof}
 Assume that there is an edge between one vertex $v_a$ of $\Component_{G[V_S]}(n_a)$ and another vertex $v_b$ of $\Component_{G[V_S]}(n_b)$ in $G[V_S]$.
 In $G - V_f$ both components are connected trough $w_c$.
 Thus, $G[\{v_a,w_c,v_b\}]$ is a $C_3$ and $G - V_f$ is cyclic; a contradiction.
\end{proof}
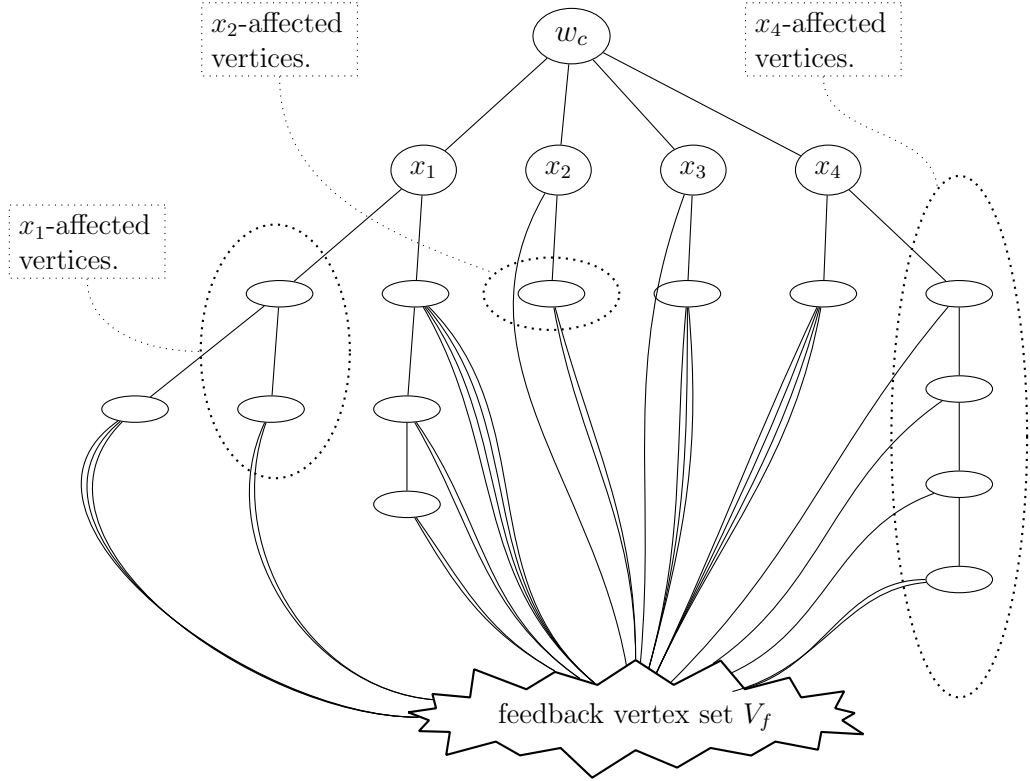
\begin{figure*}
\begin{tikzpicture}[>=latex,join=bevel,scale=0.9]
  \node (w_c) [draw,ellipse] {$w_c$};

  \node (n_1) [below left=2 of w_c] [draw,transform shape,ellipse] {$x_1$};
  \node (n_2) [right=of n_1] [draw,transform shape,ellipse] {$x_2$};
  \node (n_3) [right=of n_2] [draw,transform shape,ellipse] {$x_3$};
  \node (n_4) [right=of n_3] [draw,transform shape,ellipse] {$x_4$};

  \draw [-] (n_1) to (w_c);
  \draw [-] (n_2) to (w_c);
  \draw [-] (n_3) to (w_c);
  \draw [-] (n_4) to (w_c);

  \node (v_11) [below left=2 of n_1] [draw,transform shape,ellipse] {~~~~};
  \node (v_12) [right=of v_11] [draw,transform shape,ellipse] {~~~~};
  \node (v_21) [right=of v_12] [draw,transform shape,ellipse] {~~~~};
  \node (v_31) [right=of v_21] [draw,transform shape,ellipse] {~~~~};
  \node (v_41) [right=of v_31] [draw,transform shape,ellipse] {~~~~};
  \node (v_42) [right=of v_41] [draw,transform shape,ellipse] {~~~~};

  \draw [-] (n_1) to (v_11);
  \draw [-] (n_1) to (v_12);
  \draw [-] (n_2) to (v_21);
  \draw [-] (n_3) to (v_31);
  \draw [-] (n_4) to (v_41);
  \draw [-] (n_4) to (v_42);

  \node (v_13) [below left=2 of v_11] [draw,transform shape,ellipse] {~~~~};
  \node (v_14) [right=of v_13] [draw,transform shape,ellipse] {~~~~};
  \node (v_15) [right=of v_14] [draw,transform shape,ellipse] {~~~~};
  \node (v_16) [below=of v_15] [draw,transform shape,ellipse] {~~~~};
  \node (v_43) [below=of v_42] [draw,transform shape,ellipse] {~~~~};
  \node (v_44) [below=of v_43] [draw,transform shape,ellipse] {~~~~};
  \node (v_45) [below=of v_44] [draw,transform shape,ellipse] {~~~~};

  \draw [-] (v_11) to (v_13);
  \draw [-] (v_11) to (v_14);
  \draw [-] (v_12) to (v_15);
  \draw [-] (v_15) to (v_16);
  \draw [-] (v_42) to (v_43);
  \draw [-] (v_43) to (v_44);
  \draw [-] (v_44) to (v_45);

  \node (C) [below right=3.5 of v_16] [draw,transform shape,starburst,thick] {feedback vertex set $V_f$};

  \draw [-] (C) to[out=180,in=215] (v_13);
  \draw [-] (C) to[out=180,in=220] (v_13);
  \draw [-] (C) to[out=180,in=225] (v_13);

  \draw [-] (C) to[out=175,in=245] (v_14);
  \draw [-] (C) to[out=175,in=250] (v_14);

  \draw [-] (C) to[out=140,in=300] (v_12);
  \draw [-] (C) to[out=140,in=305] (v_12);
  \draw [-] (C) to[out=140,in=310] (v_12);
  \draw [-] (C) to[out=140,in=315] (v_12);

  \draw [-] (C) to[out=145,in=300] (v_15);
  \draw [-] (C) to[out=145,in=305] (v_15);

  \draw [-] (C) to[out=150,in=300] (v_16);
  \draw [-] (C) to[out=150,in=305] (v_16);

  \draw [-] (C) to[out=100,in=240] (n_2);

  \draw [-] (C) to[out=90,in=280] (v_21);
  \draw [-] (C) to[out=90,in=285] (v_21);

  \draw [-] (C) to[out=85,in=250] (n_3);

  \draw [-] (C) to[out=75,in=265] (v_31);
  \draw [-] (C) to[out=75,in=270] (v_31);
  \draw [-] (C) to[out=75,in=275] (v_31);

  \draw [-] (C) to[out=65,in=245] (v_41);
  \draw [-] (C) to[out=65,in=250] (v_41);
  \draw [-] (C) to[out=65,in=255] (v_41);
  \draw [-] (C) to[out=65,in=260] (v_41);

  \draw [-] (C) to[out=45,in=230] (v_42);

  \draw [-] (C) to[out=35,in=215] (v_43);

  \draw [-] (C) to[out=25,in=200] (v_44);

  \draw [-] (C) to[out=15,in=185] (v_45);
  \draw [-] (C) to[out=15,in=180] (v_45);

  \node (n1_aff) [draw,ellipse,thick,dotted,fit=(v_11) (v_14)] {};
  \node (n1_expl) [above left=4.5 of v_16] [draw,transform shape,text width=2cm,rectangle,dotted,text justified]
                  {$x_1$-affected vertices.};
  \draw [-,dotted] (n1_expl) to[out=270,in=180] (n1_aff);

  \node (n2_aff) [draw,ellipse,thick,dotted,fit=(v_21)] {};
  \node (n2_expl) [above=3 of v_11] [draw,transform shape,text width=2cm,rectangle,dotted,text justified]
                  {$x_2$-affected vertices.};
  \draw [-,dotted] (n2_expl) to[out=270,in=160] (n2_aff);

  \node (n4_aff) [draw,ellipse,thick,dotted,fit=(v_42) (v_43) (v_44) (v_45)] {};
  \node (n4_expl) [above=3 of v_41] [draw,transform shape,text width=2cm,rectangle,dotted,text justified]
                  {$x_4$-affected vertices.};
  \draw [-,dotted] (n4_expl) to[out=270,in=95] (n4_aff);

 \end{tikzpicture}
 \caption{Affected and unaffected vertices. Let $i=2$ denote the final degree of $w_c$ given by the input.
          For example, if one deletes $x_1$, then one has also to delete the marked $x_1$-affected vertices.}
 \label{Affected and unaffected vertices.}
\end{figure*}
\begin{figure}
 \begin{algorithmic}[1]
  \Require ~\hspace{0.5cm}
  \begin{itemize}
   \item $G$ is an undirected graph such that each vertex (except $w_c$) must have degree at least $i+1$
   \item $V_f$ is a feedback vertex set of $G'$ that does not contain $w_c$
  \end{itemize}
  \Procedure{\texttt{AnnotatedMDD}}{$G$, $w_c$, $k$, $V_f$, $i$}
   \State Compute to-remain-tuple $S^0:=(s_1^0,\dots,s_{|V_f|}^0)$ \Comment{Initialization}
   \For{$z$ = $1$ to $\deg(w_c)$}
    \For{each $S' \in \mathcal{S}$}
     \State $T(0,z,D) = +\infty$
    \EndFor
   \EndFor
   \For{each $S'=(s'_1,s'_2,\dots,s'_{|V'_f|}) \in \mathcal{S}$}
    \If{$s'_j < s_j^0$ for any $j \in \{1,\dots,|V'_f|\}$}
     \State $T(0,0,S') = +\infty$
    \Else
     \State $T(0,0,S') = 0$
    \EndIf
   \EndFor
   \For{$x=1,\dots,|N(w_c)|$} \Comment{Table update}
    \For{$z=1,\dots,x$}
     \For{each $S' \in \mathcal{S}$}
      \State $\text{minCostRemove} := T(x-1,z-1,S') + \cost_i(n_{x})$
      \State $\text{minCostNotRemove} := T(x-1,z,\Remain(S',n_{x}))$
      \State $T(x,z,S') := \min(\text{minCostRemove}, \text{minCostNotRemove})$
     \EndFor
    \EndFor
   \EndFor
   \If{$T(\deg(w_c), \deg(w_c)-i, (0,\dots,0)) \leq k$} \Comment{Result}
    \State \Return 'yes'
   \Else
    \State \Return 'no'
   \EndIf
  \EndProcedure
  \Ensure Returns yes if and only if there is a solution set $M^*$ such that
   \begin{itemize}
    \item $|M^*| \leq k$
    \item $V_f \cap M^* = \emptyset$
    \item $w_c$ has degree exactly $i$ in $G - M^*$
   \end{itemize}
 \end{algorithmic}
 \caption{AnnotatedMDD.}
 \label{AnnotetedMDD.} 
\end{figure}
Now, we take a look at an optimal solution set $M^*$:
Clearly, $M^*$ must contain $\deg(w_c) - i$ neighbors of $w_c$.
In addition, putting a vertex $x \in N(w_c)$ into a solution set can decrease the degree of other vertices from $\Component_{G[V_S]}(x)$.
Hence, removing $x$ can enforce to remove further vertices from $\Component_{G[V_S]}(x)$ recursively.
To measure this, we need the following definition which is illustrated in Figure~\ref{Affected and unaffected vertices.}:
\begin{definition}
\label{def_affect}
 We denote $A[x]$ as the \textit{$x$-affected} vertices, that is, the set of vertices that have to be removed when removing $x$.
 Furthermore, we denote $\cost(x):=|A[x]|$ as the number of vertices that have to be removed when removing $x$.
\end{definition}
To determine a solution set, one is interested in a set consisting of $\deg(w_c) - i$ neighbors of $w_c$ such that the sum of the corresponding costs is minimal.
The critical point is that it is also necessary to ``measure'' the effect that putting a vertex $x$ into the solution set has to the vertices from $V_f$.
By definition, we can not remove any vertex from $V_f$.
Thus, we must ensure that the final degree of every vertex from $V_f$ is at least $i+1$.
In the following we consider a partition of $V_S$ into a set $V_u$ containing vertices that must not belong to the solution and $V_r$ containing vertices that might be part of the solution.
For each vertex $v \in V_f$ one can compute ``how many neighbors from $V_r$ are not allowed to be deleted''.
To this end, we define a tuple expressing the effect of $V_u$ to the hole feedback vertex set:
\begin{definition}
\label{to-remain-tuple}
 We define the \textbf{to-remain-tuple} for $V_r \subseteq V_S$ as $S=(s_1,\dots,s_{|V_f|})$ where $s_j$ denotes the minimum number of $V_r$-neighbors of $v_j$ that are not allowed to be deleted.
\end{definition}
Now, we describe the dynamic programming algorithm \texttt{AnnotatedMDD} (see Figure~\ref{AnnotetedMDD.}).
Basically the algorithm iterates over neighbor subsets of $N(w_c)$ and computes for every such subset the minimum cost under the condition that deleting a number of vertices from $N(w_c)$ results in a specific to-remain-tuple.
More specifically, the dynamic programming table is defined as $T(x,z,S')$ with:
\begin{itemize}
 \item $x \in \{1,\dots,|N(w_c)|\}$,
 \item $z \leq x$, and
 \item $S' \subseteq \mathcal{S}:=\{(s'_1,\dots,s'_{|V_f|}) \mid 0 \leq s'_i \leq i+1\}$
\end{itemize}
The entry $T(x,z,S')$ then contains the minimum cost of deleting a size-$z$ subset $N'$ with $N' \subseteq \{n_i \in N(w_c) \mid i \leq x \}$ that ``realizes'' the to-remain-tuple $S'$ for $N'_r := N(w_c) \setminus N'$.
By definition of the table, $T(\deg(w_c),\deg(w_c)-i,(0,\dots,0)) \leq k$ if and only if $(G, V_f, w_c, k, i)$ is a yes-instance of \textsc{Annotated Minimum Degree Deletion}.

In the following, we describe the details of the algorithm \texttt{AnnotatedMDD} and show its correctness.
To state the $\Remain()$ function in line 19 (see Figure~\ref{AnnotetedMDD.}), we need a function that computes the modification of a to-remain-tuple when the algorithm fixes a specific neighbor of $w_c$ to be not contained in the solution set.
\begin{definition}
 Let $S' := \{s'_1, \dots, s'_{|V_f|}\}$ be a to-remain-tuple for a subset $V'_r \subseteq V_S$ and $x \subseteq (N(w_c) \cap V'_r)$.
 We define $\Remain(S',x) := (r_1,\dots,r_{|V_f|})$ where $r_j$ denotes the minimum number of $V''_r$-neighbors of $v_j$ that are not allowed to be deleted with $V''_r := V'_r \setminus A[x]$. More specifically:
 $$r_j := \min \{ \max \{0, s'_j + e(v_j,\{x\})\} , i+1 \}$$
 with $e(v_j,x)$ denoting the number of neighbors of a feedback vertex set element $v_j$ in $A[x]$.
\end{definition}
Furthermore, a kind of ``upper bound'' for the to-remain-tuple is given:
Let $U$ denote the \textit{unaffected vertices}, that is, those vertices that are not affected by removing any neighbor of $w_c$.
The to-remain-tuple $S^0 := (s_1^0,\dots,s_{|V_f|}^0)$ for $V_S \setminus U$ is given by:
$$s_j^0 := \max \{0, i+1 - |N(v_j) \cap U|\}$$
\begin{lemma}
 \texttt{AnnotatedMDD} is correct.
\end{lemma}
\begin{proof}
 Consider the algorithm in Figure~\ref{AnnotetedMDD.}.
 It starts with the computation of $S^0$ (line 2).
 In the initialization, the cost of ``removing at least one vertex from a set of zero vertices'' must be set to infinity (lines 3-7).
 Furthermore, without fixing any neighbor to be ``not contained in the solution set'' it is not possible to realize a to-remain-tuple that is better than  $S^0$ which means that there is an entry with a smaller value.
 Clearly, initializing these table values with infinity is correct (lines 9-10).
 By definition, it is easy to realize a to-remain-tuple $S'$ which is worse than $S^0$, that is, every entry $S'$ is at least as big as in $S^0$.
 Therefore on has no costs (lines 11-12).
 Now, consider the update step.
 It is easy to verify that the three loops (lines 15-17) ensure that every previous value that is used for computation (lines 18-19) has been computed before.
 For showing correctness it remains to prove the correctness of this computation.
 Consider a table entry $T(x,z,S')$.
 It contains the minimum costs for removing $z$ vertices from $N' \subseteq \{n_i \in N(w_c) \mid i \leq x \}$.
 Regarding the neighbor $n_x$, either it is part of the solution or not.
 If $n_x$ is removed, then the minimum costs are exactly the minimum costs for removing $z-1$ vertices from $N' \subseteq \{n_i \in N(w_c) \mid i \leq (x-1) \}$ plus $\cost(n_x)$ (line 18).
 Otherwise, if $n_x$ is not removed, then the costs are exactly the minimum costs for removing $z$ vertices from $N' \subseteq \{n_i \in N(w_c) \mid i \leq (x - 1) \}$, but realizing the to-remain-tuple under the condition that $n_x$ is fixed to be not part of the solution (line 19).
\end{proof}
Now, we analyze the running time and table size:
\begin{lemma}
 The table-size of $T$ is bounded by $O( (|V_f|+2)^{|V_f|} \cdot \deg(w_c)^2)$.
 The running time of \texttt{AnnotatedMDD} is bounded by a function in $O((|V_f|+2)^{|V_f|} \cdot \deg(w_c)^2 \cdot n^2)$.
\end{lemma}
\begin{proof}
 The first two dimensions are easily bounded by $\deg(w_c)$
 The to-remain-tuple is only defined such that each of the $|V_f|$ entries has an integer value between $0$ and $|V_f|+1$ (see Definition~\ref{to-remain-tuple}).
 Hence, $|\mathcal{S}| = {(|V_f|+2)}^{|V_f|}$.
 Clearly, the remaining steps can be accomplished in $O(n^2)$ time.
\end{proof}
Together with the preprocessing steps of the algorithm \texttt{MDD-solv} in Figure~\ref{Fixed-parameter algorithm that solves that solves Minimum Degree Deletion with respect to the parameter s_v.} we arrive at the following:
\begin{theorem}
 \textsc{Minimum Degree Deletion} can be solved in $O(({s_v^*}+2)^{s_v^*} \cdot 2^{s_v^*} \cdot n^4 \cdot \deg(w_c)^2)$ with $s_v^*$ being the size of a feedback vertex set that does not contain $w_c$.
\end{theorem}
The additional factor of $2^{|V_f|} \cdot n^2$ is due to the outer loops of \texttt{MDD-solv} (lines 2-9).

\section{No polynomial kernel with respect to $s_v$}
\label{No polynomial kernel with respect to s_v}

In the previous sections, we presented several fixed-parameter tractability results for \textsc{Minimum Degree Deletion} with respect to the parameters $t_w:=$``treewidth of the input graph'', $s_v:=$``size of a feedback vertex set'', and $s_v^*:=$``size of a feedback vertex set that does not contain $w_c$''.
These fixed-parameter tractability results imply problem kernels.
However, such kernels can have exponential (or even greater) sizes.
For theoretical and of course also for practical reasons one is interested in problem kernels of small sizes.
Although we showed a vertex-linear problem kernel with respect to the parameter $s_e:=$``size of a feedback edge set'', we did not find at least a polynomial kernel for $t_w$, $s_v$, and $s_v^*$.
In this section, we show that (unless the polynomial-time hierarchy~\cite{Stock76} collapses at third level) there is no polynomial kernel for \textsc{Minimum Degree Deletion} even with respect to the parameter $s_c^*:=$``size of a vertex cover that does not contain $w_c$''.
Since every vertex cover is also a feedback vertex set, $s_c^*$ is a weaker parameter than $s_v^*$ which is cleary a weaker parameter than $s_v$.
Hence, the non-existence of the polynomial kernel can be carried over to $t_w$, $s_v$ and $s_v^*$.

Bodlaender et al.~\cite{BDFH09} and Fortnow et al.~\cite{FS08} developed a framework for showing the existence or non-existence of a polynomial kernel.
We need the following definition:
\begin{definition} (\cite{BDFH09,FS08})
 A \textbf{composition algorithm} for a parameterized problem $L \subseteq \Sigma \times \mathbb{N}$ is an algorithm that receives as input a sequence $((x_1 , k)$, $\dots$, $(x_t , k))$, with $(x_i , k) \in \Sigma \times \mathbb{N}^+$  for each $1 \leq i \leq t$, uses time polynomial in $\sum_{i=1}^t |x_i| + k$, and outputs $(y, k') \in \Sigma \times \mathbb{N}^+$ with:
 \begin{enumerate}
  \item $(y, k' ) \in L \Leftrightarrow (x_i,k) \in L$ for some $1 \leq i \leq t$, and
  \item $k'$ is polynomial in $k$.
 \end{enumerate}
 A parameterized problem is called \textbf{compositional}, if there exists a composition algorithm.
\end{definition}

\begin{theorem} (\cite{BDFH09,FS08})
 \label{theorem_compositional}
 If any compositional problem whose unparameterized version is NP-complete has a polynomial kernel, then $\coNP \subseteq \NP / \poly$.
\end{theorem}

\begin{note}
 Theorem~\ref{theorem_compositional} implies $\PH = \Sigma^3_p$~\cite{Yap83}, but Cai et al.~\cite{CCHO02} improved this implication to the collapse $PH = S_2^{\NP}$, which is even a stronger result.
 However, the weaker collapse to $\Sigma^3_p$ may be more familiar.
\end{note}

Bodlaender et al.~\cite{BDFH09} introduced a refined concept of parameterized reduction that allows to transfer non-kernelizable results to new problems:
\begin{definition} (\cite{BTY08})
\label{def_ptp_transformation}
 Let $P$ and $Q$ be parameterized problems.
 We say that $P$ is \textbf{polynomial time and parameter reducible} to Q, written $P \leq_{Ptp} Q$, if there exists a polynomial time computable function $f : \Sigma \times \mathbb{N} \rightarrow \Sigma \times \mathbb{N}$ and a polynomial $p$, such that for all $(x, k) \in \Sigma \times \mathbb{N}$:
 \begin{enumerate}
  \item $(x, k) \in P \Leftrightarrow (x',k') = f (x,k) \in Q$, and
  \item $k' \leq p(k)$.
 \end{enumerate}
 The function $f$ is called \textbf{polynomial time and parameter transformation}. 
\end{definition}

\begin{proposition} (\cite{BTY08})
 \label{proposition_ptp}
 Let $P$ and $Q$ be parameterized problems, and suppose that $P^c$ and $Q^c$ are the derived classical problems.
 Suppose that $Q^c$ is NP-complete, and $P^c \in \NP$.
 Suppose that $f$ is a polynomial time and parameter transformation from $P$ to $Q$.
 Then, if $Q$ has a polynomial kernel, then $P$ has a polynomial kernel.
\end{proposition}

Now, we show a polynomial time and parameter transformation from \textsc{Small Universe Hitting Set} with respect to the combined parameter $d:=$``size of the universe'' and $k:=$``solution size'' to \textsc{Minimum Degree Deletion} with respect to the combined parameter $s_c^*:=$``size of a vertex cover that does not contain $w_c$'' and $k':=$``solution size''.

\begin{verse}
 \textsc{Small Universe Hitting Set}\\
 \textit{Given:} A family $S$ over a universe $U$ with $|U| \leq d$ and a positive integer $k$.\\
 \textit{Question:} Is there a subset $U' \subseteq U$ of size at most $k$ such that every set in $S$ has a non-empty intersection with $U'$?
\end{verse} 

Dom et al.~\cite{DLS09} showed that \textsc{Small Universe Hitting Set} with respect to the parameter $(d,k)$ does not have a polynomial kernel unless the polynomial-time hierarchy collapses.
In fact, they showed that a colored version of \textsc{Small Universe Hitting Set} is compositional and there is a polynomial time and parameter transformation from the colored to the uncolored version.

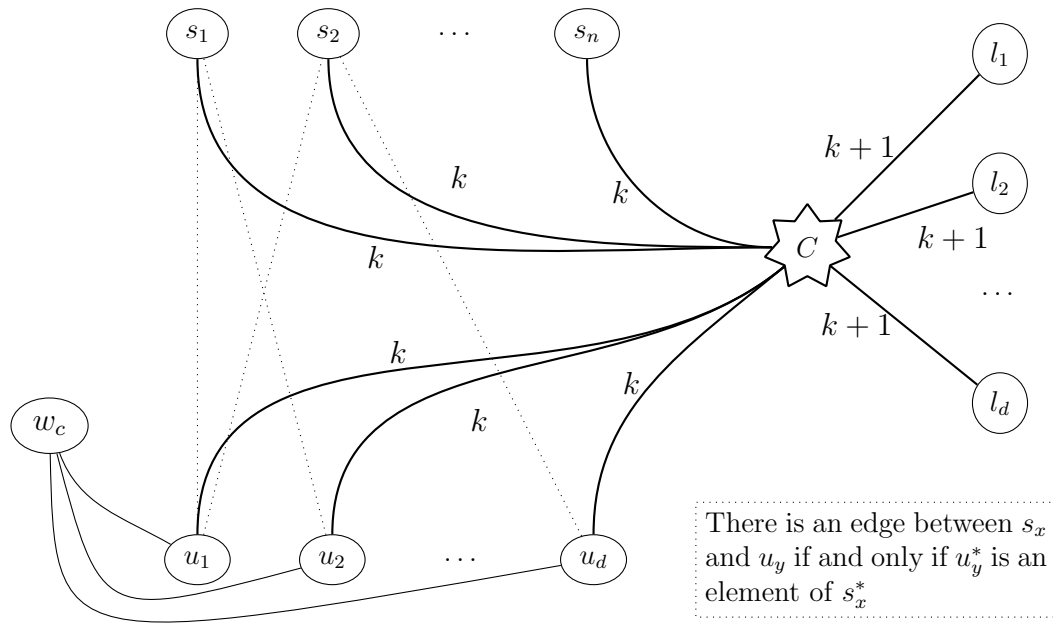
\begin{figure*}
\begin{tikzpicture}[>=latex,join=bevel,scale=0.9]
  \node (w_c) [draw,ellipse] {$w_c$};

  \node (u_1) [below right=2 of w_c] [draw,transform shape,ellipse] {$u_1$};
  \node (u_1_bl) [below left=0.2 of u_1] {};
  \node (u_2) [right=of u_1] [draw,transform shape,ellipse] {$u_2$};
  \node (u_dots) [right=of u_2] [transform shape] {$\dots$};
  \node (u_d) [right=of u_dots] [draw,transform shape,ellipse] {$u_d$};

  \node (s_1) [above=7 of u_1] [draw,transform shape,ellipse] {$s_1$};
  \node (s_1_bl) [below left=0.2 of s_1] {};
  \node (s_2) [right=of s_1] [draw,transform shape,ellipse] {$s_2$};
  \node (s_dots) [right=of s_2] [transform shape] {$\dots$};
  \node (s_n) [right=of s_dots] [draw,transform shape,ellipse] {$s_n$};

  \node (C) [below right=3.5 of s_n] [draw,transform shape,star,star points=7,star point ratio=0.8,thick] {$C$};
  \node (l_1) [above right=3 of C] [draw,transform shape,ellipse] {$l_1$};
  \node (l_2) [below=1 of l_1] [draw,transform shape,ellipse] {$l_2$};
  \node (l_dots) [below=1 of l_2] [transform shape] {$\dots$};
  \node (l_d) [below=1 of l_dots] [draw,transform shape,ellipse] {$l_d$};

  \node (C_Expl3) [above=of s_dots] [draw,transform shape,text width=15cm,rectangle,dotted,text justified]
                  {Each vertex $x \in \{u_1,\dots,u_d\} \cup \{s_1,\dots,s_n\}$ is connected to exactly $k$ vertices in $C$.};
  \node (C_Expl4) [right=of u_d] [draw,transform shape,text width=5cm,rectangle,dotted,text justified]
                  {There is an edge between $s_x$ and $u_y$ if and only if $u^*_y$ is an element of $s^*_x$};

  \draw [-] (u_1) to [out=150,in=290] (w_c);
  \draw [-] (u_2) .. controls +(195:3.5) .. (w_c);
  \draw [-] (u_d) .. controls +(190:8) .. (w_c);

  \draw [-,dotted] (u_1) to (s_1);
  \draw [-,dotted] (u_2) to (s_1);
  \draw [-,dotted] (u_1) to (s_2);
  \draw [-,dotted] (u_d) to (s_2);
  \draw [-,thick] (C) to[out=220,in=90] node[left,yshift=0.5em] {$k$} (u_1);
  \draw [-,thick] (C) to[out=220,in=90] node[left,yshift=-1.5em] {$k$} (u_2);
  \draw [-,thick] (C) to[out=220,in=90] node[left] {$k$} (u_d);

  \draw [-,thick] (C) to[out=180,in=270] node[left,yshift=-0.5em] {$k$} (s_1);
  \draw [-,thick] (C) to[out=180,in=270] node[left,yshift=1.5em] {$k$} (s_2);
  \draw [-,thick] (C) to[out=180,in=270] node[left] {$k$} (s_n);

  \draw [-,thick] (l_1) to node[left] {$k+1$} (C);
  \draw [-,thick] (l_2) to node[below right] {$k+1$} (C);
  \draw [-,thick] (l_d) to node[left] {$k+1$} (C);
 \end{tikzpicture}
 \caption{\textsc{Minimum Degree Deletion} instance $(G,w_c,k)$ obtained from a polynomial time and parameter transformation from  a \textsc{Small Universe Hitting Set} instance $(U^*,S^*,k)$ with universe-size $d$. The star $C$ represents a clique with $k+1$ vertices.
          A fat line that is labeled with a weight $j$ represents $j$ edges between a clique and a vertex.
          Dotted lines between two vertices $u_i$ and $s_j$ represent a single edge which states that $u^*_i$ is an element of $s^*_j$.}
 \label{Polynomial time and parameter transformation from Small Universe Hitting Set with respect to (d,k) to Minimum Degree Deletion with respect to (s_v^*,k).}
\end{figure*}

\paragraph{Polynomial time and parameter transformation.}
This transformation is illustrated in Figure~\ref{Polynomial time and parameter transformation from Small Universe Hitting Set with respect to (d,k) to Minimum Degree Deletion with respect to (s_v^*,k).}.
Let $(U^*,S^*,k)$ be a \textsc{Small Universe Hitting Set} instance with a universe $U^*=\{u^*_1,\dots,u^*_d\}$ , the subset family $S^*=\{s^*_1,\dots,s^*_n\}$, and the size of the solution set $k$.
We construct an undirected graph $G$ with a distinguished vertex $w_c$, such that $(G,w_c,k':=d-k)$ is a yes-instance of \textsc{Minimum Degree Deletion} if and only if $(U^*,S^*,d,k)$ is a yes-instance of \textsc{Small Universe Hitting Set}.
First, we create for each element of the universe $u^*_i$ one vertex $u_i$ and for each subset $s^*_j$ in $S^*$ one vertex $s_j$.
There is an edge between $u_i$ and $s_j$ if and only if $u^*_i$ is an element in $s^*_j$.
Additionally, we create the vertex $w_c$ which should become the vertex with minimum degree.
Our goal is to ensure that an optimal solution deletes $d-k$ vertices from $U:=\{u_1,\dots,u_d\}$ and no other vertex.
Therefore we set the degree of $w_c$ to $d$ and for each other vertex to at least $k$.
To this end, we create one clique $C$ of size $k+1$ and a set of $d$ vertices $L=\{l_1,\dots,l_d\}$, where each of them is connected to each vertex in $C$.
Each vertex that is in $U$ or in $S:=\{s_1,\dots,s_n\}$ is connected to exactly $k$ vertices in $C$.
Thus, $w_c$ has degree $d$ and each other vertex has degree at least $k+1$.

We are now ready to prove the correctness of the polynomial time and parameter transformation.
The first direction of the equivalence is quite simple.
\begin{lemma}
\label{lemma_SUHS2MDD}
 If $(U^*,S^*,k)$ is a yes-instance of \textsc{Small Universe Hitting Set} then $(G,w_c,k')$  is a yes-instance of \textsc{Minimum Degree Deletion}.
\end{lemma}

\begin{proof}
 Let $U' \subseteq U^*$ with $|U'| = k$ be a solution set of $(U^*,S^*,d,k)$, that is, each subset $\in S^*$ has a non-empty intersection with $U'$.
 The set $M:=\{u_j \mid u_j^* \in U^* \setminus U'\}$ is a solution set of $(G,w_c,d-k)$:
 The vertex $w_c$ has degree $k$ in $G - M$, since $U^*$ has size $d$ and $U'$ has size $k$.
 By the construction of $G$ each other vertex has degree at least $k+1$.
 Thus, $(G,w_c,d-k)$ is a yes-instance of \textsc{Minimum Degree Deletion}.
\end{proof}

The proof of the other direction of the equivalence needs some observations.
In the following, let $M_d$ be any solution set.

\begin{observation}
\label{obs_SUHS_exactly_k}
 The vertex $w_c$ has degree $k$ in $G - M_d$.
\end{observation}

\begin{proof}
 Assume that $w_c$ has degree more than $k$ in $G - M_d$.
 Since $w_c$ is the only vertex with minimum degree in $G - M_d$, $M_d$ must contain every vertex in $L$, because each vertex in $L$ has degree $k+1$.
 That means $M_d$ is of size $d > d-k$; a conflict.

 Assume that $w_c$ has degree less than $k$ in $G - M_d$.
 In this case, $M_d$ must contain at least $d-(k-1)$ neighbors of $w_c$; a conflict.
\end{proof}

\begin{observation}
\label{obs_SUHS_FVS_k}
 The graph $G$ has a vertex cover of size $k+1+d$ which does not contain $w_c$.
\end{observation}

\begin{proof}
 The set $C \cup U$ is a vertex cover:
 $G - (C \cup U)$ does not contain any edge; $C$ is of size $k+1$ and $U$ of size $d$. 
\end{proof}

\begin{lemma}
\label{lemma_MDD2SUHS}
 If $(G,w_c,k')$  is a yes-instance of \textsc{Minimum Degree Deletion} then $(U^*,S^*,k)$ is a yes-instance of \textsc{Small Universe Hitting Set}.
\end{lemma}

\begin{proof}
 It remains to show that there is a hitting set $U' \subseteq U^*$.
 One can build $U'$ as follows:
 For each vertex $u_i \in U$ which is not in $M_d$ add the element $u^*_i$ to $U'$.
 Due to Observation~\ref{obs_SUHS_exactly_k} the size of $U'$ is $k$.
 It remains show that $U'$ is a hitting set.
 Assume that there is a subset $s^*_j$ with $j \in \{1,\dots,n\}$ that has no intersection with any element in $U'$.
 Thus, for each element $u^*_i \in s^*_j$ the corresponding vertex $u_i$ is in $M_d$.
 Due to the construction of $G$ the vertex $s_j$ has degree $k$ in $G - M_d$.
 It follows that $M_d$ is not a solution set; a conflict.
\end{proof}

\begin{lemma}
\label{lemma_MDDptpSUHS}
 There is a polynomial time and parameter transformation from \textsc{Small Universe Hitting Set} with respect to the combined parameter $(d,k)$ to \textsc{Minimum Degree Deletion} with respect to the combined parameter $(s_c^*,k')$.
\end{lemma}

\begin{proof}
 Due to Lemma~\ref{lemma_SUHS2MDD} and~\ref{lemma_MDD2SUHS} the equivalence of both instances (point~1 of Definition~\ref{def_ptp_transformation}) is given.
 Due to Observation~\ref{obs_SUHS_FVS_k} and the construction of graph $G$, the new parameter $(s_c^*,k')$ is bounded by a polynomial only depending on the old parameter $(d,k)$ (point 2 of Definition~\ref{def_ptp_transformation}):
 \begin{gather*}
  (s_c^*,k') = (d+1+k,d-k).
 \end{gather*}
\end{proof}

Putting all together, we arrive at the following theorem:
\begin{theorem}
 \textsc{Minimum Degree Deletion} has no polynomial kernel with respect to the combined parameter $(s_c^*,k)$, with $s_c^*$ being the size of a vertex cover that does not contain the distinguished vertex and $k$ being the number of to deleted vertices, unless $\coNP \subseteq \NP / \poly$.
\end{theorem}
Of course, this implies that there is no hope for polynomial kernels for the parameters $t_w$, $s_v$, $s_v^*$ and $s_c:=$``size of a vertex cover'' as single parameters, or combined with $k$.

\chapter{Bounded Degree Deletion}
\label{Bounded Degree Deletion}
In this chapter, we analyze the parameterized complexity of \textsc{Bounded Degree Deletion}.
The problem is motivated as graph problem where one searches for a vertex subset of size at most $k$ whose removal from the graph is a graph in which each vertex has degree at most $d$.
The following section summarizes known results of the parameterized complexity of \textsc{Bounded Degree Deletion} which are obtained from~\cite{Moser10}.

\section{Known results}
\label{Known results for Bounded Degree Deletion}
\textsc{Bounded Degree Deletion} is fixed-parameter tractable with respect to the parameter $k$ for constant $d$, which can be seen by reduction to $(d+2)$\textsc{-Hitting Set}.
For $d=1$ there is a $15k$-vertex kernel and a $O(2^k \cdot k^2 \cdot k n)$ algorithm.
An $O(k^{1+\epsilon})$-vertex kernel was show for $d \leq 2$.
Furthermore, there is a fixed-parameter algorithm with running time $O((d+2)^k + n(k+d))$.
For unbounded $d$, \textsc{Bounded Degree Deletion} is W[2]-complete for the parameter $k$.
In the following we start with investigating the parameterized complexity with respect to the parameters that measure the ``degree of acyclicity'' beginning with the ``size of a feedback edge set''.
The parameterized complexity for ``treewidth of the input graph'' or ``size of a feedback vertex set'' remains open in this work.

\section{Size of a feedback edge set as parameter}
\label{Size of a Feedback Edge Set as parameter for Bounded Degree Deletion}
In this section, we show fixed-parameter tractability for \textsc{Bounded Degree Deletion} with respect to the parameter $s_e:=$``size of a feedback edge set''.
A first step will be to show that an annotated version is polynomial-time solvable on acyclic graphs.

\paragraph{An annotated version on acyclic graphs.}
In the following, we suppose that the input graph is acyclic.
We describe an algorithm that computes an \textit{optimal solution set}, that is, a vertex subset of minimum size whose removal from the graph is a graph in which each vertex has degree at most $d$.
To finally solve \textsc{Bounded Degree Deletion} on graphs with bounded feedback edge set size, we introduce a slightly modified version of \textsc{Bounded Degree Deletion} and show that it is solvable in polynomial time on acyclic graphs.
The modified problem is defined as follows:\\
\newpage
\begin{verse}
 \textsc{Annotated Bounded Degree Deletion}\\
 \textit{Given:} An undirected graph $G=(V,E)$, a vertex subset $U$, and integers $d \geq 0$ and $k \geq 0$.\\
 \textit{Question:} Does there exists a subset $V' \subseteq (V \setminus U)$ of size at most $k$ whose removal from $G$ yields a graph in which each vertex has degree at most $d$?
\end{verse}
The vertex subset $U$ is called the set of \textit{unremovable} vertices.
The algorithm uses a specialized bottom-up tree-traversal and handles each vertex exactly once:
Either the vertex will be marked to be ``not contained in the solution set'' or the vertex will be removed.
This process is called \textit{decision step} of the algorithm.
The order of processing the vertices corresponds to their depth\footnote{The depth of a vertex $x$ in a tree is the length of the path between $x$ and the root.} (from higher to lower).
Vertices that have the same depth are handled in order of their degree (from higher to lower).
This ensures three invariants in the decision step for vertex $x$:
\begin{enumerate}
 \item Each child of $x$ was either removed or was marked.
 \item Each child of every sibling of $x$ was either removed or was marked.
 \item There is no sibling of $x$ with higher degree which was not already removed or marked.
\end{enumerate}
The decision step of the algorithm is given in Figure~\ref{Decision step of the algorithm.}.
It is easy to see that after each decision step for a vertex $x$ either:
\begin{itemize}
 \item $x$ was removed, or
 \item $x$ was marked and has degree at most $d$, or
 \item the algorithm canceled due to the detection of a no-instance.
\end{itemize}

\begin{figure}
 \textbf{Decision step}\\
 \medskip
 \textbf{Input:} An undirected acyclic graph $G$ and a vertex $x$ from $G$. \\
 \textbf{Require:} Each child of $x$ was either removed or was marked.
                   Each child of every sibling of $x$ was either removed or was marked.
                   There is no sibling of $x$ with higher degree which was not already removed or marked.\\
 Let $p$ denote the parent vertex of $x$ and let $p_p$ denote the parent of $p$.
\begin{description}
 \item [Case A] $x$ is unremovable
 \begin{description}                
  \item [Case A.1] If $\deg(x) = d+1$ and $p$ is removable, then remove $p$ and mark $x$ to be ``not contained in the solution set''.
  \item [Case A.2] If $\deg(x) = d+1$ and $p$ is unremovable, then cancel and return ``no''.
  \item [Case A.3] If $\deg(x) > d+1$, then cancel and return ``no''.
  \item [Case A.4] If $\deg(x) < d+1$, then mark $x$ to be ``not contained in the solution set''.
 \end{description}
 \item [Case B] $x$ is removable
 \begin{description}
  \item[Case B.1] $p$ is removable
  \begin{description}
   \item[Case B.1.a] If $\deg(x) < d+1$, then mark $x$ to be ``not contained in the solution set''.
   \item[Case B.1.b] If $\deg(x) > d+1$, then remove $x$.
   \item[Case B.1.c] If $\deg(x) = d+1$, then remove the parent vertex of $x$. If there is no parent vertex ($x$ is the root), then just remove $x$.
  \end{description}
  \item[Case B.2] $p$ is unremovable
  \begin{description}
   \item[Case B.2.a] If $\deg(x) \geq d+1$, then remove $x$.
   \item[Case B.2.b] If $\deg(x) < d+1$ and $\deg(p) < d+1$, then mark $x$ to be ``not contained in the solution set''.
   \item[Case B.2.c] It holds that $\deg(x) < d+1$ and $\deg(p) \geq d+1$.
                     If $p_p$ exists and $p_p$ is removable, then remove $p_p$.
                     If $\deg(p) \geq d+1$ (possibly after removing $p_p$), then remove $x$.
  \end{description}
 \end{description}
\end{description}
 \textbf{Ensure:} Either $x$ is removed, or $x$ is marked and has degree at most $d$, or the algorithm cancels due to the detection of a no-instance. 
                  Furthermore, the decision is optimal with respect to the total number of removed vertices.
\caption{Decision step of the algorithm. The parent vertex of $x$ is denoted as $p$.} 
\label{Decision step of the algorithm.} 
\end{figure}
Let $M$ be the set of marked vertices and $S := V \setminus M$.
\begin{lemma}
 The algorithm computes an optimal solution set in $O(n^2)$ time.
\end{lemma}
\begin{proof}
 It remains to show that if the decision step was correct and optimal for each child of a vertex $x$, then the decision step is also correct and optimal for $x$.
 Due to the processing order, the correctness (and optimality) of the whole algorithm follows.
 This proof is more or less a complete induction:
 The base clause is ``the decision step is correct and optimal for each child of a leaf'' which is clearly given, because a leaf has no child.
 In the following we show correctness and optimality for each case:
 \vspace*{0.1cm}\\
 \begin{tabular}{l p{12cm}}
  Case A.1   & Since the decision was correct for each child, $p$ is the only neighbor of $x$ that can be removed.
               Removing $x$ is not allowed due to the problem definition.
               To decrease the degree of $x$ to $d$ one must remove $p$. \\
  Case A.2   & No neighbor of $x$ can be removed. 
               Hence, it is not possible to decrease the degree of $x$. \\
  Case A.3   & Analogously to case A.1, $p$ is the only neighbor of $x$ that can be removed.
               Removing $x$ is not allowed due to the problem definition.
               Thus, it is not possible to decrease the degree of $x$ by more than one anyway. \\
  Case A.4   & This case is correct by the problem definition. \\
  Case B.1.a & Removing $x$ does not lead to an optimal solution.
               The degree of $x$ is already small enough and $p$ is removable. 
               Hence, it is at least as good to remove $p$ instead of $x$. \\
  Case B.1.b & Clearly, one must remove either $x$ or at least $2$ neighbors of $x$.
               Since the decision was correct for each child, $p$ is the only neighbor of $x$ that can be removed. 
               Hence, one must remove $x$ anyway. \\
  Case B.1.c & Clearly, one must remove either $x$ or $p$.
               It is not necessary to remove $x$, because removing $p$ instead is always better.
               Note that the degree of each child of $x$ is already at most $d$. \\
  Case B.2.a & No neighbor of $x$ can be removed.
               One must decrease the degree of $x$ at least by one.
               Clearly, the only way to do this is removing $x$. \\
  Case B.2.b & Removing $x$ does not lead to an optimal solution.
               The degree of $x$ and the degree of the only neighbor $p$ is already small enough so that we do not need to remove $x$.
               Hence, not removing $x$ cannot be wrong. \\
  Case B.2.c & One must remove enough neighbors of $p$ such that it has final degree at most $d$.
               Due to the invariants, no remaining sibling has another neighbor than $p$ and no sibling hat degree more than $d$.
               More precisely, $x$ has degree at most $d$ and siblings with higher degree were processed before.
               It does not matter which child of $p$ will be removed, but removing the parent vertex can possibly decrease the degree of another vertex. \\
 \end{tabular}\medskip\\
 Hence, the vertex subset $S$ is an optimal solution set.
 It remains to prove the running time.
 Collecting the degree information for each vertex and ordering the vertices according to their depth and degree takes $O(n^2)$ time. 
 Each decision step is computable in $O(n)$ time:
 Marking a vertex or checking whether a vertex is marked takes constant time.
 Removing a vertex while updating the degree information and updating the ordering of the vertices takes $O(n)$ time.
 Since each vertex is only processed once, the algorithm needs $n$ decision steps.
 Thus, the overall running time is in $O(n^2)$.
\end{proof}

\paragraph{Extension to the case of bounded feedback edge set size.}
Let $(G,d,k)$ be an instance of \textsc{Bounded Degree Deletion}.
In the following, we assume that there is a feedback edge set $E_f$ of size $s_e$.
The first step is a search tree that transforms the instance into acyclic instances.
We branch on the feedback edge set elements into three cases.
Let $\{x,y\}$ be an edge in $E_f$:
\begin{description}
 \item[Branching case 1] Remove $x$.
 \item[Branching case 2] Remove $y$.
 \item[Branching case 3] Do not remove $x$ and $y$.
               Instead, remove $\{x,y\}$ and add one additional leaf $a_x$ to $x$ and one additional leaf $a_y$ to $y$.
               (This is necessary to preserve the degrees.)
               Mark $x$, $a_x$, $y$, and $a_y$ as ``unremovable''.
\end{description}
Let $G'$ be the resulting graph, $U$ be the set of vertices that are marked as ``unremovable'', and $r$ the number of removed vertices.
In each search tree leaf $G'$ is acyclic.
We have to guarantee to find an optimal solution set in at least one branch.
Thus, a second step is to solve the \textsc{Annotated Bounded Degree Deletion} instance $(G',U,d,k-r)$ in each search tree leaf.
If the \textsc{Annotated Bounded Degree Deletion} instance in any search tree leaf is a yes-instance, then return ``yes''.
Otherwise, return ``no''.
Putting all together, we arrive at the following:
\begin{theorem}
 \textsc{Bounded Degree Deletion} can be solved in $O(3^{s_e} \cdot n^2)$ time with $s_e$ being the size of a feedback edge set.
\end{theorem}
\begin{proof}
 The running time of the algorithm is clearly in $O(3^{s_e} \cdot n^2)$.
 Furthermore, it is easy to see that if there is a yes-instance of \textsc{Annotated Bounded Degree Deletion} in a search tree leaf, then the original \textsc{Bounded Degree Deletion} instance is a yes-instance, too.
 It remains to show that if the original instance is a yes-instance of \textsc{Bounded Degree Deletion}, then there is a yes-instance of \textsc{Annotated Bounded Degree Deletion} in at least one search tree leaf.
 Assume, that there is a solution set $M_d$ for the original instance $(G,d,k)$.
 Let $V_I$ denote the vertices that are incident to an edge from $E_f$.
 Let $M_I := V_I \cap M_d$ denote the solution set vertices that are incident to an edge from $E_f$.
 Let $M_I^2 := \{ x \mid x \in M_I \wedge \exists y: y \in M_I \wedge \{x,y\} \in E_f\}$ denote the solution set vertices that are connected to another solution set vertex by an edge from $E_f$.
 It is easy to verify that there is a search tree leaf $s$ with $U_s$ being the set of as ``unremovable'' marked vertices and $R_s$ is the set of removed vertices such that:
 \begin{itemize}
  \item For each vertex $x \in M_I^2$ either $x \in R_s$ and $y \notin U_s$ or $y \in R_s$ and $x \notin U_s$.
  \item For each vertex $x \in (M_I \setminus M_I^2)$ it holds that $x \in R_s$.
  \item Each vertex $u \in U_s$ is not part of the solution set $M_d$.
 \end{itemize}
 Clearly, the \textsc{Annotated Bounded Degree Deletion} instance in the search tree leaf $s$ is a yes-instance.
\end{proof}

\chapter{Conclusion and Outlook}
\label{Conclusion and outlook}
We investigated the parameterized complexity of three similar vertex deletion problems.
Since each problem is solvable in polynomial time when restricted to acyclic graphs, but $\NP$-hard in general, it was natural to study fixed-parameter tractability with respect to parameters that measure the ``degree of acyclicity''.
More precisely, we considered $s_e:=$``size of a feedback edge set'' respectively $s_a:=$``size of a feedback arc set'', $s_v:=$``size of a feedback vertex set'', and $t_w:=$``treewidth of the input graph''.

For \textsc{Minimum Indegree Deletion}, which is equivalent to constructive control by deleting candidates in Llull elections, we showed that it is W[2]-hard with respect to the parameters $s_a$ and $s_v$.
In  addition, it is at least in $\XP$ with respect to~$s_v$ (and~$s_a$).
Although it is W[2]-complete with respect to the parameter $k:=$``solution size'', we could show fixed-parameter tractability with respect the combined parameters~$(s_a,k)$ and~$(s_v,k)$.
The fixed-parameter algorithm in Section~\ref{Feedback vertex set size and solution size as combined parameter} solves \textsc{Minimum Indegree Deletion} in $O(s_v \cdot (k+1)^{s_v} \cdot n^2)$ time.
There should be room for improvements of this algorithm.

In contrast to \textsc{Minimum Indegree Deletion}, \textsc{Minimum Degree Deletion} is even fixed-parameter tractable with respect to each of the parameters that measure the ``degree of acyclicity''.
We showed that \textsc{Minimum Degree Deletion} has no polynomial kernel with respect to $s_v$ or $t_w$, unless the polynomial-time hierarchy collapses.
Moreover, there is also no polynomial kernel with respect to the combined parameter $(s_c^*,k)$ with $s_c^*:=$``size of a vertex cover that does not contain $w_c$''.
However, a vertex-linear kernel was found with respect to $s_e$.
In future research it would be desirable to develop a practicable algorithm with respect to $t_w$.
A first step towards this was done in Section~\ref{Size of a Feedback Vertex Set as parameter for Minimum Degree Deletion}, where a fixed-parameter algorithm with respect to $s_v^*:=$``size of a feedback vertex set that does not contain $w_c$'' using dynamic programming was developed.
Possibly, there are ways to adapt the ideas to the parameter $s_v$.

For \textsc{Bounded Degree Deletion} we showed a fixed-parameter algorithm using a depth-bounded search tree with running time $O(3^{s_e} \cdot n^2)$.
Improving this algorithm (for example by data reduction rules) or developing a fixed-parameter algorithm for the combined parameter $(s_v,k)$ are promising future tasks.
In addition, the parameterized complexity of \textsc{Bounded Degree Deletion} with respect to the parameters $s_v$ and $t_w$ still remains open.


\bibliography{literature}
\bibliographystyle{alpha}

\newpage
\chapter*{Selbstständigkeitserklärung}

Ich erkläre, dass ich die vorliegende Arbeit selbstständig und nur unter Verwendung der angegebenen Quellen und Hilfsmittel angefertigt habe.
\bigskip
\\
Jena, den 25. Mai 2010\\ \begin{flushright}Robert Bredereck\end{flushright}



\end{document}